\documentclass[11pt]{article}
\usepackage[margin=1in]{geometry}
\usepackage[T1]{fontenc}
\usepackage{lmodern,microtype,amsmath,amssymb,amsthm,booktabs,array,enumitem}
\usepackage[colorlinks=true,linkcolor=blue,citecolor=blue,urlcolor=blue]{hyperref}
\setlist{itemsep=3pt,topsep=4pt}
\newcommand{\F}{\mathbb F}
\newcommand{\Z}{\mathbb Z}
\newcommand{\Q}{\mathbb Q}
\newcommand{\R}{\mathbb R}
\newcommand{\wt}{\operatorname{wt}}
\newcommand{\rank}{\operatorname{rank}}
\newcommand{\dist}{\operatorname{dist}}
\newcommand{\Enc}{\operatorname{Enc}}
\newcommand{\ind}{\operatorname{ind}}
\newcommand{\Lift}{\operatorname{Lift}}
\newcommand{\poly}{\operatorname{poly}}
\newtheorem{theorem}{Theorem}[section]
\newtheorem{lemma}[theorem]{Lemma}
\newtheorem{proposition}[theorem]{Proposition}
\newtheorem{corollary}[theorem]{Corollary}
\theoremstyle{definition}

\newtheorem{example}[theorem]{Example}
\theoremstyle{remark}
\newtheorem{remark}[theorem]{Remark}

\title{Mathematics and Efficiency of Song's CVP Hardness Reduction}
\author{Shuhong Gao\thanks{AI Use Disclosure: The main idea and framework are due to the author, and GPT-6 Astra assisted in the research and the writing of the paper.}\\
\small School of Mathematical and Statistical Sciences\\
\small Clemson University, Clemson, SC 29634-0975\\
\small Email: \href{mailto:sgao@clemson.edu}{sgao@clemson.edu}}
\date{October 28, 2026}
\begin{document}
\maketitle
\begin{abstract}
Song proves that, for every fixed $0<\epsilon<1$, the closest
codeword problem (CCP) for binary linear codes under Hamming
distance is NP-hard to approximate within $N^{1-\epsilon}$,
where $N$ is the code length, under deterministic polynomial-time
many-one reductions~\cite{song}. The code-to-lattice construction
then gives hardness within $N^{1/2-\epsilon/2}$ for the Euclidean
closest vector problem (CVP),
with the same $N$ as lattice rank.
This reaches every fixed Euclidean hardness exponent below one half.
Aharonov--Regev~\cite{ar}, with the constant refined by
Gao~\cite{gao-ar}, place GapCVP at factor $cN^{1/2}$ in
$\mathrm{NP}\cap\mathrm{coNP}$ for every fixed
$c>2/(\pi\sqrt3)\approx0.3676$. Thus NP-hardness at that scale
would imply $\mathrm{NP}=\mathrm{coNP}$.

This paper has two goals: to explain Song's theorem and its
mathematics for graduate students, and to make the reduction
more efficient.
The reduction turns satisfiability into a system of quadratic
equations and then into CCP for binary linear codes, and recovers
assignments from recurrence relations of algebraic moments.
We adapt the method to linear codes over any fixed finite field
of characteristic $p$. A common recurrence for each block column
gives a moment window $L=c_p(\kappa-1)+1$, where
$c_p=\max\{2,p-1\}$ and $\kappa$ is the rank cutoff.
Together with ordinary digit products and flexible sampling,
this reduces interpolation degrees and output lengths.
For CCP gap $N^{1-\epsilon}$, our explicit schedule gives code
length and lattice rank at most $n^{4/\epsilon^2+O(1/\epsilon)}$;
Song's prescribed binary schedule gives the bound
$n^{8/\epsilon^2+O(1/\epsilon)}$. Here $n$ counts variables in
a common quadratic system, equal to the number of original variables
plus clauses for three-literal 3SAT. These are feasible dimension
bounds, not optimality claims or running-time exponents.
The limiting hardness exponents remain Song's.
\end{abstract}
\clearpage
\begingroup\small
\setcounter{tocdepth}{2}
\tableofcontents
\endgroup
\clearpage

\section{Introduction}\label{sec:introduction}

This exposition has two goals. The first is to explain Song's
theorem and its mathematics to graduate students who
know linear algebra and polynomials. The second is to make the
reduction more efficient: keep the same logical mechanism while
reducing the interpolation degree, sample count, and output length.
We derive the resulting gaps from the full output length, not just
from the number of variables in the starting formula.

\begin{theorem}[Song~\cite{song}]\label{thm:song}
For every fixed $0<\epsilon<1$, the binary closest codeword problem
(CCP) and binary
syndrome decoding are NP-hard to approximate within $N^{1-\epsilon}$,
where $N$ is binary length. Euclidean CVP is NP-hard to
approximate within $N^{1/2-\epsilon/2}$,
where $N$ is lattice rank. Both results hold under deterministic
polynomial-time many-one reductions from 3SAT.
\end{theorem}

Song's proof~\cite[Sections 3--8]{song} combines a linear-code
construction with a rank dichotomy and algebraic recovery of
satisfying assignments. We explain this mechanism and refine its
parameters and characteristic. The following theorem records our
fixed-field extension and explicit output bound; the hardness
exponents themselves are those of Theorem~\ref{thm:song}.

\begin{theorem}[Fixed-field extension and explicit output lengths]
\label{thm:refined-main}
Fix a finite field $F$ of characteristic $p$ and $0<\epsilon<1$,
independently of the input size. Set
\[
 r=\left\lceil\frac2\epsilon\right\rceil,\qquad
 \tau=\epsilon(r+1)-1,\qquad
 A=1+\left\lfloor\frac{(1-\epsilon)(r+3)+2}{\tau}\right\rfloor,
 \qquad H=A(r+1)+r+3.
\]
For a 3SAT formula with $s$ variables and $m_3$ three-literal
clauses, put $n=s+m_3$. For all sufficiently large $n$, there
are deterministic polynomial-time reductions to CCP over $F$
with gap $N^{1-\epsilon}$ and to Euclidean CVP with gap
$N^{(1-\epsilon)/2}$, where the code length and lattice rank are
\[
 N=\Theta_{F,\epsilon}\!\left(\frac{n^H}{(\log n)^{r+1}}\right),
 \qquad H=\frac4{\epsilon^2}+O(1/\epsilon).
\]
The construction over the prime subfield uses
$L=c_p(\kappa-1)+1$, $c_p=\max\{2,p-1\}$.
The lattice is formed from that prime-field instance; extension
of the code alphabet to $F$ preserves length and minimum
feasible Hamming weight. Bounded smaller instances have fixed
outputs with the same gap promises.
\end{theorem}

The proof is given by Theorem~\ref{thm:table-reduction}, the
lattice transfer in Section~\ref{sec:lattice}, and the schedules
and alphabet extension in Section~\ref{sec:gap}. Our changes
use a shared recurrence within each block column, ordinary digits
for the factors of the moment matrix, and sampling on any
sufficiently large subset of the coefficient field. The recovery
argument also extracts a consequence of Song's coefficient-matrix
structure: every active root occurs in the normalization sequence
$(0,0)$. This removes quadratic sparsity from soundness.

\paragraph{Scope of the comparison and sources.}
Section~\ref{sec:gap} compares explicit feasible schedules,
using Song's prescribed coefficient field and our least admissible
one. It does not optimize either construction or identify dimension
exponents with running-time exponents. Numbered references to
Song~\cite{song} refer to the 31-page manuscript bearing the
title-page date August 7, 2026. Song's result and the recent
refinements of Gao and Wan cited below are preprints; their
publication status is distinguished in the bibliography.

\paragraph{The route from logic to geometry.}
The construction follows
\[
 \begin{gathered}
 \text{3SAT}\ \longrightarrow\ \text{integer Boolean quadratic equations}\\
 \longrightarrow\ \text{linear codes over }\F_p\ \text{(CCP)}
       \ \longrightarrow\ \text{lattices (CVP)}.
 \end{gathered}
\]
The first reduction is valid over the integers and remains valid
over any field under the natural map on coefficients.

\paragraph{Three quantities to keep separate.}
For a 3SAT instance with $s$ variables, let $m_2,m_3$ count its
two- and three-literal clauses, and put $m=m_2+m_3$.
The quadratic system has $n=s+m_3$ variables, including one
auxiliary variable per three-literal clause and excluding the
homogenizing coordinate. When every clause has three literals,
$n=s+m$.
Write $K=\F_q$ and let $S\subseteq K$ be the sample set.
The encoding uses $r$ axes. The full
output length and weight thresholds will be
\begin{equation}\label{eq:three}
 \boxed{N=(n+1)^2q^r|S|,\qquad
 R_1=|S|(n+1)^2,\qquad R_0=\kappa(|S|-D).}
\end{equation}
Here $D$ is an interpolation-degree budget and $\kappa$ a rank
threshold. Completeness gives weight at most $R_1$ on a YES
instance; soundness excludes weight below $R_0$ on a NO instance.
For $|S|\ge4D$, the certified Hamming ratio is at least
$3\kappa/(4(n+1)^2)$. Its Euclidean counterpart is the square root.

\paragraph{Notation.}
Fix the prime $p$ before the input size grows.
Write $c_p=\max\{2,p-1\}$ for the window constant; $c$ without
a subscript denotes the extension degree $[E:K]$.
We reserve $t$ for target vectors, use $k$ for power and sequence
indices, $\nu$ for quadratic-equation indices, and $\ell$ for
recurrence orders. All vectors are columns, and matrices act on
the left. We use
$zz^\top$ for an outer product and $z^\top z$ for a scalar.
Rank over an integral domain means fraction-field rank; it is
unchanged by injective extension of fields, so we omit field
subscripts. We use $0^0=1$ and $\deg0=-\infty$.

\paragraph{Organization.}
Section~\ref{sec:quadratics} gives the quadratic encoding and the
intuition for the code construction. Section~\ref{sec:assignment-encoding}
sets up the fields and powers. Section~\ref{sec:table} contains
the seven conditions, the reduction theorem, and completeness.
Section~\ref{sec:rank} proves soundness through rank bounds,
recurrences, and assignment recovery. Sections~\ref{sec:lattice}
and~\ref{sec:gap} give the code-to-lattice transfer and the efficiency
comparison, and Section~\ref{sec:conclusions} discusses open problems.
All background material is collected in
Appendix~\ref{app:mathematics-background}, beginning with the
computational problems in Appendix~\ref{sec:computational} and
continuing with field theory, polynomial rank bounds, quotient
algebras, recurrences, and BCH decoding. It can be read first or
consulted as needed. Appendix~\ref{app:notation} collects notation
and the main proof dependencies.

\paragraph{A shorter reading route.}
Readers familiar with codes and lattices may start with
Sections~\ref{sec:quadratics}--\ref{sec:rank}, consulting
Appendix~\ref{sec:field-foundation} for field arithmetic,
Appendix~\ref{sec:crt-foundation} for quotient algebras, and
Lemma~\ref{lem:atoms} for the root expansion. The BCH material
in Appendix~\ref{sec:bch-background} is optional: it illustrates
the same algebra but neither decoder is an input to the reduction.

\section{3SAT and sparse quadratic equations}\label{sec:quadratics}

The signed equations below have integer coefficients. They may be
interpreted over $\Z$ or over any field $F$, with integer
coefficients mapped into $F$. Represent false by $0$ and true by
$1$. These are distinct elements in every field.
For each original variable impose
\[
 x_i^2-x_i=x_i(x_i-1)=0,
\]
which is equivalent to $x_i\in\{0,1\}$ over any field or over $\Z$, since a product is zero only if one factor is zero.
A positive literal has value $x_i$ and a negative literal has value
$1-x_i$. Write a three-literal clause as $C=(\ell_1\vee\ell_2\vee\ell_3)$,
where each $\ell_j$ denotes its literal's value. Then
\[
\begin{aligned}
 C\text{ is false}
 &\iff \ell_1=\ell_2=\ell_3=0\\
 &\iff 1-\ell_1=1-\ell_2=1-\ell_3=1\\
 &\iff (1-\ell_1)(1-\ell_2)(1-\ell_3)=1.
\end{aligned}
\]
Hence
\[
 C\text{ is true}
 \iff (1-\ell_1)(1-\ell_2)(1-\ell_3)=0.
\]
For fixed Boolean values of the original variables, this is
equivalent to the following two equations having a common solution
in an auxiliary variable $u$:
\begin{equation}\label{eq:quadratic-clause}
 u-(1-\ell_1)(1-\ell_2)=0,\qquad u(1-\ell_3)=0.
\end{equation}
Indeed, the first equation forces $u=(1-\ell_1)(1-\ell_2)$,
and substitution into the second gives the product condition.
This value of $u$ is automatically $0$ or $1$. We also include
$u^2-u=0$, so every original and introduced variable has its own
Booleanity equation.

\paragraph{A clause with two literals.}
For Boolean literal values, a clause $\ell_1\vee\ell_2$ is true
if and only if
\[
 (1-\ell_1)(1-\ell_2)=0.
\]
This is a single quadratic clause equation over the integers or
any field, and introduces no auxiliary variable.

\paragraph{A clause with $k$ literals.}
More generally, any clause $\ell_1\vee\cdots\vee\ell_k$ with
$k\ge2$ can be converted into $k-1$ quadratic clause equations
using $k-2$ new variables. The case $k=2$ is the preceding
identity. For $k\ge3$,
introduce $u_1,\ldots,u_{k-2}$ and impose
\[
 \begin{aligned}
 u_1-(1-\ell_1)(1-\ell_2)&=0,\\
 u_r-u_{r-1}(1-\ell_{r+1})&=0
       &&(2\le r\le k-2),\\
 u_{k-2}(1-\ell_k)&=0.
 \end{aligned}
\]
The first $k-2$ equations force
$u_r=\prod_{j=1}^{r+1}(1-\ell_j)$; the last equation is therefore
equivalent to $\prod_{j=1}^k(1-\ell_j)=0$, exactly the condition
that the clause is true. For Boolean literal values all the new
variables are automatically Boolean. The count $k-1$ refers to
the clause equations; Booleanity equations are counted separately.

For the 3SAT construction, now homogenize all equations by adding
one coordinate $z_0$.
For each literal $\ell_j$, define its homogeneous linear form by
\[
 L_j(z)=
 \begin{cases}
 x_i,&\text{if the literal is }x_i,\\
 z_0-x_i,&\text{if the literal is }\neg x_i.
 \end{cases}
\]
At $z_0=1$, this form equals $\ell_j$. The two clause equations become
\begin{equation}\label{eq:homogeneous-clause}
 z_0u-(z_0-L_1(z))(z_0-L_2(z))=0,
\qquad u(z_0-L_3(z))=0.
\end{equation}
A two-literal clause instead gives the single homogeneous form
\[
 (z_0-L_1(z))(z_0-L_2(z))=0.
\]
Relabeling the $n$ original and auxiliary variables as $x_1,\ldots,x_n$,
we also impose
\begin{equation}\label{eq:homogeneous-booleanity}
 x_i^2-z_0x_i=0\qquad(1\le i\le n).
\end{equation}
Every equation is now homogeneous of degree two.

The three-literal clause has a satisfying assignment if and only if its two
homogeneous clause equations, together with the Booleanity equations,
have a nonzero common solution over any field or over $\Z$. The same equivalence holds for the
whole formula when we include the appropriate equations for every clause.
To prove it, a satisfying assignment gives a solution with $z_0=1$
and each auxiliary variable set to the indicated product.
Conversely, if a solution has $z_0=0$, then
\eqref{eq:homogeneous-booleanity} gives $x_i^2=0$ for every $i$,
so the entire vector is zero. A nonzero solution therefore has
$z_0\ne0$. Working in the field (or in $\Q$ for an integer solution), dividing
all coordinates by $z_0$ preserves the homogeneous equations and gives
\[
 \left(\frac{x_i}{z_0}\right)^2-\frac{x_i}{z_0}=0,
 \qquad \frac{x_i}{z_0}\in\{0,1\}.
\]
The normalized clause equations are exactly
\eqref{eq:quadratic-clause} for three literals and the product
identity above for two literals, so the recovered Boolean assignment
satisfies every clause. This also shows that a nonzero real solution
exists if and only if a nonzero integer solution exists: the
normalized solution itself has only $0$--$1$ coordinates, with
first coordinate one. The Booleanity equations are part of this
nonzero-solution statement; the two clause equations alone can have
unwanted solutions with $z_0=0$.

More explicitly, let the 3SAT formula $\Phi$ have $s$ original
variables, $m_2$ two-literal clauses, and $m_3$ three-literal
clauses, with $m=m_2+m_3$. Only the three-literal clauses
introduce auxiliary variables, so the total number of variables
before homogenization is $n=s+m_3$. There are $n$ Booleanity
equations and $m_2+2m_3$ clause equations, hence
\[
 v=n+m_2+2m_3=s+m_2+3m_3
\]
homogeneous quadratic equations. Thus $n$ counts variables and $v$
counts equations; the homogeneous vector has $n+1$ coordinates.
In particular, $n\le s+m$ and $v\le3(s+m)$.
When $m_2=0$, these counts reduce to $n=s+m$ and
$v=s+3m\le3n$. Write the resulting system as
\begin{equation}\label{eq:forms}
 Q_\nu(z)=\sum_{i,j=0}^n a^{(\nu)}_{ij}z_i z_j=0,\qquad 1\le\nu\le v,
\end{equation}
where $a^{(\nu)}_{ij}\in\Z$ and $z=(z_0,x_1,\ldots,x_n)^\top$.

The formula $\Phi$ is satisfiable if and only if the
system~\eqref{eq:forms} has a nonzero solution over $\Z$, over $\R$,
or over any field after mapping the coefficients into that field.
In each case normalization produces the same Boolean witness.

\begin{remark}[Do not symmetrize the quadratic coefficient matrices]
The coefficients $a^{(\nu)}_{ij}$ need not form a symmetric matrix. To represent the monomial $z_1z_2$, put its coefficient at one ordered position, such as $(1,2)$. Putting a one at both $(1,2)$ and $(2,1)$ would make the two copies cancel in characteristic two. The assignment matrix $zz^\top$ below is symmetric; the coefficient matrix of a quadratic form need not be.
\end{remark}

\begin{example}[All variables and all quadratic equations for the running formula]\label{ex:quadratics}
Keep the five original bits $x_1,\ldots,x_5$ of~\eqref{eq:clause}. Introduce precisely three auxiliary bits, one per clause:
\[
 x_6=u_1,\qquad x_7=u_2,\qquad x_8=u_3.
\]
The following table lists the three factors $1-\ell_j$ for each clause and defines its auxiliary variable.
\begin{center}
\begin{tabular}{ccl}
\toprule
Clause & $(1-\ell_1,1-\ell_2,1-\ell_3)$ & Auxiliary definition\\
\midrule
$C_1$ & $(1-x_1,\ x_2,\ 1-x_3)$ & $x_6=(1-x_1)x_2$\\
$C_2$ & $(x_1,\ 1-x_4,\ 1-x_5)$ & $x_7=x_1(1-x_4)$\\
$C_3$ & $(1-x_2,\ 1-x_3,\ x_5)$ & $x_8=(1-x_2)(1-x_3)$\\
\bottomrule
\end{tabular}
\end{center}
Before homogenization, the six clause equations are
\begin{equation}\label{eq:running-affine-quadratics}
\begin{aligned}
 x_6-(1-x_1)x_2&=0,& x_6(1-x_3)&=0,\\
 x_7-x_1(1-x_4)&=0,& x_7(1-x_5)&=0,\\
 x_8-(1-x_2)(1-x_3)&=0,& x_8x_5&=0.
\end{aligned}
\end{equation}
The first equation in each row defines the auxiliary bit, and the second says that its clause is satisfied.
We also impose the eight Booleanity equations
\begin{equation}\label{eq:running-booleanity}
\begin{aligned}
 x_1^2-x_1&=0,&x_2^2-x_2&=0,&x_3^2-x_3&=0,&x_4^2-x_4&=0,\\
 x_5^2-x_5&=0,&x_6^2-x_6&=0,&x_7^2-x_7&=0,&x_8^2-x_8&=0.
\end{aligned}
\end{equation}
The first five constrain the original variables; the last three
constrain the introduced variables. Thus the full system has
fourteen quadratic equations in eight variables: $s=5$, $m_2=0$, $m_3=3$, $n=8$, and $v=14$.

Introduce the homogenizing coordinate $z_0$ and write
\[
 z=(z_0,x_1,x_2,x_3,x_4,x_5,x_6,x_7,x_8)^\top.
\]
Thus $n=8$, and the homogeneous coordinate vector has nine entries. Replacing constants and linear terms as above gives the six homogeneous clause forms:
\begin{equation}\label{eq:running-six-forms}
\begin{aligned}
 Q_1(z)&=z_0x_6-(z_0-x_1)x_2
        =z_0x_6-z_0x_2+x_1x_2,\\
 Q_2(z)&=x_6(z_0-x_3)=z_0x_6-x_3x_6,\\
 Q_3(z)&=z_0x_7-x_1(z_0-x_4)
        =z_0x_7-z_0x_1+x_1x_4,\\
 Q_4(z)&=x_7(z_0-x_5)=z_0x_7-x_5x_7,\\
 Q_5(z)&=z_0x_8-(z_0-x_2)(z_0-x_3)
        =z_0x_8-z_0^2+z_0x_2+z_0x_3-x_2x_3,\\
 Q_6(z)&=x_8x_5.
\end{aligned}
\end{equation}
The eight remaining homogeneous forms are
\begin{equation}\label{eq:running-boolean-forms}
 Q_{6+i}(z)=B_i(z)=x_i^2-z_0x_i
 \qquad(1\le i\le8).
\end{equation}
The full homogeneous system is $Q_\nu(z)=0$ for $1\le\nu\le v=14$.
The formula is satisfiable if and only if this system has a nonzero
real or integer solution. As proved above, every nonzero solution
has $z_0\ne0$ and normalizes to $z_0=1$, with all eight other
coordinates in $\{0,1\}$. In this example $z_i=x_i$ for $i\ge1$.
The later table construction encodes these normalized solutions.

For the satisfying original assignment $(1,0,1,0,1)^\top$, the auxiliary definitions give $(x_6,x_7,x_8)^\top=(0,1,0)^\top$. Consequently
\[
 x=(1,0,1,0,1,0,1,0)^\top,\qquad
 z=(1,1,0,1,0,1,0,1,0)^\top.
\]
Substitution in the six expanded clause forms gives, in order,
\[
 (0-0+0,\ 0-0,\ 1-1+0,\ 1-1,\ 0-1+0+1-0,\ 0)^\top
 =(0,0,0,0,0,0)^\top.
\]
All eight Booleanity forms also vanish, since $z_0=1$ and every
$x_i$ is zero or one. Thus all fourteen equations are satisfied.
For the false original assignment $(0,1,0,0,1)^\top$, the defining equations $Q_1=Q_3=Q_5=0$ force $(x_6,x_7,x_8)^\top=(1,0,0)^\top$. But then $Q_2=1$, so no auxiliary extension satisfies all six forms. More generally, for each of the $2^5$ original assignments the three defining equations force a unique auxiliary tuple, and that tuple satisfies the other three equations exactly when all three clauses are true.
\end{example}

The matrix $zz^\top$ for the satisfying example is
\begin{equation}\label{eq:running-matrix}
 zz^\top=
 \begin{pmatrix}
 1&1&0&1&0&1&0&1&0\\
 1&1&0&1&0&1&0&1&0\\
 0&0&0&0&0&0&0&0&0\\
 1&1&0&1&0&1&0&1&0\\
 0&0&0&0&0&0&0&0&0\\
 1&1&0&1&0&1&0&1&0\\
 0&0&0&0&0&0&0&0&0\\
 1&1&0&1&0&1&0&1&0\\
 0&0&0&0&0&0&0&0&0
 \end{pmatrix}.
\end{equation}
It has rank one and $25$ nonzero entries: the nonzero row and column indices are $\{0,1,3,5,7\}$. In general,
$\wt(zz^\top)=\wt(z)^2$ because an entry is one exactly when both corresponding coordinates of $z$ are one. Also $Q_\nu(z)=\sum_{ij}a^{(\nu)}_{ij}(zz^\top)_{ij}$. Thus quadratic information in $z$ becomes linear information in an assignment matrix. The challenge is to force a light linear-system solution to recover this rank-one form without imposing nonlinear rank-one constraints directly.

\begin{remark}[Boolean circuits as quadratic equations]
\label{rem:circuit-quadratics}
An acyclic Boolean circuit with $s$ input variables and $g$
two-input AND or OR gates can be encoded by $n=s+g$ quadratic
equations in $n$ variables over $\Z$ or $\R$.
Use one variable for each input and one new variable $z$ for
each gate output. Each gate input may be a literal: a previously
available variable $v$ or its negation, represented by $1-v$.
For the two literal values $x,y$, impose
\[
 \begin{aligned}
 z=x\wedge y&\quad\Longleftrightarrow\quad xy-z=0,\\
 z=x\vee y&\quad\Longleftrightarrow\quad x+y-xy-z=0.
 \end{aligned}
\]
Add $x_i^2-x_i=0$ for each of the $s$ input variables. There are
then $s+g=n$ equations, each of degree at most two. Once the
inputs are Boolean, the gate equations force the correct Boolean
output at every gate, by induction in a gate order respecting
the circuit edges. Thus every input assignment has a unique
extension to a solution, and every solution describes a valid
circuit evaluation; separate Booleanity equations for gate
outputs are unnecessary. To require the circuit to be satisfied,
also impose the linear equation $z_{\rm out}=1$ for its designated
output; this condition is additional to the $n$ evaluation equations.
\end{remark}

\section{Encoding assignments and powers efficiently}\label{sec:assignment-encoding}

\subsection{Injective encoding in a field extension}\label{sec:encoding-bases}

Reduce the signed integer coefficients of Section~\ref{sec:quadratics}
into $\F_p$, retaining every Booleanity equation
$x_i^2-z_0x_i=0$. The proof there shows that normalized solutions
over $\F_p$ are exactly the same Boolean witnesses as over
$\Z$ or $\R$. For odd $p$, Booleanity is not automatic and
will be essential in the recovered coefficient matrices.
The main construction assumes $n\ge\max\{4,p\}$.
Smaller instances are solved directly and mapped to fixed
outputs as described in Section~\ref{sec:code-instance}.

Choose $K=\F_q$, $q=p^a$, and a monic irreducible polynomial
$f\in K[X]$ of degree $c\ge2$, with $ac\ge n$. Put
\[
 E=K[X]/(f),\qquad \zeta=X\bmod f.
\]
If $(\beta_1,\ldots,\beta_a)$ is an ordered $\F_p$-basis of
$K$, the products $\beta_\ell\zeta^j$, in order of increasing
$j$ and then $\ell$, form an $\F_p$-basis of $E$.
Indeed, independence follows first by the $K$-basis
$1,\zeta,\ldots,\zeta^{c-1}$, then by the basis of $K$.
Choose the first $n$ product-basis elements as
$\alpha_1,\ldots,\alpha_n$ and define
\begin{equation}\label{eq:encoding}
 \Enc:\F_p^n\longrightarrow E,\qquad
 \Enc(x)=\sum_{i=1}^n\alpha_i x_i.
\end{equation}
Prime-field independence proves injectivity. Every element of
$E$ is stored using $ac$ prime-field coordinates, or
$O(ac\log p)$ bits; enumerating $E$ is unnecessary.
Appendix~\ref{sec:field-foundation} gives the field and basis details.

\subsection{High powers from low-degree polynomials}
\label{sec:powers-table}

For a normalized Boolean vector $z=(1,x^\top)^\top$, put
$\xi=\Enc(x)\in E$. For every polynomial $Q$ over $\F_p$,
\[
 Q(z)=0\quad\Longleftrightarrow\quad
 Q(z)\xi^k=0\text{ for every integer }k\ge0;
\]
the reverse implication follows by taking $k=0$. In particular,
\[
 Q_\nu(z)\xi^k
   =\sum_{i,j}a^{(\nu)}_{ij}\bigl(z_i z_j\xi^k\bigr).
\]
Replacing the quantities in parentheses by unknown moments
$\eta_{ij,k}$ makes these equations linear in the moments.
The soundness proof must establish the difficult converse:
a light solution of the resulting linear system yields a
satisfying assignment.

Storing all assignments would require exponential space, while
expanding $(\sum_i\alpha_i x_i)^k$ directly gives degree $k$.
Instead, Frobenius expresses large powers through products of
polynomials of small degree, as follows.

Choose a radix $d=p^b\ge2$ and a fixed number $r\ge1$ of axes.
Write
\[
 \alpha_i^{d^h}=\sum_{j=0}^{c-1}\gamma_{i,j,h}\zeta^j,
 \qquad \gamma_{i,j,h}\in K,
\]
and define
\begin{equation}\label{eq:axis-coefficients}
 P_{x,h}(X)=\sum_{j=0}^{c-1}
       \left(\sum_i\gamma_{i,j,h}x_i\right)X^j\in K[X]
       \qquad(0\le h<r).
\end{equation}
The identity $(u+v)^{p^k}=u^{p^k}+v^{p^k}$ and
$x_i^{p^k}=x_i$ for $x_i\in\F_p$ give
\begin{equation}\label{eq:axis}
 \deg P_{x,h}<c,\qquad
 P_{x,h}(\zeta)=\Enc(x)^{d^h}.
\end{equation}
It is the unique polynomial of degree less than $c$ with that
value: a nonzero polynomial of smaller degree cannot vanish at
$\zeta$. For a nonnegative digit vector
$\mathbf e=(e_0,\ldots,e_{r-1})^\top\in\Z_{\ge0}^r$, define
$k=\ind(\mathbf e)=\sum_{h=0}^{r-1}e_hd^h$. Then
\begin{equation}\label{eq:fold}
 \Enc(x)^k=\prod_{h=0}^{r-1}P_{x,h}(\zeta)^{e_h}.
\end{equation}
Although $d^h$ is large, $P_{x,h}$ is linear in the $x_i$ and
has degree less than $c$ in $X$.

For any nonnegative digit vector $\mathbf e$, the polynomial
$\prod_hP_{x,h}(X)^{e_h}$ therefore has degree at most
$(c-1)\sum_he_h$. Its value at $\zeta$ is the possibly very
large power $\xi^{\ind(\mathbf e)}$.

This representation suggests the table used in the reduction.
Let $S\subseteq K$ be a set of distinct sample points, whose
size will be chosen below. At each $\sigma\in S$, the vector
\[
 \boldsymbol\lambda_x(\sigma)
   =\bigl(P_{x,0}(\sigma),\ldots,P_{x,r-1}(\sigma)\bigr)\in K^r
\]
is a label for the assignment at that sample. Introduce the
common table
\[
 T=(T_{ij,\boldsymbol\lambda,\sigma}),\qquad
 T_{ij,\boldsymbol\lambda,\sigma}\in\F_p,
 \quad 0\le i,j\le n,\quad
 \boldsymbol\lambda\in K^r,\quad\sigma\in S.
\]
For each label and sample, its entries form the matrix
$T_{\boldsymbol\lambda,\sigma}
=(T_{ij,\boldsymbol\lambda,\sigma})_{0\le i,j\le n}$.
The honest table associated with $x$ is
\[
 T_{ij,\boldsymbol\lambda,\sigma}
 =\begin{cases}
    z_i z_j,&\boldsymbol\lambda=\boldsymbol\lambda_x(\sigma),\\
    0,&\boldsymbol\lambda\ne\boldsymbol\lambda_x(\sigma).
   \end{cases}
\]
Thus it places $zz^\top$ at one label per sample and zero at
every other label. In particular, its entries are in $\{0,1\}$,
even though arbitrary table entries range over $\F_p$.
For this honest table, using $0^0=1$,
\[
 \sum_{\boldsymbol\lambda\in K^r}
       T_{ij,\boldsymbol\lambda,\sigma}
          \prod_h\lambda_h^{e_h}
   =z_i z_j\prod_hP_{x,h}(\sigma)^{e_h}
   =G_{ij,\mathbf e}(\sigma),
\]
where
\[
 G_{ij,\mathbf e}(X)
   =z_i z_j\prod_{h=0}^{r-1}P_{x,h}(X)^{e_h}\in K[X],
 \qquad
 \deg G_{ij,\mathbf e}\le(c-1)\sum_{h=0}^{r-1}e_h.
\]
For an unknown table, the left side is a linear readout of its
entries. Requiring these readouts to interpolate to polynomials
of the stated small degree lets us evaluate them at $\zeta$,
where the honest values are $z_i z_j\xi^{\ind(\mathbf e)}$.
The same table thus supplies linear substitutes for all the
quadratic equations and all the required power tests. Its
coordinate set is fixed without knowing $x$; the affine conditions
on its entries are given in Section~\ref{sec:table}.

\subsection{From the rank cutoff to the moment window and digit products}
\label{sec:cutoff-digits}

Choose the rank cutoff $\kappa\ge2$ first. For a sample set
$S\subseteq K$ and an interpolation-degree budget $D<|S|$,
the soundness threshold is
\begin{equation}\label{eq:cutoff-threshold}
 R_0=\kappa(|S|-D).
\end{equation}
A table of weight less than $R_0$ will have moment-matrix rank
less than $\kappa$, by Lemma~\ref{lem:rank-weight}.
The recovery proof uses a common recurrence for all moment
sequences in a fixed block column. Their joint root support has
size at most $\kappa-1$, independently of $n$; see
Lemma~\ref{lem:block-column-atoms}. Symmetry compares two such
supports, and the Frobenius tests require moments at indices
$pb$ for recurrence-basis indices $b$. A sufficient row range
covering both requirements is
\begin{equation}\label{eq:capacity}
 c_p=\max\{2,p-1\},\qquad L=c_p(\kappa-1)+1.
\end{equation}
In particular, $L=2\kappa-1$ for $p=2$ or $p=3$.
We claim sufficiency of this window, not its optimality.
The cutoff is a design parameter, not a radix capacity.
Section~\ref{sec:gap} chooses it as a function of $n$ and the
desired gap exponent, charging the full output length.

\paragraph{Connection with BCH syndromes.}
Consecutive powers also underlie BCH codes and their decoding.
Their syndromes are weighted power sums
$s_k=\sum_j e_j\beta_j^k$. Vandermonde independence separates
the location contributions, while their annihilating polynomial
specifies a recurrence. Appendix~\ref{sec:bch-background} gives
the code definition, the BCH distance bound, and a decoding
procedure, including the connection with Berlekamp--Massey.
The additional information here is simultaneous recurrence data:
the $n+1$ sequences $(\eta_{ij,k})_k$, with $j$ fixed, share
one recurrence polynomial. Bounding its root set directly avoids
adding $n+1$ separate support bounds. Section~\ref{sec:atoms}
proves this joint bound before the isolation step of
Section~\ref{sec:synchronization}.

The decomposition $k=u+v$ is dictated by the later moment
matrix, whose rows are indexed by $0\le u<L$ and columns by
$0\le v<\kappa$. Its honest $(i,j)$ entry is
\[
 \eta_{ij,u+v}=z_i z_j\xi^{u+v}
       =(z_i\xi^u)(z_j\xi^v).
\]
This factorization gives a rank-one matrix for a single
assignment and, at samples, bounds matrix rank by table weight.
The largest required index is $L+\kappa-2$. Define the number
of exposed moments by
\begin{equation}\label{eq:moment-window}
 J_0=L+\kappa-1=(c_p+1)(\kappa-1)+1,
       \qquad 0\le k<J_0.
\end{equation}
Thus $J_0$ denotes the length of the window, not its last index.
For this prescribed window, choose $r$, $q$, and $d$ jointly
to minimize the label-space size $q^r$, while keeping $d^r$
close to $J_0$ from above. The choices must satisfy the field,
degree, and sampling constraints below; in particular, $q$ must
be large enough to accommodate the samples. The full output
length also includes the factor $|S|$, so final parameter
comparisons use $q^r|S|$, as detailed in Section~\ref{sec:gap}.

Only now choose the radix $d=p^b\ge2$ and the number $r$ of
axes so that
\begin{equation}\label{eq:digit-capacity}
 d^r\ge J_0.
\end{equation}
Every required $k$ then has a unique ordinary base-$d$ digit
column $\Lift(k)\in\{0,\ldots,d-1\}^r$. For fixed $r$,
choose the smallest power of $p$ meeting
\eqref{eq:digit-capacity}, unless another admissible choice
gives a smaller full output. In the reverse direction, fixed
$d,r$ permit the cutoff
\[
 \kappa=1+\left\lfloor\frac{d^r-1}{c_p+1}\right\rfloor,
\]
provided this is at least two. Notice that $L$ need not equal
$d^r$.

Two kinds of exponent vectors are useful: canonical digits name
the moments, and uncarried sums preserve multiplication at each
sample. Accordingly define
\begin{equation}\label{eq:exponent-set}
 \begin{aligned}
 \mathcal P_{\mathrm{can}}
   &=\{\Lift(k):0\le k<J_0\},\\
 \mathcal P_{\mathrm{prod}}
   &=\{\Lift(u)+\Lift(v):0\le u<L,\ 0\le v<\kappa\},\\
 \mathcal P&=\mathcal P_{\mathrm{can}}\cup\mathcal P_{\mathrm{prod}}.
 \end{aligned}
\end{equation}
For any $x\in\F_p^n$ and $\mathbf e,\mathbf e'\in\mathcal P$,
whenever $\ind(\mathbf e)=\ind(\mathbf e')$, we have
\[
 f(X)\mid\bigl(G_{ij,\mathbf e}(X)-G_{ij,\mathbf e'}(X)\bigr)
 \qquad(0\le i,j\le n).
\]
Indeed, by \eqref{eq:fold} both polynomials have the same value
$z_i z_j\Enc(x)^{\ind(\mathbf e)}$ at $\zeta$, whose minimal
polynomial over $K$ is $f(X)$.
Every vector in $\mathcal P$ has index in $[0,J_0-1]$.
Every such index also has a product representation: take
$u=\min\{k,L-1\}$ and $v=k-u$.
No exponent vectors outside these two sets are needed.

The ordinary digits are less than $d$, but an uncarried sum can
have an entry as large as $2d-2$. Condition 2 of
Section~\ref{sec:table} identifies representations of the same
index only after evaluation at $\zeta$. It says
\[
 f(X)\mid F_{ij,\mathbf e}(X)-F_{ij,\mathbf e'}(X),
       \qquad\ind(\mathbf e)=\ind(\mathbf e'),
\]
not that the two polynomials or their sample values coincide.
For example, if $r\ge2$, $L\ge d$, and $\kappa\ge4$, then
$u=d-1$, $v=3$ have uncarried first digit $d+2$ when $d\ge4$,
even though their sum has ordinary digits $(2,1,0,\ldots,0)$.
The corresponding label weights are $\lambda_0^{d+2}$ and
$\lambda_0^2\lambda_1$, which differ at
$(\lambda_0,\lambda_1)=(1,0)$. Keeping the uncarried vector
supplies the samplewise product identity used in
Lemma~\ref{lem:moment-evaluations}; canonical digits alone do
not supply it.

For a simple uniform degree budget use
\begin{equation}\label{eq:degree-budget}
 D=2r(d-1)(c-1).
\end{equation}
Indeed $\mathcal P\subseteq\{0,\ldots,2d-2\}^r$ and each
$P_{x,h}$ has degree less than $c$. The exponent set itself is
already restricted to the required moments and products.
Song's degree budget includes a synchronization
multiplier $n+2$~\cite[Definitions 3.5--3.6]{song}; here the
shared block-column recurrence gives a row range depending on
$p$ and $\kappa$, and ordinary factor digits suffice.

\begin{example}[Exponent eleven from a polynomial of degree four]\label{ex:axes}
Use the binary tower $K=\F_4$, $E=\F_{16}$ from
Example~\ref{ex:tower}, where $\omega^2=\omega+1$ and
$\zeta^2=\zeta+\omega$. With basis
$(1,\omega,\zeta,\omega\zeta)$ and $x=(1,0,1,0)^\top$,
the encoding is $\xi=1+\zeta$. For $d=2,r=3$, its representing
polynomials are
\[
 P_0(X)=1+X,\qquad P_1(X)=1+\omega+X,\qquad P_2(X)=X,
\]
since $\xi^2=1+\omega+\zeta$ and $\xi^4=\zeta$.
The nonordinary digit vector $(1,1,2)^\top$ lies in the enclosing box
$\{0,1,2\}^3$ and has index $1+2+2\cdot4=11$. Its product is
\[
 P_0P_1P_2^2=X^4+\omega X^3+(1+\omega)X^2.
\]
Using $\zeta^3=(1+\omega)\zeta+\omega$ and
$\zeta^4=\zeta+1$, direct evaluation gives
\[
 \zeta^4+\omega\zeta^3+(1+\omega)\zeta^2
       =(1+\omega)(1+\zeta)=\zeta^{-1}.
\]
This agrees with $\xi^{11}=(\xi^4)^{-1}=\zeta^{-1}$, since
$\xi^{15}=1$. The example illustrates the encoding identity
for general nonnegative digit vectors; its tiny field is not claimed to meet the later
sample-size budget for a full reduction instance.
\end{example}

\section{From quadratic equations to a linear code: construction and completeness}
\label{sec:table}

\subsection{Parameters and interpolation}

\paragraph{Parameter hypotheses and sample readouts.}
The construction uses
\begin{equation}\label{eq:budgets}
 \begin{gathered}
 n\ge\max\{4,p\},\quad q=|K|=p^a,\quad c\ge2,\quad ac\ge n,
 \quad d=p^b,\quad \kappa\ge2,\\
 c_p=\max\{2,p-1\},\qquad L=c_p(\kappa-1)+1,\\
 J_0=L+\kappa-1=(c_p+1)(\kappa-1)+1\le d^r,\\
 D=2r(d-1)(c-1),\qquad S\subseteq K,\qquad |S|\ge4D,\\
 N=(n+1)^2q^r|S|.
 \end{gathered}
\end{equation}
Here $\mathcal P$ is the exponent set in \eqref{eq:exponent-set},
$S$ is any explicitly listed set of distinct field elements, and
$T$ is the common table from Section~\ref{sec:powers-table}.
Its $N$ entries over $\F_p$ are the unknowns, and its Hamming
weight counts nonzero entries. Only samples are restricted to $S$;
the label alphabet remains $K^r$.

For a digit vector $\mathbf e$, define
\[
 w_{\boldsymbol\lambda}(\mathbf e)
    =\prod_{h=0}^{r-1}\lambda_h^{e_h},\qquad 0^0=1.
\]
For integer indices $0\le u<L$ abbreviate
$w_{\boldsymbol\lambda}(u)=w_{\boldsymbol\lambda}(\Lift(u))$.
Vector weights multiply under addition of digit vectors, but
need not be equal just because their integer indices agree.
Define sample readouts
\begin{equation}\label{eq:weighted-values}
 f_{ij,\mathbf e}(\sigma)
       =\sum_{\boldsymbol\lambda\in K^r}
             T_{ij,\boldsymbol\lambda,\sigma}
                    w_{\boldsymbol\lambda}(\mathbf e)
       \qquad(\mathbf e\in\mathcal P).
\end{equation}

\paragraph{Interpolation on an arbitrary sample set.}

Let $g_S(X)=\prod_{\tau\in S}(X-\tau)$. Distinctness of the
samples gives $g_S'(\sigma)=\prod_{\tau\ne\sigma}(\sigma-\tau)\ne0$.
The explicit interpolation polynomials are
\begin{equation}\label{eq:sample-lagrange}
 \ell^S_\sigma(X)
  =\prod_{\tau\in S\setminus\{\sigma\}}
                      \frac{X-\tau}{\sigma-\tau}
  =\frac{g_S(X)}{(X-\sigma)g_S'(\sigma)}.
\end{equation}
They are one at their own sample and zero at all other samples.
Consequently the unique polynomial of degree less than $|S|$
with the readout values is
\[
 F_{ij,\mathbf e}(X)
       =\sum_{\sigma\in S}f_{ij,\mathbf e}(\sigma)\ell^S_\sigma(X).
\]
Existence follows from the formula; uniqueness follows because
a nonzero polynomial of degree less than $|S|$ cannot have $|S|$
distinct roots. When $S=K$, this specializes to
$\ell^K_\sigma(X)=1-(X-\sigma)^{q-1}$.
For a proper subset one must use \eqref{eq:sample-lagrange} instead.

\subsection{The seven families of affine conditions}

The integer coefficients $a^{(\nu)}_{ij}$ of $Q_\nu$ are naturally
viewed as elements of $\F_p$ when used in field equations.
Impose the following on the common table.
\begin{enumerate}[leftmargin=*]
\item \textbf{Degree.} For every $(i,j)$ and $\mathbf e\in\mathcal P$,
require $\deg F_{ij,\mathbf e}(X)\le D$.
\item \textbf{Consistency of exponents.} If
$\ind(\mathbf e)=\ind(\mathbf e')$, require
$F_{ij,\mathbf e}(\zeta)=F_{ij,\mathbf e'}(\zeta)$.
This defines moments
\[
 \eta_{ij,k}=F_{ij,\mathbf e}(\zeta),\qquad
       \ind(\mathbf e)=k,\quad 0\le k<J_0.
\]
For each fixed pair $(i,j)$, we call
$(\eta_{ij,k})_{0\le k<J_0}$ the \emph{moment sequence associated
with $(i,j)$}; only the power index $k$ varies.
\item \textbf{Frobenius.} Require
$\eta_{ij,pk}=\eta_{ij,k}^{\,p}$ whenever $pk<J_0$.
\item \textbf{Assignment encoding.} For $0\le j\le n$ and
$0\le k<J_0-1$, require
\begin{equation}\label{eq:encoding-moment}
 \eta_{0j,k+1}=\sum_{i=1}^n\alpha_i\eta_{ij,k}.
\end{equation}
\item \textbf{The quadratic equations.} For every source form $Q_\nu$,
including every Booleanity form, and $0\le k<J_0$, require
\begin{equation}\label{eq:quadratic-moment}
 \sum_{i,j=0}^na^{(\nu)}_{ij}\eta_{ij,k}=0.
\end{equation}
\item \textbf{Symmetry.} Require
$T_{ij,\boldsymbol\lambda,\sigma}=T_{ji,\boldsymbol\lambda,\sigma}$.
\item \textbf{Normalization.} Require
$\sum_{\boldsymbol\lambda}T_{00,\boldsymbol\lambda,\sigma}=1$
for every $\sigma\in S$.
\end{enumerate}
All quadratics constrain the same moments. They are not replaced
by one sum of quadratic equations, in which violations could cancel.

\paragraph{Remark on Condition 4: multiplication by the assignment encoding.}
For an honest table associated with a satisfying assignment $x$,
put $z=(1,x^\top)^\top$ and $\xi=\Enc(x)$. Its moments are
$\eta_{ij,k}=z_i z_j\xi^k$. Since $z_0=1$ and
$\sum_{i=1}^n\alpha_i z_i=\xi$, the right side of Condition 4 is
\[
 \sum_{i=1}^n\alpha_i\eta_{ij,k}
   =z_j\xi^k\sum_{i=1}^n\alpha_i z_i
   =z_j\xi^{k+1}=\eta_{0j,k+1}.
\]
Thus shifting the moment index by one represents multiplication
by $\xi$. The table equations express this relation linearly,
using the fixed coefficients $\alpha_i$.

\paragraph{Remark on Condition 5: the quadratic equations in moment form.}
For the same honest moments,
\[
 \sum_{i,j=0}^n a^{(\nu)}_{ij}\eta_{ij,k}
   =\xi^k\sum_{i,j=0}^n a^{(\nu)}_{ij}z_i z_j
   =\xi^k Q_\nu(z).
\]
Each satisfying assignment therefore obeys Condition 5 for every
available $k$. Requiring all these moment equations is also
essential to soundness: the consecutive powers allow us to isolate
each root contribution and recover $Q_\nu(z)=0$ for its assignment.

\paragraph{Remark on Condition 7: normalization after summing over labels.}
At each sample $\sigma$, the honest table has one nonzero block,
located at the assignment-dependent label
$\boldsymbol\lambda_x(\sigma)$, and the $(0,0)$ entry of that
block is $z_0^2=1$. Summing over all labels therefore gives one,
independently of the location of this block. This imposes a fixed
affine normalization without knowing the assignment. For every
feasible table it gives
\[
 f_{00,\mathbf0}(\sigma)
   =\sum_{\boldsymbol\lambda}T_{00,\boldsymbol\lambda,\sigma}=1,
 \qquad F_{00,\mathbf0}(X)=1,
 \qquad \eta_{00,0}=1.
\]
The middle equality follows by uniqueness of interpolation on
$S$. This normalization will ensure that the root expansion used
in the recovery argument has a nonzero contribution.

\begin{lemma}[The seven families are affine over the prime field]
\label{lem:affine}
Expanding the field equations in fixed $\F_p$-bases gives an
explicit affine system $H\theta=b$ over $\F_p$, with exactly $N$
table coordinates. No product of unknown coordinates is required.
\end{lemma}
\begin{proof}
Readouts, interpolation, and evaluation are linear. The degree
condition sets the coefficients of $X^{D+1},\ldots,X^{|S|-1}$
to zero. Multiplication by a fixed field element is linear over
$\F_p$. Finally, if $\eta=\sum_\nu c_\nu\theta_\nu$, then
\[
 \eta^p=\sum_\nu c_\nu^p\theta_\nu^p
       =\sum_\nu c_\nu^p\theta_\nu
       \qquad(\theta_\nu\in\F_p).
\]
Thus the Frobenius conditions are linear too. Normalization has
right side one. A $K$-valued equation expands into $a$ scalar
equations; an $E$-valued equation expands into $ac$.
\end{proof}

For a consistent system compute $t_0$ with $Ht_0=b$ and a column
basis for $C=\ker H$. Equation~\eqref{eq:syndrome-code} identifies
its closest codeword distance with minimum feasible table weight.
Thus the affine description already specifies the code reduction.
Section~\ref{sec:lattice} supplies the feasible NO replacement and
the integer lattice output.

For a concrete view of the linear map, fix any representative
$\mathbf e(k)\in\mathcal P$ of $k$. Direct substitution gives
\[
 \eta_{ij,k}=\sum_{\sigma\in S}
            \sum_{\boldsymbol\lambda\in K^r}
      \ell^S_\sigma(\zeta)w_{\boldsymbol\lambda}(\mathbf e(k))
                       T_{ij,\boldsymbol\lambda,\sigma}.
\]
Every displayed coefficient is fixed before solving the system.
The prime-field coordinate vector of all moments is therefore
$\mathbf m=\Pi\theta$, with $(n+1)^2J_0ac$ rows in $\Pi$.
This count is not a rank assertion and does not say that the map
preserves Hamming weight.

\subsection{The reduction theorem and completeness}\label{sec:completeness}

\begin{theorem}[Quadratics to a prime-field CCP instance]
\label{thm:table-reduction}
Under \eqref{eq:budgets}, the seven families construct an affine
system whose length and thresholds are \eqref{eq:three}.
A satisfying formula gives a feasible table in $\{0,1\}^N$
of weight at most $R_1$. Every feasible table over $\F_p$ of
weight less than $R_0$ implies a satisfying formula.
For fixed $p,r$ and polynomial explicit parameters, the system
and the corresponding feasible code instance are constructed in
deterministic polynomial time, using the field-construction input
specified in Section~\ref{sec:running-time}.
\end{theorem}

We prove completeness here and soundness in the next section.
Suppose $x$ satisfies the source, and put
$z=(1,x^\top)^\top$ and $\xi=\Enc(x)$. For $\sigma\in S$,
define the feature label
\[
 \boldsymbol\lambda_x(\sigma)
     =(P_{x,0}(\sigma),\ldots,P_{x,r-1}(\sigma))^\top.
\]
Set $T_{\boldsymbol\lambda_x(\sigma),\sigma}=zz^\top$ and
$T_{\boldsymbol\lambda,\sigma}=0_{(n+1)\times(n+1)}$
for all other labels. Then
\[
 f_{ij,\mathbf e}(\sigma)
    =z_i z_j\prod_hP_{x,h}(\sigma)^{e_h}.
\]
The product polynomial has degree at most $D<|S|$. Uniqueness
of interpolation identifies it with $F_{ij,\mathbf e}(X)$.
Evaluation at $\zeta$ gives
\begin{equation}\label{eq:honest-moments}
 \eta_{ij,k}=z_i z_j\xi^k.
\end{equation}
Degree and exponent consistency follow. Frobenius holds because
$z_i z_j\in\F_p$; the encoding identity follows from
$\xi=\sum_i\alpha_i x_i$; the quadratic identity is
$\xi^kQ_\nu(z)=0$. Symmetry and normalization use the symmetric
matrix $zz^\top$ and $z_0=1$. Since $z\in\{0,1\}^{n+1}$,
\begin{equation}\label{eq:yes-weight}
 \wt(T)=|S|\wt(z)^2\le |S|(n+1)^2=R_1.
\end{equation}
This witness is used to prove completeness, not to construct the
linear system: its unknown satisfying assignment is never enumerated.

\section{Soundness: rank, recurrences, and assignment recovery}\label{sec:rank}

The honest moments satisfy
$\eta_{ij,u+v}=(z_i\xi^u)(z_j\xi^v)$, so their moment matrix
has rank one at $\zeta$. An arbitrary feasible table need not have
that form. We first relate its weight to polynomial matrix rank.
In the low-rank branch we then obtain short recurrences, remove
nilpotent contributions using Frobenius, and pass each constraint
to the coefficient of each distinct root. These steps establish
that every sufficiently light feasible table yields a satisfying
assignment.

\subsection{The polynomial moment matrix and rank bounds}

Define $\mathcal U(X)\in K[X]^{(n+1)L\times(n+1)\kappa}$ by
\begin{equation}\label{eq:moment-matrix}
 \mathcal U(X)_{(i,u),(j,\ell)}
       =F_{ij,\Lift(u)+\Lift(\ell)}(X),\qquad
       0\le u<L,\quad0\le\ell<\kappa.
\end{equation}
The digit sum lies in $\mathcal P$, so every matrix entry is
defined. The next lemma describes the evaluations at $\zeta$ and
at the samples, and relates sample ranks to table weight.

\begin{lemma}[Evaluations of the moment matrix and sample ranks]
\label{lem:moment-evaluations}
For any table $T$ satisfying the exponent-consistency condition,
\begin{equation}\label{eq:hankel-specialization}
 \mathcal U(\zeta)_{(i,u),(j,\ell)}=\eta_{ij,u+\ell}.
\end{equation}
For every table $T$, whether or not it satisfies that condition,
\begin{equation}\label{eq:sample-rank}
 \sum_{\sigma\in S}\rank\mathcal U(\sigma)
    \le\sum_{\sigma\in S}\sum_{\boldsymbol\lambda\in K^r}
                              \rank T_{\boldsymbol\lambda,\sigma}
    \le\wt(T),
\end{equation}
and $\rank\mathcal U(\zeta)\le\rank\mathcal U(X)$.
\end{lemma}
\begin{proof}
The index map is additive on digit vectors. Thus
\[
 \ind\bigl(\Lift(u)+\Lift(\ell)\bigr)=u+\ell,
 \qquad u+\ell\le L+\kappa-2=J_0-1.
\]
Condition 2 therefore identifies the evaluation of each entry in
\eqref{eq:moment-matrix} with the indicated moment, proving
\eqref{eq:hankel-specialization}.

At a sample $\sigma$, interpolation gives
$F_{ij,\mathbf e}(\sigma)=f_{ij,\mathbf e}(\sigma)$. Also,
$w_{\boldsymbol\lambda}(\mathbf e+\mathbf e')
=w_{\boldsymbol\lambda}(\mathbf e)
 w_{\boldsymbol\lambda}(\mathbf e')$. Substitution into
\eqref{eq:weighted-values} yields
\[
 \mathcal U(\sigma)_{(i,u),(j,\ell)}
       =\sum_{\boldsymbol\lambda}
           T_{ij,\boldsymbol\lambda,\sigma}
            w_{\boldsymbol\lambda}(u)w_{\boldsymbol\lambda}(\ell).
\]
Define columns
$\mathbf v_{\boldsymbol\lambda}
 =(w_{\boldsymbol\lambda}(u))_{0\le u<L}\in K^L$ and
$\mathbf w_{\boldsymbol\lambda}
 =(w_{\boldsymbol\lambda}(\ell))_{0\le\ell<\kappa}\in K^\kappa$.
In block form the preceding identity is
\[
 \mathcal U(\sigma)
   =\sum_{\boldsymbol\lambda\in K^r}
      T_{\boldsymbol\lambda,\sigma}\otimes
       (\mathbf v_{\boldsymbol\lambda}
        \mathbf w_{\boldsymbol\lambda}^{\top}),
\]
where $A\otimes B$ is the block matrix with $(i,j)$ block
$A_{ij}B$. To bound one summand, write a rank-$\tau$ table block as
$T_{\boldsymbol\lambda,\sigma}
=\sum_{h=1}^\tau a_h b_h^\top$. Its contribution is then
\[
 \sum_{h=1}^\tau
 (a_h\otimes\mathbf v_{\boldsymbol\lambda})
 (b_h\otimes\mathbf w_{\boldsymbol\lambda})^\top,
\]
a sum of $\tau$ matrices of rank at most one. Rank subadditivity
therefore gives
$\rank\mathcal U(\sigma)
\le\sum_{\boldsymbol\lambda}\rank T_{\boldsymbol\lambda,\sigma}$.
For any matrix $B$, writing $B$ as the sum of its single-entry
matrices gives $\rank B\le\wt(B)$. Summing over the disjoint
table blocks and samples proves \eqref{eq:sample-rank}.

Finally, every minor larger than $\rank\mathcal U(X)$ vanishes
as a polynomial. It remains zero after evaluation at $\zeta$,
so specialization cannot increase rank.
\end{proof}

The degree condition converts the preceding sample-rank bound
into a weight bound for the polynomial matrix rank, without
first choosing a threshold.

\begin{lemma}[Polynomial rank forces table weight]
\label{lem:rank-weight}
Let $\rho=\rank\mathcal U(X)$. If $T$ satisfies Condition 1,
that is, $\deg F_{ij,\mathbf e}(X)\le D$ for every
$0\le i,j\le n$ and $\mathbf e\in\mathcal P$, then
\begin{equation}\label{eq:rank-weight}
 \wt(T)\ge\sum_{\sigma\in S}\rank\mathcal U(\sigma)
       \ge\rho(|S|-D).
\end{equation}
In particular, $\rho\ge\kappa$ implies
$\wt(T)\ge\kappa(|S|-D)=R_0$.
\end{lemma}
\begin{proof}
The first inequality is \eqref{eq:sample-rank}. For the second,
apply Proposition~\ref{prop:sample-rank-foundation} over $K$ to
$\mathcal U(X)$: its rank over $K(X)$ is $\rho$, and Condition 1
bounds the degree of each entry by $D$.
The final assertion uses $|S|\ge4D$, so $|S|-D>0$.
\end{proof}

The bound charges the rank lost at each sample to the multiplicity
of a zero of one determinant. Thus a sample count of order $D$
suffices to retain a constant proportion of the total rank.
The underlying determinant-multiplicity statement is proved
over an arbitrary field in Appendix~\ref{sec:rank-deficiency}.

\subsection{A shared block-column recurrence on a finite window}
\label{sec:atoms}

Assume $\wt(T)<R_0$. By \eqref{eq:rank-weight}, $\rho<\kappa$.
All $(\rho+1)$-minors vanish identically, so specialization gives
$\rank\mathcal U(\zeta)\le\rho$. For a fixed $j$, retain all
rows and the $\kappa$ columns indexed by $(j,b)$:
\begin{equation}\label{eq:block-column-matrix}
 B_j=\bigl(\eta_{ij,u+b}\bigr)_{(i,u),b},\qquad
 0\le i\le n,\quad 0\le u<L,\quad 0\le b<\kappa.
\end{equation}
This submatrix has rank at most $\rho$. Its column dependence
therefore supplies a recurrence simultaneously for all $n+1$
sequences in that column, which is stronger than choosing a
separate recurrence for each scalar block.

\begin{lemma}[Common recurrence and joint root support]
\label{lem:block-column-atoms}
Let $T$ be a feasible table under \eqref{eq:budgets}, with
$\wt(T)<R_0$, and put $\rho=\rank\mathcal U(X)$.
For each $j$ there is a monic polynomial $\chi_j\in E[Z]$ of degree
$0\le\ell_j\le\rho$ and root sets $\mathcal E_{ij}$ such that,
in a common splitting field $\Omega/E$,
\begin{equation}\label{eq:entry-atoms}
 \eta_{ij,k}=\sum_{\beta\in\mathcal E_{ij}}a_{ij,\beta}\beta^k,
 \qquad a_{ij,\beta}\in\F_p^\times,\qquad 0\le k<L.
\end{equation}
All roots in $\mathcal E_{ij}$ are roots of $\chi_j$. In particular,
\begin{equation}\label{eq:column-root-bound}
 \mathcal E_j=\bigcup_{i=0}^n\mathcal E_{ij},\qquad
 |\mathcal E_j|\le\ell_j\le\rho<\kappa.
\end{equation}
If $\ell_j=0$, then $\chi_j=1$ and all these sets are empty.
\end{lemma}
\begin{proof}
Choose the first column of $B_j$ dependent on its predecessors,
and let its index be $\ell_j$. Since $\rank B_j\le\rho<\kappa$,
this column exists and $0\le\ell_j\le\rho$. If $\ell_j=0$,
the first column is zero, so $\eta_{ij,k}=0$ for every $i$ and
$k<L$; use the zero infinite extensions and $\chi_j=1$.

Suppose $\ell_j\ge1$. The preceding columns are independent,
and there are unique $\gamma_{j,0},\ldots,\gamma_{j,\ell_j-1}\in E$
such that
\begin{equation}\label{eq:finite-recurrence}
 \eta_{ij,u+\ell_j}=\sum_{b<\ell_j}\gamma_{j,b}\eta_{ij,u+b}
 \qquad(0\le i\le n,\ 0\le u<L).
\end{equation}
Set
\[
 \chi_j(Z)=Z^{\ell_j}-\sum_{b<\ell_j}\gamma_{j,b}Z^b.
\]
For each $i$, initialize an infinite sequence with its first
$\ell_j$ moments and extend using $\chi_j$. Induction in
\eqref{eq:finite-recurrence} proves agreement with the supplied
data through index $L+\ell_j-1$, with no assertion beyond that
prefix. This is the same finite-window argument as in
Lemma~\ref{lem:finite-hankel}, now with the coefficients shared
by all $i$.

For $0\le b<\ell_j$, both $b$ and $pb$ lie in that agreement
range. Indeed, $c_p\ge p-1$ and $\ell_j\le\kappa-1$ give
\[
 \begin{aligned}
 p(\ell_j-1)
   &=(p-1)(\ell_j-1)+\ell_j-1\\
   &\le c_p(\kappa-1)+\ell_j-1=L+\ell_j-2,\\
 L+\ell_j-1&\le L+\kappa-2=J_0-1.
 \end{aligned}
\]
Consequently Condition 3 transfers the required Frobenius
identities to each recurrence extension. Choose a common splitting
field $\Omega$ for the finitely many $\chi_j$. Lemma~\ref{lem:atoms}
applies with the same $\chi_j$ for every sequence in column $j$:
it need not be the minimal recurrence of any individual sequence.
It yields \eqref{eq:entry-atoms}, with all active roots among
the distinct roots of $\chi_j$. Their union therefore has size
at most $\deg\chi_j=\ell_j$, proving \eqref{eq:column-root-bound}.
\end{proof}

The quotient-algebra proof of Lemma~\ref{lem:atoms} uses Frobenius
to kill nilpotent contributions and force the surviving weights
into $\F_p$. The nonzero weights equal one only for $p=2$.
The symbol $\beta$ denotes these recurrence roots, which are
initially in $\Omega$; Section~\ref{sec:assignment-recovery}
will prove that they lie in $E$ and equal encodings of assignments.

\subsection{Synchronization of symmetry and assignment encoding}
\label{sec:synchronization}

Put $\mathcal E=\bigcup_{i,j}\mathcal E_{ij}$ and let $A_\beta(i,j)$ be the
coefficient $a_{ij,\beta}$, or zero if that root is absent in $\mathcal E_{ij}$.
Thus $A_\beta\in\F_p^{(n+1)\times(n+1)}$ is nonzero and
\[
 \eta_{ij,k}=\sum_{\beta\in\mathcal E}A_\beta(i,j)\beta^k
       \qquad(0\le k<L).
\]
The joint bound \eqref{eq:column-root-bound} applies within each
column; it does not yet bound the global union by $\rho$.
We first synchronize symmetry and assignment encoding using
their local supports. The resulting matrix structure will then
give the global bound.

For a set $B$ of $s$ distinct roots, the identities
$\sum_{\beta\in B}c_\beta\beta^k=0$ for $0\le k<s$ force
$c_\beta=0$. Indeed the degree-less-than-$s$ Lagrange polynomial
at a root $\beta$ is a linear combination of those powers, and
applying that combination isolates $c_\beta$.
This is the Vandermonde or interpolation isolation rule. The same
independence proves the BCH distance bound in
Proposition~\ref{prop:bch-bound}.

Symmetry compares the moment sequences associated with $(i,j)$
and $(j,i)$, so it involves at most $2\rho$ roots. Assignment
encoding combines $(0,j),(1,j),\ldots,(n,j)$ in one column,
whose joint support has size at most $\rho$. Since $c_p\ge2$,
the two required bounds are
\begin{equation}\label{eq:local-window}
 2\rho\le2(\kappa-1)\le L-1,\qquad
 |\mathcal E_j|\le\rho\le\kappa-1\le L-1.
\end{equation}
\paragraph{Symmetry.}
Fix $i,j$ and let
\[
 \mathcal E_{ij}^{\mathrm{sym}}
    =\mathcal E_{ij}\cup\mathcal E_{ji},\qquad
 s=|\mathcal E_{ij}^{\mathrm{sym}}|.
\]
Condition 6 makes the two table readouts equal, so their
interpolants and moments are equal as well:
$\eta_{ij,k}=\eta_{ji,k}$ for $0\le k<L$. Substituting
\eqref{eq:entry-atoms} and inserting zero coefficients for absent
roots gives
\[
 \sum_{\beta\in\mathcal E_{ij}^{\mathrm{sym}}}
       \bigl(A_\beta(i,j)-A_\beta(j,i)\bigr)\beta^k=0
       \qquad(0\le k<L).
\]
Here $s\le2\rho\le L-1$ by
\eqref{eq:local-window}. If $s>0$, the first $s$ equations,
with $k=0,\ldots,s-1$, form a square Vandermonde system over
$\Omega$. Its columns correspond to distinct roots, so it is
invertible and forces every coefficient difference to be zero.
If $s=0$, both moment expansions have no roots and all the
coefficient differences are already zero. For a root outside
$\mathcal E_{ij}^{\mathrm{sym}}$, both matrix entries are zero
by definition. Since $i,j$ were arbitrary, we obtain
\begin{equation}\label{eq:atom-symmetry}
 A_\beta=A_\beta^\top\qquad(\beta\in\mathcal E).
\end{equation}

\paragraph{Assignment encoding.}
Fix $j$ and put
\[
 \mathcal E_j^{\mathrm{enc}}
    =\bigcup_{i=0}^n\mathcal E_{ij},\qquad
 s_j=|\mathcal E_j^{\mathrm{enc}}|
       \le\ell_j\le\rho\le L-1.
\]
For $0\le k<L-1$, both $k$ and $k+1$ lie in the range
$0,\ldots,L-1$ where the root expansions are valid.
Condition 4, namely \eqref{eq:encoding-moment}, therefore gives
\[
 \begin{aligned}
 0&=\eta_{0j,k+1}-\sum_{i=1}^n\alpha_i\eta_{ij,k}\\
  &=\sum_{\beta\in\mathcal E_j^{\mathrm{enc}}}
       \left(\beta A_\beta(0,j)
                -\sum_{i=1}^n\alpha_i A_\beta(i,j)\right)\beta^k
       \qquad(0\le k<L-1).
 \end{aligned}
\]
The factor $\beta$ in the coefficient comes from
$\beta^{k+1}=\beta\,\beta^k$. If $s_j>0$, take
$k=0,\ldots,s_j-1$. These indices are available because
$s_j-1\le L-2$. The same Vandermonde argument forces each
coefficient in parentheses to vanish. If $s_j=0$ the assertion
is immediate; for roots outside $\mathcal E_j^{\mathrm{enc}}$
all the matrix entries involved are zero. Thus, for every root
in $\mathcal E$, we have
\begin{equation}\label{eq:atom-encoding}
 \beta A_\beta(0,j)=\sum_{i=1}^n\alpha_i A_\beta(i,j)
       \qquad(0\le j\le n).
\end{equation}
These arguments isolate roots within each individual relation;
they do not require a bound on the entire union $\mathcal E$.

\subsection{The common root set and assignment recovery}
\label{sec:assignment-recovery}

Fix a nonzero $A=A_\beta$. If $A(0,j)=0$, then
\eqref{eq:atom-encoding} and $\F_p$-independence of the $\alpha_i$
force every entry in column $j$ to vanish. Some column is nonzero,
and its first entry $r_j=A(0,j)$ is therefore nonzero. Divide
that column by $r_j$ to obtain a vector $(1,x^\top)^\top$ with
$x\in\F_p^n$ and
\[
 \beta=\sum_i\alpha_i x_i.
\]
Injectivity of $\Enc$ makes the normalized vector the same for
every nonzero column. Thus $A=z r^\top$, with $z=(1,x^\top)^\top$
and $r$ its first row written as a column. Symmetry gives
$r_j=A(j,0)=z_jr_0$. If $r_0=0$ the whole matrix is zero,
a contradiction. Consequently
\begin{equation}\label{eq:atom-assignment}
 A_\beta=a_\beta zz^\top,\qquad
 a_\beta=A_\beta(0,0)\in\F_p^\times,\qquad
 \beta=\Enc(x)\in E.
\end{equation}
The vectors $x$ and $z$ were constructed for the fixed root
$\beta$ and may depend on it. To make this dependence explicit,
write $x_\beta=x$ and $z_\beta=z=(1,x_\beta^\top)^\top$.
Thus \eqref{eq:atom-assignment} gives
$A_\beta=a_\beta z_\beta z_\beta^\top$ for every $\beta\in\mathcal E$.
Each active root $\beta$ therefore lies in $E$ and encodes the
vector $x_\beta\in\F_p^n$, although initially it was only known
to lie in $\Omega$. Distinct roots give distinct vectors by
injectivity of $\Enc$. Booleanity and clause satisfaction will
be proved below.

\begin{lemma}[The normalization entry contains every active root]
\label{lem:common-roots}
For the coefficient matrices above,
\[
 \mathcal E=\mathcal E_{00},\qquad |\mathcal E|\le\rho<\kappa.
\]
This conclusion uses symmetry and assignment encoding, without
using any quadratic moment equation.
\end{lemma}
\begin{proof}
For each $\beta\in\mathcal E$, the matrix $A_\beta$ is nonzero by definition.
Equation~\eqref{eq:atom-assignment} gives
$A_\beta(0,0)=a_\beta\ne0$, so this root occurs in the expansion
of $\eta_{00,k}$. By the definition of $A_\beta(0,0)$, this means
$\beta\in\mathcal E_{00}$. The reverse inclusion follows from the
definition of $\mathcal E$ as the union of the root sets for all
pairs $(i,j)$.
Finally \eqref{eq:entry-atoms} gives $|\mathcal E_{00}|\le\rho$.
\end{proof}

We can now impose each quadratic equation on this common root
set. For any $Q_\nu$, Condition 5 gives
\[
 0=\sum_{i,j=0}^n a^{(\nu)}_{ij}\eta_{ij,k}
   =\sum_{\beta\in\mathcal E}
       \left(\sum_{i,j=0}^n a^{(\nu)}_{ij}A_\beta(i,j)\right)\beta^k
       \qquad(0\le k<L).
\]
Since $|\mathcal E|\le\rho<L$, the isolation rule applies to this
root set regardless of how many coefficients of $Q_\nu$ are
nonzero. Hence
\begin{equation}\label{eq:atom-quadratic}
 \sum_{i,j=0}^n a^{(\nu)}_{ij}A_\beta(i,j)=0
       \qquad(\beta\in\mathcal E,\ 1\le\nu\le v).
\end{equation}
For the vector $z_\beta$ associated with a root $\beta$, substitute
$A_\beta(i,j)=a_\beta z_{\beta,i}z_{\beta,j}$ from
\eqref{eq:atom-assignment} into \eqref{eq:atom-quadratic}. This gives
\[
 0=a_\beta\sum_{i,j=0}^n a^{(\nu)}_{ij}z_{\beta,i}z_{\beta,j}
   =a_\beta Q_\nu(z_\beta).
\]
Since $a_\beta\ne0$, we obtain $Q_\nu(z_\beta)=0$ over $\F_p$.
In particular the Booleanity forms give
$x_{\beta,i}(x_{\beta,i}-1)=0$, so $x_{\beta,i}\in\{0,1\}$.
The clause-equation argument of Section~\ref{sec:quadratics}
then recovers a satisfying assignment to the original formula,
over the integers as well as over $\F_p$.

Finally normalization makes $f_{00,\mathbf0}(\sigma)=1$ on $S$.
Its canonical interpolant is the constant one. Therefore
\[
 1=\eta_{00,0}=\sum_{\beta\in\mathcal E}A_\beta(0,0)
             =\sum_{\beta\in\mathcal E}a_\beta
             \qquad\text{in }\F_p.
\]
Hence at least one coefficient $a_\beta$ is nonzero in $\F_p$,
so there is at least one nonzero matrix
$A_\beta=a_\beta z_\beta z_\beta^\top$. The corresponding vector
$x_\beta$ therefore gives a satisfying Boolean assignment by the
preceding argument. Although coefficients can cancel in $\F_p$,
normalization rules out an empty root set. In characteristic two,
every nonzero $a_\beta$ equals one, so the number of roots is odd;
only nonemptiness is needed for the recovery argument.

\begin{remark}[Quadratic sparsity is unnecessary for recovery]
The bound of five nonzero coefficients in the clause forms is
a property of the 3SAT encoding, not a hypothesis of the recovery
argument. Once Lemma~\ref{lem:common-roots} has been established,
the same moment window isolates every quadratic equation over
$\F_p$, including dense ones. The number of its coefficients
affects the description of the affine checks, but does not enter
the rank threshold or soundness bound. Song's proof already
obtains the rank-one coefficient matrices and their nonzero
normalization entries~\cite[Lemma 6.2]{song}, but invokes sparsity
to isolate the quadratic equations. Lemma~\ref{lem:common-roots}
extracts a consequence of that matrix structure and uses it to
remove the sparsity dependence.
\end{remark}

\begin{theorem}[Soundness]\label{thm:soundness}
Under \eqref{eq:budgets}, every feasible table of weight less than
$R_0=\kappa(|S|-D)$ yields a satisfying assignment. Thus an
unsatisfiable source has no feasible table below $R_0$.
\end{theorem}
\begin{proof}
The high-rank estimate forces $\rho<\kappa$. The shared block-column
recurrences, Frobenius expansion, and local isolation proved above
give the nonzero matrices \eqref{eq:atom-assignment}. Each gives
a satisfying Boolean assignment, and normalization ensures one
exists. This completes the soundness part of
Theorem~\ref{thm:table-reduction}.
\end{proof}
This implication applies to every light feasible table; it is not
an assumption about a decoder, positivity, or an initial support
representation. The reduction constructs its instance without
having to execute this existential recovery argument.

\section{From CCP to Euclidean CVP}\label{sec:lattice}

\subsection{A feasible prime-field code instance}
\label{sec:code-instance}

Assemble $H\theta=b$ over $\F_p$ and use elimination to check
consistency. For a consistent system, find a particular solution
$t_0$ and a column basis of $C=\ker H$. Its solution set is
$t_0+C$, so
\[
 W(H,b)=\min_{H\theta=b}\wt(\theta)=\dist_H(t_0,C).
\]
This is the closest codeword instance; $H,b$ also directly give
the syndrome-decoding instance. Finding $t_0$ is easy and does
not require finding a light solution.

The preceding conversion requires a solution $t_0$ of $Ht_0=b$.
If elimination finds that the affine system is inconsistent, no
such target exists. In contrast, an ordinary closest codeword
instance always has a finite distance, since its code contains
zero. We therefore cannot use an inconsistent affine system
directly in this conversion. Completeness guarantees a feasible
table for every satisfiable formula, so inconsistency certifies
that the source formula is a NO instance.

In this case output the zero code $C=\{0\}\subseteq\F_p^N$
and a target $t_0$ with $m=\lceil R_0\rceil$ coordinates equal
to one and the rest zero. This is a valid NO instance because
$\dist_H(t_0,C)=m\ge R_0$. The construction guarantees
$m\le N$. Equivalently, replace the affine system by
$H=I_N$ and $b=t_0$. The replacement preserves the length and
NO threshold and supplies a target for the lattice construction.
It is used only when elimination detects inconsistency; consistent
systems retain the code and target obtained above.

For every NO source handled by the main construction, soundness
(or its same-length replacement above) gives
\[
 W(H,b)=\dist_H(t_0,C)\ge R_0.
\]
This is a Hamming-distance bound. Section~\ref{sec:lattice-from-code} proves that the
squared Euclidean distance of the associated lattice instance
is at least this quantity.

\paragraph{Bounded small instances.}
For fixed $p$ and exponent loss $\epsilon$, any remaining range
$n<n_0$, where $n_0$ is a fixed constant, can be handled by
enumerating the $2^s\le2^n$ assignments to the original variables.
This takes polynomial time in the source description length.
According to the result, return fixed YES or NO instances
satisfying the required CCP and CVP gap promises.
This also handles the lower bound on $n$ in
\eqref{eq:budgets}; it does not alter the asymptotic size formulas.

\subsection{Lattices from linear codes over \texorpdfstring{$\F_p$}{F p}}
\label{sec:lattice-from-code}

Let $C\subseteq\F_p^N$ be any linear code. Define
\[
 \Lambda(C)=\{v\in\Z^N:v\bmod p\in C\}.
\]
This additive subgroup of $\Z^N$ contains $p\Z^N$, so it is
a full-rank lattice. Put $k=\dim C$. After permuting coordinates,
choose a systematic column generator $G=\binom{I_k}{A}$,
with $A\in\F_p^{(N-k)\times k}$, and lift the entries of $A$
to integers in $\{0,\ldots,p-1\}$. An explicit column basis is
\begin{equation}\label{eq:lattice-basis}
 B=\begin{pmatrix}I_k&0\\A&pI_{N-k}\end{pmatrix}.
\end{equation}
Its columns have residues in $C$. Conversely, if $v\bmod p\in C$,
subtracting the first $k$ columns with coefficients equal to its
first $k$ integer coordinates leaves a vector whose upper part is
zero and whose lower coordinates are multiples of $p$.
Thus these columns generate precisely $\Lambda(C)$.
The determinant is $p^{N-k}$ and lattice rank is $N$.
Applying the inverse coordinate permutation to the basis gives
a basis in the original coordinate order.

Now take $C=\ker H$ and the target $t_0$ from
Section~\ref{sec:code-instance}. Lift $t_0$ to an integer vector
$t$ with coordinates in $\{0,\ldots,p-1\}$.
For any $v\in\Lambda(C)$, put $\theta=(t-v)\bmod p$.
Since $t\bmod p=t_0$ and $v\bmod p\in C=\ker H$,
\[
 H\theta=Ht_0-H(v\bmod p)=b.
\]
Thus $\theta$ is a feasible error vector, so
$\wt(\theta)\ge W(H,b)$. If $\theta_j\ne0$, then the
integer $t_j-v_j$ is nonzero and $(t_j-v_j)^2\ge1$.
Consequently,
\[
 \|t-v\|_2^2=\sum_{j=1}^N(t_j-v_j)^2
       \ge\wt(\theta)\ge W(H,b).
\]
Taking the minimum over all $v\in\Lambda(C)$ gives
\begin{equation}\label{eq:lattice-lower}
 \dist_2(t,\Lambda(C))^2\ge W(H,b).
\end{equation}
For a NO source, Section~\ref{sec:code-instance} gives
$W(H,b)\ge R_0$. Thus $\dist_2(t,\Lambda(C))^2\ge R_0$.

In the YES case, Section~\ref{sec:completeness} constructs a
particular feasible error vector with zero-one entries, even though
general feasible vectors have entries in $\F_p$. Indeed, a
satisfying assignment gives $z=(1,x^\top)^\top\in\{0,1\}^{n+1}$;
at each sample the honest table has one block $zz^\top$ and all
other blocks zero. Every table entry is therefore zero or one.
Flatten this table and view its entries as integers to obtain
$e\in\{0,1\}^N$ with $\wt(e)\le R_1$ by
\eqref{eq:yes-weight}. Its residue $\bar e=e\bmod p$ satisfies
$H\bar e=b$. Since $Ht_0=b$, the corresponding codeword is
\[
 w=t_0-\bar e\in C.
\]
Thus $e$ is the error vector, while $w$ is a codeword over
$\F_p$ and need not have zero-one entries. Because
$(t-e)\bmod p=w$, the integer vector $t-e$ belongs to
$\Lambda(C)$. Consequently,
\begin{equation}\label{eq:lattice-yes}
 \dist_2(t,\Lambda(C))^2\le\|e\|_2^2=\wt(e)\le R_1.
\end{equation}
Together with soundness, these inequalities transfer the same
squared-distance thresholds $R_1,R_0$ for every fixed $p$.

For $p=2$ one has the stronger exact identity
\begin{equation}\label{eq:parity-identity}
 \dist_2(t,\Lambda(C))^2=W(H,b),
\end{equation}
because every residue has a minimum representative zero or one.
The identity also holds for $p=3$ using representatives $0,1,-1$.
For $p\ge5$ it need not hold: $C=\{0\}\subset\F_5$, target
$2$, has Hamming distance one and squared distance four to
$5\Z$. The honest zero-one witness, not an incorrect general
metric identity, proves \eqref{eq:lattice-yes}.

\section{Efficiency and gap exponents}
\label{sec:gap}
\label{sec:degree-comparison}

For each fixed prime $p$, all parameter choices can be fixed as
functions of $n$ and the desired exponent loss alone. For a 3SAT instance with $s$ variables,
$m_2$ two-literal clauses, and $m_3$ three-literal clauses,
$n=s+m_3\le s+m$, where $m=m_2+m_3$, is the number of variables
$x_1,\ldots,x_n$ in the quadratic system before homogenization;
$z_0$ is an additional coordinate. The number of quadratic equations
changes the number of affine checks, but not the table length.
The degree and sampling observations below apply to every fixed
$p$. Section~\ref{sec:explicit-epsilon-size} gives explicit binary
lengths for a prescribed gap exponent, and
Section~\ref{sec:general-characteristic} explains the changes
for fixed $p>2$.

\paragraph{Degree and sample savings.}
At a fixed target cutoff $\kappa\ge n+1$, Song's binary
construction uses a radix $d_0$ with $d_0^r=\kappa$ and degree budget
$D_{\mathrm S}=2r(n+2)(d_0-1)(c-1)$~\cite[Definition 3.5]{song}.
For the binary comparison, our window has $L=2\kappa-1$ and
$J_0=3\kappa-2$. Choose the smallest power of two with
$d^r\ge3\kappa-2$ and degree budget $D=2r(d-1)(c-1)$.
At fixed $r$ and growing $n,\kappa$,
\[
 d=\Theta(\kappa^{1/r}),\qquad
 \frac{D}{D_{\mathrm S}}=\Theta(n^{-1})
\]
when both constructions use the same extension degree $c\ge2$;
rounding radices to powers of two changes constants only.
Thus the same rank threshold can be certified at a smaller degree
for every fixed $r\ge1$. The two ingredients are the shared
block-column recurrence and ordinary digits for each Hankel factor.
For any fixed $p$, the analogous capacity is
$d^r\ge(c_p+1)(\kappa-1)+1$, still giving
$d=\Theta_{p,r}(\kappa^{1/r})$.

\paragraph{Which constructions are compared.}
We feed the same sparse binary quadratic system from
Section~\ref{sec:quadratics} to both code constructions. Its forms
have at most five nonzero coefficients, as required by Song's
soundness proof. Song's original preprocessing uses occurrence
variables; the count $n=s+m_3$ here refers to our common input
system, not literally to that preprocessing. In the table below
each construction uses its own cutoff $\kappa$; $q_0,d_0,c_0$
denote Song's coefficient field size, radix, and extension degree.
The present column is binary, so $c_2=2$.
\begin{center}
\small
\begin{tabular}{>{\raggedright\arraybackslash}p{0.19\textwidth}>{\raggedright\arraybackslash}p{0.35\textwidth}>{\raggedright\arraybackslash}p{0.35\textwidth}}
\toprule
Quantity & Song's construction & Present construction\\
\midrule
Rank cutoff & $\kappa=d_0^r$ & $\kappa=1+\lfloor(d^r-1)/3\rfloor$\\
Row range $L$ & $(n+2)(\kappa-1)+1$ & $2\kappa-1$\\
Moment count & $2(n+2)(\kappa-1)+1$ & $3\kappa-2$\\
Digit-product cap & $2(n+2)(d_0-1)$ per axis & $2d-2$ per axis\\
Degree budget & $2r(n+2)(d_0-1)(c_0-1)$ & $D=2r(d-1)(c-1)$\\
Samples & All $q_0$ field elements & Any $4D$ field elements\\
Output length & $(n+1)^2q_0^{r+1}$ & $4D(n+1)^2q^r$\\
\bottomrule
\end{tabular}
\end{center}
Song's coefficient-matrix structure is retained. Our common-root
lemma extracts its implication for dense quadratics, while the
block-column recurrence supplies joint-support control earlier
in the proof.

\paragraph{Sampling and field size.}
The sample set is a second independent economy. For fixed $K,D$,
take $|S|=4D$ whenever possible instead of all $q$ field elements.
Then
\[
 \frac{R_0}{R_1}=\frac{3\kappa}{4(n+1)^2},\qquad
 N=4D(n+1)^2q^r.
\]
Compared with sampling all of $K$, $N$ and $R_1$ are reduced
by $4D/q$. This does not itself increase the gap: its factor
$1-D/|S|$ must still be included. Sampling all $K$ may give a
slightly better constant gap. The reduction allows either choice.
In the least-field schedule below, $4D=\Theta(q)$ for fixed
$p,r,A$, so this sample restriction saves a bounded factor rather
than a further power of $n$. Choosing the least admissible
coefficient field is a separate saving, and its effect is charged
in the full factor $q^r|S|$.

\subsection{Binary codes: smaller lengths at prescribed gaps}
\label{sec:explicit-epsilon-size}

Take $p=2$ for the comparison with Song's construction.

\paragraph{A binary parameter schedule.}

Fix $r$ and a rational $A>0$ with $Ar>2$ before $n$ grows.
Choose $d$ as the least power of two at least $n^A$. Choose
$q=2^a$ as the least power of two satisfying
\[
 a\ge1,\quad c=\max\{2,\lceil n/a\rceil\},\quad
 q\ge8r(d-1)(c-1).
\]
Set
\[
 \begin{aligned}
 D&=2r(d-1)(c-1),\\
 \kappa&=1+\left\lfloor\frac{d^r-1}{3}\right\rfloor,\qquad
 L=2\kappa-1,\qquad J_0=3\kappa-2,
 \end{aligned}
\]
and choose any $4D$ samples from $K$.
The hypotheses \eqref{eq:budgets} hold for large $n$.
The least acceptable field can be found by increasing $a$.
This is a finite parameter search, not a search over assignments.

For these choices,
\begin{equation}\label{eq:efficient-field-orders}
 q=\Theta(dn/\log n),\quad c=\Theta(n/\log n),\quad
 D=\Theta(dn/\log n),\quad |S|=4D=\Theta(q).
\end{equation}
To verify both bounds, a field of order a sufficiently large
constant times $dn$ first shows that the least acceptable
$a=O(\log n)$. Its defining inequality then forces
$q=\Omega(dn/\log n)$. Conversely, a sufficiently large
constant times $dn/\log n$, rounded to a power of two, has
$a=\Theta(\log n)$ and meets that inequality. Thus
$a=\Theta(\log n)$ and all four orders follow.
The shortened window gives $\kappa=\Theta(d^r)$, with no
factor of $n$ lost in converting radix capacity to rank cutoff.

Substituting the full ledger gives
\begin{align}
 N&=\Theta\!\left(
       \frac{n^{A(r+1)+r+3}}{(\log n)^{r+1}}\right),
                                      \label{eq:size-exponent}\\
 R_0/R_1&=\Theta(n^{Ar-2}),
                                      \label{eq:gap-exponent-input}\\
 \delta(r,A)&=\frac{Ar-2}{A(r+1)+r+3}.
                                      \label{eq:delta}
\end{align}
The binary construction has gap $N^{\delta(r,A)+o(1)}$; its Euclidean
exponent is $\delta(r,A)/2$. Every fixed smaller exponent is
obtained for all sufficiently large instances.

\paragraph{Explicit choices for a prescribed code exponent.}
Fix $0<\epsilon<1$ and write $1-\epsilon$ for the desired
closest codeword hardness exponent. Define integers depending
only on $\epsilon$ by
\begin{equation}\label{eq:epsilon-parameters}
\begin{gathered}
 r=\left\lceil\frac2\epsilon\right\rceil,
 \qquad \tau=\epsilon(r+1)-1>0,\\
 A=1+\left\lfloor
       \frac{(1-\epsilon)(r+3)+2}{\tau}
     \right\rfloor,\\
 B=\max\left\{3,\ 1+\left\lfloor
       \frac{(1-\epsilon)(2r+4)+2}{\tau}
     \right\rfloor\right\}.
\end{gathered}
\end{equation}
The integer $A$ specifies our radix growth; $B$ specifies Song's.
The strict surplus provided by these choices absorbs fixed
constants and rounding in both the code and lattice promises.

\begin{proposition}[Explicit binary lengths and lattice dimensions]
\label{prop:epsilon-size}
Take $p=2$ and the parameters in \eqref{eq:epsilon-parameters}.
For all sufficiently large $n$, depending only on $\epsilon$,
Song's construction and the present construction yield code gaps
larger than $N_{\mathrm S}^{1-\epsilon}$ and $N^{1-\epsilon}$,
respectively, with lengths
\begin{equation}\label{eq:epsilon-size}
\begin{aligned}
 N_{\mathrm S}
   &=\Theta_\epsilon(n^{H_{\mathrm S}}),
 &H_{\mathrm S}&=B(r+1)+2r+4,\\
 N&=\Theta_\epsilon\!\left(
          \frac{n^H}{(\log n)^{r+1}}\right),
 &H&=A(r+1)+r+3.
\end{aligned}
\end{equation}
The lattices from the previous section have
exactly the corresponding code lengths as ranks, and Euclidean
gaps larger than $N_{\mathrm S}^{(1-\epsilon)/2}$ and
$N^{(1-\epsilon)/2}$.
\end{proposition}
\begin{proof}
For our construction use the preceding schedule with $p=2$ and
the integer $A$ in \eqref{eq:epsilon-parameters}. Thus $d$ is
the least power of two at least $n^A$, $q$ is the least power
of two meeting the displayed degree budget, and $|S|=4D$.
These are deterministic functions of $n,\epsilon$.
Equations~\eqref{eq:size-exponent}--\eqref{eq:delta} give
\[
 N=\Theta_\epsilon\!\left(
       n^H/(\log n)^{r+1}\right),\qquad
 R_0/R_1=\Theta_\epsilon(n^{Ar-2}).
\]
The choice of $A$ gives
\[
 (Ar-2)-(1-\epsilon)H
 =A\tau-\bigl((1-\epsilon)(r+3)+2\bigr)>0.
\]
In particular $Ar>2$, as required. It follows that
$(R_0/R_1)/N^{1-\epsilon}$ tends to infinity.

For Song's construction follow the parameter choices in
\cite[proof of Theorem 7.1]{song}: take $d_0$ to be the least
power of two at least $n^B$, and take $q_0$ to be the least power
of two satisfying
\[
 q_0\ge12r d_0n^2.
\]
Use $c_0=\lceil n/\log_2q_0\rceil$, $S=\F_{q_0}$, and
$\kappa_0=d_0^r$. For sufficiently large $n$, $c_0\ge2$ and
these choices satisfy Song's field and degree conditions.
Indeed $d_0=\Theta_\epsilon(n^B)$,
$q_0=\Theta_\epsilon(n^{B+2})$, and
\[
 \frac{2r(n+2)(d_0-1)(c_0-1)}{q_0}
 <\frac{n+2}{6n\log_2q_0}=o(1).
\]
Consequently
\[
 N_{\mathrm S}=(n+1)^2q_0^{r+1}
     =\Theta_\epsilon(n^{H_{\mathrm S}}),\qquad
 \frac{R_{0,\mathrm S}}{R_{1,\mathrm S}}
     =\Theta_\epsilon(n^{Br-2}).
\]
The choice of $B$ similarly gives
\[
 (Br-2)-(1-\epsilon)H_{\mathrm S}
 =B\tau-\bigl((1-\epsilon)(2r+4)+2\bigr)>0.
\]
This proves the code gap with strict exponent surplus for both
constructions. On exceptional branches, use the same-length
feasible outputs of Section~\ref{sec:lattice}; thus the lengths
in \eqref{eq:epsilon-size} apply independently of satisfiability.
The lattice rank equals the number of table
coordinates. For the Euclidean promise, take the integer YES
radius $\lceil\sqrt{R_1}\rceil\le2\sqrt{R_1}$. In either
construction the same surplus absorbs this factor of two and
gives the asserted Euclidean gap. The bounded smaller values of
$n$ use the fixed outputs of Section~\ref{sec:code-instance}.
\end{proof}

These formulas make the polynomial bounds on code length and
lattice rank explicit. They do not assert the same exponents
for construction time or dense matrix size: a dense lattice
basis already has $N^2$ entries, and elimination has its own
polynomial cost. As $\epsilon$ tends to zero,
\begin{equation}\label{eq:epsilon-leading-size}
\begin{gathered}
 r=2/\epsilon+O(1),\qquad
 A=2/\epsilon+O(1),\qquad B=4/\epsilon+O(1),\\
 H=4/\epsilon^2+O(1/\epsilon),\qquad
 H_{\mathrm S}=8/\epsilon^2+O(1/\epsilon).
\end{gathered}
\end{equation}
For example, the choices above give the following exponents in
\eqref{eq:epsilon-size}; the logarithmic denominator in our
length is retained there.
\begin{center}
\begin{tabular}{cccccc}
\toprule
$\epsilon$ & $r$ & $A$ & $B$ & $H$ (present) & $H_{\mathrm S}$ (Song)\\
\midrule
$1/2$ & $4$ & $4$ & $6$ & $27$ & $42$\\
$1/4$ & $8$ & $9$ & $14$ & $92$ & $146$\\
$1/10$ & $20$ & $21$ & $38$ & $464$ & $842$\\
\bottomrule
\end{tabular}
\end{center}
The leading estimates follow directly from
$\tau=1+O(\epsilon)$ in \eqref{eq:epsilon-parameters}.
The strict surpluses $(Ar-2)-(1-\epsilon)H$ for the three
rows are $1/2$, $1$, and $2/5$, respectively.
For bounds without $O$-notation in the exponents, $A\le r+1$,
$B\le2r+7$, and $r\le2/\epsilon+1$ give
\[
 H\le\frac4{\epsilon^2}+\frac{10}\epsilon+8,
 \qquad
 H_{\mathrm S}\le\frac8{\epsilon^2}+\frac{30}\epsilon+24.
\]
For the first bound, $\tau\ge1+\epsilon$ and $\epsilon r\ge2$
imply
\[
 (1-\epsilon)(r+3)+2<(r+1)(1+\epsilon)\le(r+1)\tau,
\]
so the integer choice in \eqref{eq:epsilon-parameters} has
$A\le r+1$ and $H\le r^2+3r+4$.
These are explicit feasible choices, rather than optimal choices
over all integers $r,A,B$. The comparison uses Song's prescribed
coefficient field for his construction and the least admissible
coefficient field in our parameter schedule.
For a Euclidean target written instead as $N^{1/2-\delta}$,
substitute $\epsilon=2\delta$. The leading dimension bounds are
then $n^{2/\delta^2+O(1/\delta)}$ for Song's construction and
$n^{1/\delta^2+O(1/\delta)}$ for the present one.

\paragraph{Limiting exponents and the Euclidean gap.}

For each fixed $r$,
\[
 \lim_{A\to\infty}\delta(r,A)=\frac r{r+1},\qquad
 1-\delta(r,A)=\frac{A+r+5}{A(r+1)+r+3}>0.
\]
For a target Hamming exponent $1-\epsilon$, with
$0<\epsilon<1$, choose $r>(1-\epsilon)/\epsilon$ and then
choose a fixed $A$ with
\[
 A>\frac{2+(1-\epsilon)(r+3)}{\epsilon(r+1)-1}.
\]
This gives $\delta(r,A)>1-\epsilon$, using the same loss
parameter $\epsilon$ as in Theorem~\ref{thm:song}.
Constants and integer rounding
are absorbed by that strict exponent surplus. Choosing all
constants before the input grows yields deterministic polynomial
reductions for every fixed loss from exponent one. Letting them
grow with $n$ would not justify the same polynomial-time bound.

The lattice YES radius can be the integer
$R=\lceil\sqrt{R_1}\rceil\le2\sqrt{R_1}$.
For a target Euclidean exponent $1/2-\epsilon/2$, choose Hamming
exponent strictly greater than $1-\epsilon$. Equations
\eqref{eq:lattice-lower}--\eqref{eq:lattice-yes} and the exponent
surplus absorb the factor two and give the desired NO distance.
Thus fixed exponents below one half are approached; the endpoint
is not attained.

\subsection{General prime characteristic: what changes when \texorpdfstring{$p>2$}{p > 2}}
\label{sec:general-characteristic}

Fix a prime $p>2$ independently of the input size. Use the same
$r,A$ as in Section~\ref{sec:explicit-epsilon-size}, but take $d$
to be the least power of $p$ at least $n^A$, and choose the least
$q=p^a$ satisfying
\[
 c=\max\{2,\lceil n/a\rceil\},\qquad
 q\ge8r(d-1)(c-1).
\]
Put $D=2r(d-1)(c-1)$, $|S|=4D$, and
\[
 c_p=\max\{2,p-1\},\qquad
 \kappa=1+\left\lfloor\frac{d^r-1}{c_p+1}\right\rfloor,\qquad
 L=c_p(\kappa-1)+1.
\]
Then $J_0=(c_p+1)(\kappa-1)+1\le d^r$ and
$\kappa=\Theta_p(d^r)$.
For sufficiently large $n$, the standing condition $n\ge p$
and all parameter hypotheses hold.

The field-size argument used in the binary case now gives
\[
 q=\Theta_{p,r,A}\!\left(\frac{dn}{\log_p n}\right),\qquad
 c=\Theta_{p,r,A}\!\left(\frac{n}{\log_p n}\right),\qquad
 D=\Theta_{p,r,A}\!\left(\frac{dn}{\log_p n}\right).
\]
Indeed, rounding a positive real bound up to a power of $p$
changes it by a factor less than $p$, and the least admissible
$a$ is $\Theta_{p,r,A}(\log_p n)$. The same lower and upper
bounds as in \eqref{eq:efficient-field-orders} therefore apply.
Since $p$ is fixed, changing the logarithm base only changes a
constant. Consequently
\[
 N=\Theta_{p,r,A}\!\left(
       \frac{n^{A(r+1)+r+3}}{(\log n)^{r+1}}\right),\qquad
 R_0/R_1=\Theta_{p,r,A}(n^{Ar-2}).
\]
Thus $\delta(r,A)$ in \eqref{eq:delta} is unchanged. In
particular, the same explicit $r,A$ from
\eqref{eq:epsilon-parameters} give the same exponent $H$ in
\eqref{eq:epsilon-size}, with constants depending also on $p$.
There is no improvement in the asymptotic gap exponent merely
from choosing a larger fixed characteristic.

The algebraic and metric differences are summarized below.
\begin{center}
\small
\begin{tabular}{>{\raggedright\arraybackslash}p{0.22\textwidth}>{\raggedright\arraybackslash}p{0.32\textwidth}>{\raggedright\arraybackslash}p{0.35\textwidth}}
\toprule
Quantity & Binary case $p=2$ & General fixed prime $p>2$\\
\midrule
Radix & $d=2^b$ & $d=p^b$\\
Window constant & $c_2=2$ & $c_p=\max\{2,p-1\}$\\
Row range & $L=2\kappa-1$ & $L=c_p(\kappa-1)+1$\\
Frobenius condition & $\eta_{ij,2k}=\eta_{ij,k}^2$
 & $\eta_{ij,pk}=\eta_{ij,k}^p$\\
Nonzero recovered weights & Always $1$ & Elements of $\F_p^\times$\\
Booleanity & Every prime-field element is Boolean
 & The equations $x_i^2-x_i=0$ force Boolean values\\
Honest table entries & $0$ or $1$ & Still $0$ or $1$\\
Squared Euclidean distance & Equals the minimum Hamming weight
 & Equality also holds for $p=3$; the threshold bounds suffice for $p\ge5$\\
\bottomrule
\end{tabular}
\end{center}
The same integer quadratic equations and the same seven families
of affine conditions apply. For $p>2$ the recovery proof must
retain the nonzero prime-field coefficients and the Booleanity
equations explicitly. Section~\ref{sec:lattice} proves that the
honest zero-one witness gives squared distance at most $R_1$
and every NO witness has squared distance at least $R_0$,
without requiring a general metric equality.

The length $N$ counts prime-field symbols, and remains the
lattice rank. Storing an arbitrary table uses
$N\lceil\log_2 p\rceil$ bits; this is a constant factor for
fixed $p$. Field-coordinate equations and lattice entries also
carry their stated $\log p$ bit cost. Allowing $p$ to grow with
$n$ is outside this fixed-characteristic analysis.

\begin{corollary}[The same fixed-loss exponents over each fixed finite field]
\label{cor:prime-field-gaps}
For every fixed finite field $F$ of characteristic $p$ and
$0<\epsilon<1$, the construction gives a deterministic polynomial-time
reduction from 3SAT to CCP over $F$ with gap $N^{1-\epsilon}$.
The lattice associated with
the underlying prime-field instance gives Euclidean gap
$N^{1/2-\epsilon/2}$. The exponents agree with those in the binary case; the constants
in the size bounds may depend on the fixed field $F$.
\end{corollary}
\begin{proof}
Choose fixed $r,A$ with $\delta(r,A)>1-\epsilon$ as in
Section~\ref{sec:explicit-epsilon-size}, and use the binary or
general prime-characteristic schedule as appropriate.
The estimates above give $R_0>N^{1-\epsilon}R_1$ over
$\F_p$ for all sufficiently large $n$; bounded smaller instances
use the fixed outputs of Section~\ref{sec:code-instance}.
To extend the output alphabet from $\F_p$ to $F$, keep
$H,b,t_0$ in the prime subfield and use $C=\ker_F H$.
Choose an $\F_p$-linear map $\pi:F\to\F_p$ with $\pi(1)=1$,
by extending $1$ to a basis. If $H\theta=b$ over $F$, then
$H\pi(\theta)=b$ and $\wt(\pi(\theta))\le\wt(\theta)$.
Conversely, every prime-field solution is a solution over $F$.
Thus the minimum feasible weight, and hence the CCP gap and
length, are unchanged.
The same choice gives the Euclidean gap by the radius-rounding argument.
Completeness uses a zero-one witness for every $p$, so no
odd-characteristic metric-equality assumption is needed.
The field and row counts in the next subsection prove deterministic
polynomial time. Feasible NO replacements preserve both promises.
\end{proof}

\subsection{All coordinates and equations are polynomially explicit}
\label{sec:running-time}

For fixed $p,r,A$, the quantities $d,q,L,J_0,|S|,N$ and
$|\mathcal P|\le(2d-1)^r$ are polynomial in $n$. A field element
in $E$ has $O(ac\log p)=O(n)$ bits. The external field-construction
input is the following fixed-characteristic consequence of
Shoup's theorem~\cite{shoup}: given a represented field
$K=\F_{p^a}$ and a positive integer $c$, a monic irreducible
polynomial of degree $c$ over $K$ can be constructed
deterministically in time polynomial in $a,c$ for fixed $p$.
Shoup reduces this construction to polynomial factorization over
the prime field, which has deterministic polynomial cost when
$p$ is fixed. First apply the same result over $\F_p$ to
construct $K$, and then over the represented $K$ to construct
$E$. The tower bases, arithmetic, and Frobenius matrices are
then explicit. The cost is polynomial in representation lengths;
the exponentially many elements of $E$ are never enumerated.

For transparency, put $C=(n+1)^2$ and $P=|\mathcal P|\le(2d-1)^r$.
Degree constraints use $CP(|S|-D-1)$ equations over $K$.
Comparing each digit vector with one chosen representative of its
index uses at most $CP$ equations over $E$; all pairwise comparisons
are unnecessary. Frobenius, encoding, and quadratic constraints
use respectively $O(CJ_0)$, $O(nJ_0)$, and $O(vJ_0)$ equations
over $E$. Symmetry and normalization add $O(N+|S|)$ prime-field
equations. Expanding multiplies these counts by at most $ac$.
All interpolation and evaluation coefficients are computable by
polynomial-time field arithmetic; coefficient descriptions, dense
affine matrices if used, and elimination cost are polynomial.
The moment matrix in the soundness proof is polynomial size too.
The lattice basis has dimension $N$ and entries of $O(\log p)$
bits, hence $O(N^2\log p)$ bits in dense form. The dimension
exponent $H$ is not asserted to be the running-time exponent.
This accounts for both output coordinates and all checks.

\clearpage
\section{Conclusions and open problems}\label{sec:conclusions}

The mathematics separates into three mechanisms. Frobenius and
an extension basis turn large powers into low-degree products.
Determinant multiplicities turn high polynomial rank into a weight
lower bound. In the other branch, a shared block-column recurrence,
CRT, Frobenius, and local interpolation recover weighted rank-one
assignment matrices. Their nonzero normalization entries identify
a common root set of size at most the rank. This permits recovery
from dense quadratic equations using the same moment window.

The shared recurrence bounds the union of the supports within
each column by the rank, before any coefficient matrices have
been recovered. This justifies the characteristic-dependent
window $L=c_p(\kappa-1)+1$ and capacity
$J_0=(c_p+1)(\kappa-1)+1$. The resulting cutoff is
$\Theta_p(d^r)$ and the input-size gap is $\Theta(n^{Ar-2})$.
Ordinary digits for each Hankel factor reduce the product cap
from $2(n+2)(d_0-1)$ in Song's construction to $2d-2$.
Sampling on any set of size at least $4D$ allows the sample count
to reflect the degree budget. The characteristic-$p$ proof
tracks the prime-field weights and the Boolean equations explicitly.
The final exponent remains the same; smaller fixed constructions
and better finite parameter tradeoffs are the gain.

\subsection{The final gap for CCP}

The hardness factor $N^{1-\epsilon}$ holds for every fixed
$0<\epsilon<1$. Can deterministic polynomial-time reductions from
3SAT establish a gap
\[
 \frac{N}{(\log_2 N)^A}
 \qquad\text{for some fixed }A>1,
\]
and ultimately for every fixed $A>1$?

For every fixed $u\ge1$, Alon--Panigrahy--Yekhanin~\cite[Theorem 1]{apy}
give an $N^{O(u)}$-time algorithm with approximation factor
$\lceil N/(u\log_2 N)\rceil$. Thus, for any constant $c>0$,
choosing a sufficiently large fixed $u$ gives a polynomial-time
algorithm deciding CCP instances with gap $cN/\log_2 N$;
bounded lengths can be handled exactly. Consequently, NP-hardness
at any such gap would imply $\mathrm P=\mathrm{NP}$.
The open question is how close hardness can approach this scale
from below, or whether better approximation algorithms lower it.

\subsection{The final gap for Euclidean CVP}

The corresponding question is whether the fixed power loss
in $N^{1/2-\epsilon/2}$ can be replaced by a logarithmic or
other smaller loss. A binary CCP gap
$N/(\log_2 N)^A$ would, by Section~\ref{sec:lattice}, give a
Euclidean gap of order
\[
 \frac{N^{1/2}}{(\log_2 N)^{A/2}}.
\]
The constant factor needed to round the YES radius does not
change this order. The present fixed-parameter analysis does
not prove such a bound: allowing $\epsilon$ to decrease with
the input size also changes the reduction's running-time
exponent, so it requires a new uniform polynomial-time analysis.

The lattice construction from binary codes has its own restriction.
By \eqref{eq:parity-identity}, applying the CCP algorithm and
lifting its binary error vector gives a polynomial-time
Euclidean approximation factor
$\sqrt{\lceil N/(u\log_2 N)\rceil}$ on these instances,
for every fixed $u\ge1$. Thus this route cannot establish
NP-hardness at any fixed positive multiple of
$\sqrt{N/\log N}$ unless $\mathrm P=\mathrm{NP}$.

For general Euclidean lattices, reaching a fixed multiple
$cN^{1/2}$ is a separate question. The
Aharonov--Regev theorem and Gao's refinement~\cite{ar,gao-ar}
imply that hardness for $c>2/(\pi\sqrt3)$ would force
$\mathrm{NP}=\mathrm{coNP}$. They do not settle hardness at
smaller positive constants, nor do they give a polynomial-time
approximation algorithm at the square-root scale. Determining
the final gap therefore includes both the loss below
$N^{1/2}$ and the role of the leading constant.

\clearpage
\appendix
\section{Computational and mathematical background}\label{app:mathematics-background}

This appendix collects the background for the reduction. It begins
with 3SAT, complexity and gap reductions, linear codes, and lattices,
then develops field theory, polynomial rank bounds, quotient
algebras, interpolation formulas, and recurrence lemmas. It also explains BCH codes and syndrome decoding,
which share the power-sum and recurrence structure of the recovery
proof. The BCH subsection is optional and independent of the
hardness proof; its decoding algorithms are not external inputs
to the reduction. The general results hold in arbitrary characteristic
unless a characteristic-$p$ hypothesis is stated. Binary field
examples illustrate the constructions explicitly.

\subsection{Computational problems: 3SAT, codes, and lattices}\label{sec:computational}

\subsubsection{3SAT, polynomial time, and gap reductions}

A computational problem comes with a finite input description.
Running time is measured in its bit length. For example, a code
can be given by a matrix over a represented finite field, and a
lattice by a rational basis matrix. An algorithm is
\emph{polynomial time} if its running time is bounded by a fixed
polynomial in that description length. Enumerating $2^n$ Boolean
assignments is exponential in $n$, even though writing down one
assignment takes only $n$ bits.

\paragraph{SAT and 3SAT.}
A Boolean variable takes a value in $\{0,1\}$, with zero meaning
false and one meaning true. A Boolean formula is an expression
built from variables using AND ($\wedge$), OR ($\vee$), and
NOT ($\neg$). An assignment \emph{satisfies} the formula if its
value is one. The \emph{Boolean satisfiability problem} (SAT)
asks whether a given formula has a satisfying assignment.

A \emph{literal} is a variable or its negation; a \emph{clause}
is a disjunction of literals. A formula is in \emph{conjunctive
normal form} (CNF) if it is a conjunction of clauses.
The problem \emph{3SAT} is SAT restricted here to CNF formulas
with two or three literals in every clause. Repeated literals
are allowed. This convention includes the NP-complete restriction
with exactly three literals per clause. We retain two-literal
clauses directly, since each gives one quadratic clause equation
without an auxiliary variable.

\paragraph{Complexity classes and reductions.}
The class $\mathrm P$ consists of decision problems solvable in
deterministic polynomial time. A decision problem belongs to
$\mathrm{NP}$ if its YES instances have certificates of polynomial
length that a deterministic algorithm can verify in polynomial
time. More precisely, for an input string $I$, there are a
polynomial $p$ and a polynomial-time verifier $V$ such that
\[
 I\text{ is a YES instance}
 \quad\Longleftrightarrow\quad
 \exists y\in\{0,1\}^{\le p(|I|)}\ \ V(I,y)=1.
\]
The verifier must reject every certificate on a NO instance.
For satisfiability, $y$ is a Boolean assignment and verification
means evaluating the formula.

A \emph{polynomial-time many-one reduction} from a decision
problem $A$ to a decision problem $B$ is a polynomial-time map
$R$ with
\[
 I\in A\quad\Longleftrightarrow\quad R(I)\in B.
\]
An algorithm for $B$ then solves $A$ by constructing $R(I)$ and
running that algorithm once. A problem is \emph{NP-hard} if every
problem in $\mathrm{NP}$ reduces to it, and \emph{NP-complete}
if it is both NP-hard and in $\mathrm{NP}$.
Consequently a polynomial-time algorithm for any NP-complete
problem would imply $\mathrm P=\mathrm{NP}$. These are statements
about algorithms on all inputs; individual instances and special
families can still admit efficient solutions.

We use the following classical theorem and its standard 3SAT
form~\cite{cook,levin}.

\begin{theorem}[Cook--Levin and its 3SAT form]\label{thm:cook-levin}
SAT is NP-complete, and its restriction 3SAT is also NP-complete.
More explicitly, for every decision problem $A\in\mathrm{NP}$,
there is a polynomial-time map taking an input $I$ to a 3SAT
formula $\Psi_I$ such that
\[
 I\in A\quad\Longleftrightarrow\quad
 \Psi_I\text{ is satisfiable}.
\]
\end{theorem}

\paragraph{Algorithms and complexity of SAT.}
For a 3SAT formula with $s$ variables and $m$ clauses, exhaustive
search takes $2^s\poly(s+m)$ time. Faster exponential-time
algorithms are possible: Sch\"oning's randomized local search
with restarts solves 3SAT in $(4/3)^s\poly(s+m)$ time with
bounded one-sided error~\cite{schoening}. It never reports a
satisfying assignment on an unsatisfiable formula; on a
satisfiable formula it can fail to find one. NP-completeness
excludes a polynomial-time algorithm if $\mathrm P\ne\mathrm{NP}$,
but does not by itself exclude subexponential time. The
\emph{Exponential Time Hypothesis} (ETH) of Impagliazzo and
Paturi~\cite{ip} asserts a positive exponential-time lower bound
for 3SAT; in particular, it rules out deterministic
$2^{o(s)}\poly(s+m)$ algorithms. ETH is an additional, unproved
assumption. Song's polynomial-time reductions from 3SAT use
NP-completeness and require no ETH assumption.

\paragraph{A running 3SAT instance and gap reductions.}
Consider
\begin{equation}\label{eq:clause}
 \begin{aligned}
 \Phi&=C_1\wedge C_2\wedge C_3,\\
 C_1&=(x_1\vee\neg x_2\vee x_3),\\
 C_2&=(\neg x_1\vee x_4\vee x_5),\\
 C_3&=(x_2\vee x_3\vee\neg x_5).
 \end{aligned}
\end{equation}
This is a formula with three clauses and five variables, $x_1,\ldots,x_5$. We will use it as the running formula. For the assignment $(1,0,1,0,1)^\top$, the literal values in its clauses are $(1,1,1)$, $(0,0,1)$, and $(0,1,0)$, respectively. Thus all three clauses are true. By contrast, $(0,1,0,0,1)^\top$ makes $C_1$ false and $C_2,C_3$ true, so the entire conjunction is false. A satisfying assignment must make \emph{every} clause true.

A satisfying assignment is the certificate in the definition of $\mathrm{NP}$. Theorem~\ref{thm:cook-levin} makes 3SAT an NP-complete starting problem. The remaining task is to preserve its YES and NO cases through the algebraic and geometric constructions.

A gap problem has a near threshold $R$ and a farther threshold $\gamma R$. An input is promised to belong to either
\[
 \mathrm{YES}:\ d\le R,\qquad
 \mathrm{NO}:\ d>\gamma R.
\]
The algorithm is not constrained on inputs between the two thresholds. A reduction must send satisfiable formulas to YES inputs and unsatisfiable formulas to NO inputs. Establishing both implications is called completeness and soundness, respectively.

\subsubsection{Linear codes, affine systems, and Hamming weight}

Let $q$ be any prime power. A \emph{linear code over $\F_q$} is an $\F_q$-vector subspace $C\le\F_q^N$. Its length is $N$, and its dimension is $k=\dim_{\F_q}C$, the dimension of that subspace rather than the ambient space. It is called an $[N,k]_q$ code and has $q^k$ codewords: each of the $k$ coefficients in a basis expansion has $q$ choices. A binary linear code is the special case $q=2$ used in Song's reduction.

For $x\in\F_q^N$, its Hamming weight $\wt(x)$ is the number of nonzero coordinates, regardless of their nonzero values. Hamming distance is
\[
 d_H(x,y)=|\{i:x_i\ne y_i\}|=\wt(x-y).
\]
In characteristic two, subtraction equals addition, so this is also $\wt(x+y)$. Over a field of odd characteristic that replacement is generally incorrect. For a nonzero code, its minimum distance is
\[
 d_{\min}(C)=\min_{\substack{c,c'\in C\\c\ne c'}}d_H(c,c')
 =\min_{0\ne c\in C}\wt(c).
\]
The equality holds because differences of distinct codewords are precisely the nonzero codewords.

A generator matrix $G\in\F_q^{N\times k}$ has a basis of $C$ as its columns, so $C=\{Gu:u\in\F_q^k\}$. A full-row-rank parity-check matrix $H\in\F_q^{(N-k)\times N}$ describes the same code by
\[
 C=\ker H=\{c\in\F_q^N:Hc=0\},\qquad HG=0.
\]
Such an $H$ can be obtained by taking a basis of the linear equations vanishing on $C$; Gaussian elimination extends a basis of $C$ to a basis of the ambient space and produces $N-k$ independent equations. Conversely, $HG=0$ gives containment of the generated code in the kernel, and their equal dimensions give equality. Redundant parity-check rows are also allowed, as they will be in the constructed affine system.

For example, over $\F_3$,
\[
 G=\begin{pmatrix}1\\1\\1\end{pmatrix},\qquad
 H=\begin{pmatrix}1&-1&0\\0&1&-1\end{pmatrix}
\]
give the $[3,1]_3$ code $C=\{(0,0,0)^\top,(1,1,1)^\top,(2,2,2)^\top\}$. Here $-1=2$ in $\F_3$, and the two equations require all three coordinates to be equal. Both nonzero codewords have weight three, so $d_{\min}(C)=3$.

\paragraph{The minimum distance problem.}
Given a generator matrix for a nonzero linear code, the
\emph{minimum distance problem} asks for $d_{\min}(C)$; its
search version asks for a nonzero codeword attaining that weight.
For a positive integer threshold $w$, the decision version asks
whether there exists $0\ne c\in C$ with $\wt(c)\le w$.
Equivalently, using a parity-check matrix, it asks for
\[
 Hc=0,\qquad c\ne0,\qquad \wt(c)\le w.
\]
The nonzero condition is essential: without it, $c=0$ always
attains weight zero. Although Gaussian elimination finds a basis
of $\ker H$, it does not optimize the weight of its nonzero
linear combinations.

\paragraph{Decoding and the closest codeword problem.}
Given a received vector $t\in\F_q^N$, the \emph{closest codeword
problem} (CCP) asks to find $c\in C$ minimizing $d_H(t,c)$, or to
compute $\dist_H(t,C)$. Its decision version asks whether this
distance is at most a supplied integer radius $w$.
In the communication interpretation, $t=c+e$, where $e$ is an
error vector; recovering a closest codeword means finding a
minimum-weight error compatible with $t$. Indeed
\[
 H(t-c)=Ht\qquad(c\in C),
\]
so closest codeword decoding is exactly the syndrome problem
below with syndrome $b=Ht$. Minimum distance concerns the
lightest nonzero vector in the code; decoding concerns a given
target and its affine coset.

\paragraph{Exact algorithms and approximation hardness.}
Enumerating all $q^k$ codewords solves CCP for an $[N,k]_q$
code in $q^k\poly(N)$ field operations. For unrestricted binary
codes, however, the decision version of syndrome decoding is
NP-complete by Berlekamp, McEliece, and van Tilborg~\cite{bmt};
the equivalence above gives the same conclusion for binary CCP.
Arora, Babai, Stern, and Sweedyk~\cite{abss} proved NP-hardness
of approximation within every fixed constant factor.
Bhattiprolu, Guruswami, and Ren~\cite[Theorem 4]{bgr} gave a
deterministic proof avoiding the PCP theorem, over every fixed
finite field. They also obtain hardness within
$2^{(\log_2 N)^{1-\delta}}$ for every fixed $0<\delta<1$,
assuming that NP has no quasipolynomial-time algorithms
(algorithms with running time exponential in a fixed power of
the logarithm of the input length). These conditional factors
are $N^{o(1)}$. Chapter~7 of the OpenAI report~\cite[Corollary 15]{openai-ten}
subsequently obtained the fixed polynomial factor $N^{1/200}$
under deterministic polynomial-time reductions. Song strengthens
this to $N^{1-\epsilon}$ for every fixed $0<\epsilon<1$
as stated in Theorem~\ref{thm:song}; see also his historical
discussion~\cite[Section 1]{song}.

\paragraph{Polynomial-time approximation.}
Alon, Panigrahy, and Yekhanin~\cite[Theorem 1 and Section 1]{apy}
give a deterministic approximation algorithm for closest codeword
over any finite field $\F_q$. For every fixed constant $c\ge1$,
it uses $N^{O(c)}$ field operations and achieves approximation
factor $\lceil N/(c\log_2 N)\rceil$. More explicitly, given a
linear code $C\le\F_q^N$ and a target $t\in\F_q^N$, the
algorithm returns $\widehat c\in C$ satisfying
\[
 d_H(t,\widehat c)
   \le\left\lceil\frac{N}{c\log_2 N}\right\rceil
            \dist_H(t,C).
\]
Thus closest codeword has a deterministic polynomial-time
$O(N/\log N)$ approximation algorithm over every finite field,
using its represented field arithmetic. The theorem is stated
for binary codes in their paper, whose introduction explicitly
notes the generalization to other finite fields. Its proof uses
spans of coordinate unit vectors and linear-system solving,
which apply over $\F_q$ as well. The constant $c$ is
chosen independently of $N$; increasing it improves the
approximation factor while increasing the running-time exponent.

\paragraph{The guaranteed unique-decoding radius.}
If $d=d_{\min}(C)$ and an integer $e\ge0$ satisfies $2e<d$,
at most one codeword lies within Hamming distance $e$ of a
given target. Otherwise two such codewords $c,c'$ would satisfy
\[
 d_H(c,c')\le d_H(c,t)+d_H(t,c')\le2e<d,
\]
contradicting the definition of $d$. Thus
$\lfloor(d-1)/2\rfloor$ is a guaranteed unique-decoding radius.
Decoding under a promise of that error bound is often called
\emph{bounded-distance decoding}. General closest codeword
decoding places no such restriction on the target, and minimizers
can tie. Uniqueness is an information statement; the argument
above supplies no algorithm for finding the codeword.
BCH codes provide an explicit family with a polynomial-time
algorithm for correcting up to $\lfloor(\delta-1)/2\rfloor$
errors, where $\delta$ is their designed distance, as proved in
Appendix~\ref{sec:bch-background}. That is a
guarantee within a specified radius, rather than an algorithm for
unrestricted closest codeword decoding.
For the complexity of minimum distance and decoding, see
\cite{dms}; the reduction studied here constructs general
closest codeword instances.

For any matrix $H\in\F_q^{s\times N}$ and target syndrome $b\in\F_q^s$, syndrome decoding seeks
\[
 W(H,b)=\min\{\wt(\theta):\theta\in\F_q^N,\ H\theta=b\}.
\]
If the system is inconsistent, this minimum has no feasible vector (one may set its value to $+\infty$). If $Ht=b$ and $C=\ker H$, all solutions are the affine coset $t+C$: indeed $H\theta=Ht$ exactly when $\theta-t\in C$. Since $C=-C$, minimizing over that coset gives
\begin{equation}\label{eq:syndrome-code}
 W(H,b)=\min_{c\in C}\wt(t-c)=\dist_H(t,C).
\end{equation}
This identity holds over every finite field. Gaussian elimination finds one feasible $t$ or detects inconsistency. Finding a feasible vector is easy; finding a minimum-weight feasible vector is the difficult optimization task. The main reduction uses the prime field $\F_p$.

\subsubsection{Lattices, shortest vectors, and closest vectors}

Given independent vectors $b_1,\ldots,b_N\in\R^N$, their lattice is
\[
 \Lambda=\Z b_1+\cdots+\Z b_N
        =\{Bz:z\in\Z^N\},\qquad B=(b_1\ \cdots\ b_N).
\]
The columns of $B$ form a lattice basis. Its rank is $N$, which
is also the ambient dimension in this exposition. For
computational problems, the entries of $B$, the target, and
any supplied radius have finite rational descriptions; the later
reduction outputs an integer basis and an integer target.
Unlike a linear code, a lattice has infinitely many elements,
with integer coefficients instead of finite-field coefficients.

\paragraph{The shortest vector problem (SVP).}
The first minimum of a lattice in the Euclidean norm is
\[
 \lambda_1(\Lambda)=\min_{0\ne v\in\Lambda}\|v\|_2
       =\min_{0\ne z\in\Z^N}\|Bz\|_2,
 \qquad \|v\|_2^2=\sum_i v_i^2.
\]
Exact SVP asks to compute this minimum or find a nonzero vector
attaining it. The decision version, with radius $R>0$, asks
whether $\lambda_1(\Lambda)\le R$. The requirement $v\ne0$
plays the same role as excluding the zero codeword in the
minimum distance problem. A basis vector need not be shortest:
for $b_1=(2,1)^\top$ and $b_2=(3,1)^\top$, the lattice vector
$b_2-b_1=(1,0)^\top$ has norm one, smaller than the respective
basis-vector norms $\sqrt5$ and $\sqrt{10}$.

\paragraph{The closest vector problem (CVP).}
Given a target $t\in\R^N$, define
\[
 \dist_2(t,\Lambda)=\min_{v\in\Lambda}\|t-v\|_2
       =\min_{z\in\Z^N}\|t-Bz\|_2.
\]
Exact CVP asks to compute this distance or find a lattice point
attaining it. Its decision version asks whether the distance
is at most a supplied radius $R$. Zero is an allowed lattice
point here, and a target already in $\Lambda$ has distance zero.
Setting $t=0$ therefore gives a CVP instance with optimum zero;
the nonzero constraint in SVP defines a different optimization
problem. For computational background on these lattice problems,
see Micciancio and Goldwasser~\cite{mg-lattices}; for their
geometry and connections with codes, see Conway and
Sloane~\cite{cs-lattices}.

The minima above are attained. Since $B$ is invertible,
$\|z\|_2\le\|B^{-1}\|\,\|Bz\|_2$, where $\|B^{-1}\|$ is the
operator norm. A bounded ball therefore contains only finitely
many lattice points, because the possible integer columns $z$
are bounded. For SVP choose any nonzero basis vector as an
initial norm bound; for CVP use zero as an initial candidate.
Minimizing over the resulting finite nonempty sets proves the
assertion.

\paragraph{Approximation and gap versions.}
For $\gamma\ge1$, search $\mathrm{SVP}_\gamma$ asks for
$0\ne v\in\Lambda$ with $\|v\|_2\le\gamma\lambda_1(\Lambda)$,
and search $\mathrm{CVP}_\gamma$ asks for $v\in\Lambda$ with
$\|t-v\|_2\le\gamma\dist_2(t,\Lambda)$.
The promise problems $\mathrm{GapSVP}_\gamma$ and
$\mathrm{GapCVP}_\gamma$ distinguish
\[
 \begin{array}{c|cc}
 &\mathrm{YES}&\mathrm{NO}\\\hline
 \mathrm{GapSVP}_\gamma&\lambda_1(\Lambda)\le R&
                         \lambda_1(\Lambda)>\gamma R\\
 \mathrm{GapCVP}_\gamma&\dist_2(t,\Lambda)\le R&
                         \dist_2(t,\Lambda)>\gamma R
 \end{array}
\]
for a supplied $R>0$. As before, inputs in the intervening
region are outside the promise. The approximation factor may
depend on $N$. A search algorithm with the stated approximation
guarantee distinguishes its corresponding gap problem by
comparing the returned norm or distance with $\gamma R$.
Thus hardness for the gap problem also rules out such a search
algorithm under the same complexity assumption.

\paragraph{Recent deterministic hardness for Euclidean SVP.}
Wan~\cite[Theorem 1.8]{wan-svp} proves that Euclidean
$\mathrm{GapSVP}_\gamma$ is NP-hard for every fixed
$1\le\gamma<\sqrt2$ under deterministic polynomial-time
reductions, including exact SVP when $\gamma=1$.
His construction derandomizes a locally dense Reed--Solomon
lattice gadget using estimates for rational points on varieties
over finite fields. In subsequent work,
Wan~\cite[Theorem 1.2]{wan-svp-constant} extends deterministic
polynomial-time many-one hardness to every fixed constant
$\gamma>1$, using tensor amplification of a binary lattice
gadget through the coordinate $\ell_1$ minimum.
These results address the nonzero-vector constraint in SVP
through additional lattice constructions.

\paragraph{Algorithms for Euclidean CVP.}
Babai's nearest-plane algorithm, applied after polynomial-time
lattice basis reduction, gives a deterministic polynomial-time
$2^{N/2}$ approximation~\cite{babai}. Micciancio and
Voulgaris~\cite{mv} obtain a deterministic polynomial-time
approximation factor $2^{O(N\log\log N/\log N)}$.
Their Voronoi-cell method also solves exact Euclidean CVP
deterministically in $2^{O(N)}$ time and space, with polynomial
dependence on the input bit length. Thus exact algorithms
with exponential running time coexist with polynomial-time
algorithms having much larger approximation factors.

\paragraph{Hardness and the square-root scale.}
Arora, Babai, Stern, and Sweedyk~\cite{abss} proved that
Euclidean CVP is NP-hard to approximate within every fixed
constant, which also implies hardness of exact CVP.
Dinur, Kindler, Raz, and Safra~\cite{dkrs} obtained the factor
$N^{a/\log\log N}$ for an absolute constant $a>0$ under
polynomial-time reductions. Its exponent tends to zero.
The OpenAI report~\cite[Chapter 7, Theorem 1]{openai-ten}
obtained the fixed exponent $1/400$, before Song's
$N^{1/2-\epsilon/2}$ result for every fixed $0<\epsilon<1$.
These milestones, reviewed in~\cite[Section 1]{song}, concern
Euclidean CVP and deterministic reductions from ordinary 3SAT.
At the larger factor $cN^{1/2}$, Aharonov--Regev~\cite{ar},
with Gao's refinement~\cite{gao-ar}, place the gap decision
problem in $\mathrm{NP}\cap\mathrm{coNP}$ for every fixed
$c>2/(\pi\sqrt3)$. This supplies short certificates for both
promised cases; it does not supply a polynomial-time
approximation algorithm. NP-hardness at that scale would imply
$\mathrm{NP}=\mathrm{coNP}$.

There is also a lattice bounded-distance decoding problem:
recover a lattice point promised to lie within a radius
$r<\lambda_1(\Lambda)/2$ of $t$. It is unique, since two such
points would have a nonzero lattice difference of length at
most $2r<\lambda_1(\Lambda)$ by the triangle inequality.
The general CVP instances constructed in this paper carry no
promise of this kind.

\subsubsection{The four optimization problems and the reduction}

The code and lattice problems have the following parallel forms.
\begin{center}
\begin{tabular}{lll}
\toprule
Problem & Given data & Quantity to minimize\\
\midrule
Minimum distance & A nonzero linear code $C$ &
 $\wt(c)$ over $0\ne c\in C$\\
Closest codeword (CCP) & $C$ and target $t$ &
 $\wt(t-c)$ over $c\in C$\\
SVP & A lattice $\Lambda$ &
 $\|v\|_2$ over $0\ne v\in\Lambda$\\
CVP & $\Lambda$ and target $t$ &
 $\|t-v\|_2$ over $v\in\Lambda$\\
\bottomrule
\end{tabular}
\end{center}
Membership itself is elementary: $c\in C$ is tested by $Hc=0$;
for an invertible rational lattice basis, $v\in\Lambda$ is tested
by computing $B^{-1}v$ and checking that every coordinate is an
integer. The optimization tasks add the requirement of finding
a vector of minimum weight or distance among the feasible
vectors.

The construction in this exposition starts with an NP-complete
3SAT instance and builds an affine system $H\theta=b$.
Its lightest feasible vector gives a closest codeword problem
through~\eqref{eq:syndrome-code}. The lattice construction in Section~\ref{sec:lattice-from-code} then transfers the YES and NO bounds to squared CVP distance;
for binary codes it gives an exact metric identity. Minimum distance and
SVP provide the homogeneous counterparts for perspective;
their nonzero constraints are treated by separate hardness
reductions~\cite{dms,wan-svp,wan-svp-constant}.
The proofs in Sections~\ref{sec:quadratics}--\ref{sec:gap} establish
the stated closest codeword and CVP gaps.

\subsection{Field extensions, bases, and Frobenius}\label{sec:field-foundation}

\subsubsection{Extensions as vector spaces and algebras}

A \emph{field extension} $E/F$ is an inclusion of fields, with the
operations of $E$ restricting to those of $F$. It is also an $F$-vector
space: scalars come from $F$, while vectors are elements of $E$.
Its degree $[E:F]$ is the dimension of this vector space. A basis
$(b_1,\ldots,b_c)$ means that every $\xi\in E$ has a unique expression
$\xi=\sum_j\gamma_jb_j$, with $\gamma_j\in F$.
The basis elements are field elements; their coordinate columns are
vectors in $F^c$. Multiplication by a fixed $\xi\in E$ is $F$-linear,
since $\xi(\gamma u+\delta v)=\gamma\xi u+\delta\xi v$.
Multiplying two unknown field elements, however, is a bilinear operation.

\paragraph{The polynomial B\'ezout identity.}
Let $a(Z),b(Z)\in F[Z]$ be polynomials, not both zero, and let
$d(Z)=\gcd(a,b)$ denote their monic greatest common divisor.
The Euclidean algorithm gives polynomials $u(Z),v(Z)\in F[Z]$
such that
\[
 u(Z)a(Z)+v(Z)b(Z)=d(Z).
\]
This equation is called a \emph{B\'ezout identity} for $a,b$.
To justify it, repeatedly divide with remainder: the nonzero
remainders have strictly decreasing degrees, and replacing a pair
by the divisor and remainder preserves its common divisors.
The last nonzero remainder, normalized to be monic, is therefore
$d$. Each remainder is a polynomial combination of $a,b$, so
back substitution gives the displayed identity.
In particular, $a,b$ are \emph{relatively prime} when $d=1$,
in which case $ua+vb=1$. Conversely, such an identity forces
every common divisor to divide one, so $\gcd(a,b)=1$.

An element $\zeta$ is \emph{algebraic} over $F$ if a nonzero polynomial
over $F$ vanishes at it. Among monic such polynomials choose one of
least degree, its \emph{minimal polynomial} $m_\zeta$.
Every element of a degree-$c$ extension is algebraic, because its
$c+1$ powers $1,\zeta,\ldots,\zeta^c$ are linearly dependent over $F$.
It is irreducible: a proper factorization would make a factor of
smaller degree vanish at $\zeta$, since a field has no zero divisors.
Division with remainder shows that every polynomial vanishing at
$\zeta$ is a multiple of $m_\zeta$. Consequently evaluation induces
an isomorphism
\[
 F[Z]/(m_\zeta(Z))\longrightarrow F(\zeta),\qquad
 [h(Z)]\longmapsto h(\zeta).
\]
Indeed the quotient is a field by the B\'ezout argument below, so its
image $F[\zeta]$ is already the smallest field containing $F$ and
$\zeta$. If $\deg m_\zeta=c$, its power basis is
$(1,\zeta,\ldots,\zeta^{c-1})$.

For a tower $F\subset K\subset E$ of finite extensions, a basis
$(\beta_i)$ of $K/F$ and a basis $(b_j)$ of $E/K$ give the product
basis $(\beta_i b_j)$ of $E/F$. Spanning follows by expanding each
$K$-coefficient in the $\beta_i$; independence follows by grouping
a relation by the $b_j$ and then by the $\beta_i$. Thus
\[
 [E:F]=[E:K][K:F].
\]
The following finite-field examples instantiate these definitions;
the later construction uses precisely this product-basis convention.

\subsubsection{Prime fields and characteristic}

A field is a set with addition and multiplication in which every nonzero element has a multiplicative inverse. For a prime $p$, the residues modulo $p$ form the field $\F_p$. For example, $2\cdot3=1$ in $\F_5$, so $3$ is the inverse of $2$. Primality matters: modulo $6$, the nonzero residues $2$ and $3$ multiply to zero, which cannot happen in a field.

The characteristic of a field is the least positive integer $p$ for which $p\cdot1=0$, if such an integer exists. In a finite field such an integer exists because the list $0,1,1+1,\ldots$ must repeat. It is prime: a composite characteristic $uv$ would give nonzero elements $u\cdot1,v\cdot1$ whose product is zero. The smallest subfield is a copy of $\F_p$, and the finite field is a finite-dimensional vector space over it. If the dimension is $a$, its size is $p^a$, since a vector has $a$ independently chosen coordinates in $\F_p$.

For every prime $p$, intermediate binomial coefficients vanish modulo $p$,
so $(u+v)^p=u^p+v^p$ and the $p$th-power map is $\F_p$-linear.
Song's original binary case has characteristic two. In every such field,
\[
 1+1=0,\qquad -\alpha=\alpha,\qquad
 (\alpha+\beta)^2=\alpha^2+\beta^2.
\]
The last identity follows by expanding the square: the cross term is $2\alpha\beta=0$. Squaring is therefore linear over $\F_2$. It is generally \emph{not} linear over the larger field: $(\omega\alpha)^2=\omega^2\alpha^2$, which need not equal $\omega\alpha^2$.

\subsubsection{Building a field from an irreducible polynomial}

Let $f(Z)\in\F_p[Z]$ be monic and irreducible of degree $a$. In the quotient
\[
 K=\F_p[Z]/(f(Z)),\qquad \omega=Z\bmod f(Z),
\]
two polynomials represent the same element if their difference is divisible by $f$. Every element has a unique representative of degree less than $a$; hence $1,\omega,\ldots,\omega^{a-1}$ is an $\F_p$-basis. Irreducibility makes the quotient a field: for a nonzero representative $g$ of degree less than $a$, the Euclidean algorithm gives $ug+vf=1$, so the residue of $u$ is its inverse.

\begin{example}[All arithmetic in $\F_4$]\label{ex:f4}
The polynomial $Z^2+Z+1$ has no root in $\F_2$, so it is irreducible. Put
\[
 \F_4=\F_2(\omega),\qquad \omega^2=\omega+1.
\]
Its elements are $0,1,\omega,1+\omega$. For example,
\[
 \omega(1+\omega)=\omega+\omega^2=1,
 \quad (1+\omega)^2=\omega,
 \quad\omega^3=1.
\]
If $a,b\in\F_2$, then
\[
 (a+b\omega)^2=(a+b)+b\omega.
\]
In the ordered basis $(1,\omega)$, squaring has binary matrix
\[
 \begin{pmatrix}1&1\\0&1\end{pmatrix}.
\]
This is a concrete reason an equation involving a square can still be linear in binary coordinates.
\end{example}

\subsubsection{Field extensions and the Frobenius map}\label{sec:frobenius}

Let $q=p^a$ be any prime power, where $p$ is any prime and $a\ge1$. The following identity holds in every finite field $\F_q$, including fields of odd characteristic; it is not restricted to characteristic two:
\begin{equation}\label{eq:finite-field-identity}
 \boxed{\beta^q=\beta\qquad\text{for every }\beta\in\F_q.}
\end{equation}
For $\beta=0$ this is immediate. For $\beta\ne0$, multiplication by $\beta$ permutes the $q-1$ nonzero elements of $\F_q$. Multiplying those elements before and after the permutation gives
\[
 \prod_{\gamma\in\F_q^\times}\gamma
 =\prod_{\gamma\in\F_q^\times}(\beta\gamma)
 =\beta^{q-1}\prod_{\gamma\in\F_q^\times}\gamma.
\]
Canceling the nonzero product yields $\beta^{q-1}=1$, hence $\beta^q=\beta$. Thus the $q$th-power map is the identity on $\F_q$, although it need not be the identity on a larger extension field. This distinction is central to the Frobenius discussion below.

For the main construction take the following tower:
\[
 K=\F_q,\quad q=p^a,\qquad E=K(\zeta),\quad [E:K]=c.
\]
The bracket $[E:K]$ denotes vector-space dimension. If $\zeta$ has an irreducible polynomial of degree $c$ over $K$, every element $\xi\in E$ is uniquely
\begin{equation}\label{eq:coordinates}
 \xi=\gamma_0+\gamma_1\zeta+\cdots+\gamma_{c-1}\zeta^{c-1},
 \qquad \gamma_j\in K.
\end{equation}
If $(\beta_1,\ldots,\beta_a)$ is an $\F_p$-basis of $K$, then
\[
 \{\beta_i\zeta^j:1\le i\le a,\ 0\le j<c\}
\]
is an $\F_p$-basis of $E$ by the product-basis argument above.
It follows that $[E:\F_p]=ac$ and $|E|=p^{ac}$.

A large field does not require an enormous description of each element: one element of $E$ is stored using $ac$ prime-field coordinates. Arithmetic consists of polynomial operations and reduction modulo the defining irreducible polynomials. The construction will enumerate $K$, but never enumerate all of $E$ when $|E|$ is exponential in the input length.

\begin{example}[A sixteen-element field as a two-stage tower]\label{ex:tower}
Start with $K=\F_4$ from Example~\ref{ex:f4}, and put
\[
 E=K(\zeta),\qquad \zeta^2+\zeta+\omega=0.
\]
For $u=0,1,\omega,1+\omega$, the values of $u^2+u$ are $0,0,1,1$. None is $\omega$, so the defining quadratic has no root in $K$ and is irreducible. Thus $E$ has sixteen elements, with binary basis
\[
 (\alpha_1,\alpha_2,\alpha_3,\alpha_4)
 =(1,\omega,\zeta,\omega\zeta).
\]
For $\xi=1+\zeta$,
\[
 \xi^2=1+\omega+\zeta,\qquad
 \xi^3=1+(1+\omega)\zeta.
\]
For instance, expand $(1+\zeta)(1+\omega+\zeta)$ and replace $\zeta^2$ by $\zeta+\omega$ to obtain the second equality.
\end{example}

\paragraph{Absolute Frobenius: raising to the characteristic.}
Let $F=\F_{p^m}$, where $p$ is any prime. The \emph{absolute Frobenius map} is
\[
 \varphi:F\longrightarrow F,\qquad \varphi(\xi)=\xi^p.
\]
It preserves addition and multiplication:
\[
 (\xi+\upsilon)^p=\xi^p+\upsilon^p,\qquad
 (\xi\upsilon)^p=\xi^p\upsilon^p,\qquad \varphi(1)=1.
\]
For the first equality, each intermediate binomial coefficient $\binom pi$, $0<i<p$, is divisible by $p$, so all the cross terms vanish. If $\varphi(\xi)=\varphi(\upsilon)$, then $(\xi-\upsilon)^p=0$, which forces $\xi=\upsilon$ in a field. Thus $\varphi$ is injective; finiteness makes it surjective as well. A bijective map preserving the field operations is a \emph{field automorphism}.

Every element of the prime field $\F_p$ is fixed, since it is a sum of copies of $1$. Consequently $\varphi$ is $\F_p$-linear:
\[
 \varphi(a\xi+b\upsilon)=a\varphi(\xi)+b\varphi(\upsilon)
 \qquad(a,b\in\F_p).
\]
It is generally not linear over a larger field of scalars. Instead,
$\varphi(\gamma\xi)=\gamma^p\varphi(\xi)$; the scalar is also raised to the $p$th power. For example, on $\F_4$, squaring sends $\omega\cdot1$ to $\omega^2=1+\omega$, not to $\omega\cdot1^2$. Its matrix in a prime-field basis nevertheless describes an ordinary linear transformation over $\F_p$, as in Example~\ref{ex:f4}.

\paragraph{Iteration, inverse, and fixed elements.}
Repeated application gives $\varphi^j(\xi)=\xi^{p^j}$. Applying~\eqref{eq:finite-field-identity} with $q=p^m$ gives $\xi^{p^m}=\xi$ for every $\xi\in F$. Hence
\[
 \varphi^m=\operatorname{id},\qquad
 \varphi^{-1}(\xi)=\xi^{p^{m-1}}.
\]
The order is exactly $m$. If $0<j<m$ and $\varphi^j$ were the identity, the nonzero polynomial $Z^{p^j}-Z$ would have all $p^m$ elements as roots, exceeding its degree $p^j$. Here we use the elementary root bound: factoring out $Z-\xi$ for each distinct root shows that their number cannot exceed the degree.

The fixed elements of $\varphi$ are exactly $\F_p$. Indeed they are the roots of $Z^p-Z$, whose degree is $p$; all $p$ elements of $\F_p$ are already roots. In particular, in characteristic two,
\begin{equation}\label{eq:frobenius-fixed-bits}
 \xi^2=\xi\quad\Longleftrightarrow\quad\xi\in\{0,1\}.
\end{equation}
Every element also has a unique $p$th root, given by the inverse formula. Thus taking $p$th roots inside a finite field requires no further field extension.

\paragraph{Relative Frobenius: fixing the base field.}
For an extension $E=\F_{q^c}$ of $K=\F_q$, where more generally $q=p^a$, define
\[
 \sigma_q:E\longrightarrow E,\qquad \sigma_q(\xi)=\xi^q.
\]
This is the $a$th iterate of absolute Frobenius. It fixes every element of $K$ and is $K$-linear, since
\[
 (\gamma\xi+\delta\upsilon)^q
 =\gamma\xi^q+\delta\upsilon^q\qquad(\gamma,\delta\in K).
\]
Its order is exactly $c$, and its inverse sends $\xi$ to $\xi^{q^{c-1}}$, by the same argument just given. Its fixed elements in $E$ are precisely $K$: the polynomial $Z^q-Z$ has degree $q$ and already has all $q$ elements of $K$ as roots. In our characteristic-two tower, distinguish the two maps carefully:
\[
 \begin{array}{c|c|c}
 \text{map on }E & \text{formula} & \text{guaranteed scalar linearity}\\\hline
 \text{absolute Frobenius} & \xi\mapsto\xi^2 & \F_2\text{-linear}\\
 \text{relative Frobenius over }K & \xi\mapsto\xi^{2^a}=\xi^q & K\text{-linear}.
 \end{array}
\]

\paragraph{Conjugates and minimal polynomials.}
The distinct elements in the relative Frobenius orbit of $\xi$ are
$\xi,\xi^q,\ldots,\xi^{q^{\ell-1}}$, where $\ell$ is the least positive integer for which $\xi^{q^\ell}=\xi$. Since $\sigma_q^c$ is the identity, the orbit length $\ell$ divides $c$: divide $c$ by $\ell$ and use minimality to see that the remainder is zero. The polynomial
\[
 m_\xi(Z)=\prod_{j=0}^{\ell-1}(Z-\xi^{q^j})
\]
has coefficients in $K$, because applying $\sigma_q$ to its coefficients permutes its factors and therefore fixes those coefficients. It is the minimal polynomial of $\xi$ over $K$. To justify this last assertion, any polynomial over $K$ vanishing at $\xi$ also vanishes at every $\xi^{q^j}$, by applying Frobenius to its equation. Such a nonzero polynomial must have degree at least $\ell$, whereas $m_\xi$ has degree $\ell$. Moreover, a proper factorization over $K$ would give a smaller-degree factor vanishing at $\xi$, which is impossible. Thus the minimal polynomial is irreducible and its roots are these distinct conjugates.

\paragraph{Trace and norm.}
Two useful ways to combine the conjugates are
\[
 \operatorname{Tr}_{E/K}(\xi)=\sum_{j=0}^{c-1}\xi^{q^j},\qquad
 \operatorname{N}_{E/K}(\xi)=\prod_{j=0}^{c-1}\xi^{q^j}.
\]
Relative Frobenius cyclically permutes the terms or factors, so both values lie in its fixed field $K$. The trace is $K$-linear, by the displayed linearity of each Frobenius iterate; the norm is multiplicative, by taking the product of $(\xi\upsilon)^{q^j}=\xi^{q^j}\upsilon^{q^j}$.

\begin{example}[Both Frobenius maps in the sixteen-element tower]\label{ex:frobenius}
In Example~\ref{ex:tower}, $K=\F_4$, $E=\F_{16}$, and $\zeta^2=\zeta+\omega$. Absolute Frobenius is squaring, with the orbit
\[
 \zeta\ \longmapsto\ \zeta+\omega\ \longmapsto\ \zeta+1
 \ \longmapsto\ \zeta+\omega+1\ \longmapsto\ \zeta.
\]
Each arrow follows by squaring and using $\omega^2=\omega+1$. Relative Frobenius over $K$ takes every second step: $\zeta\mapsto\zeta^4=\zeta+1\mapsto\zeta$. For $\gamma,\delta\in K$ the two formulas are
\[
 \begin{aligned}
 (\gamma+\delta\zeta)^2
   &=\gamma^2+\delta^2\omega+\delta^2\zeta,\\
 (\gamma+\delta\zeta)^4
   &=(\gamma+\delta)+\delta\zeta.
 \end{aligned}
\]
The first squares the coefficients, while the second is $K$-linear. The two relative conjugates of $\zeta$ give
\[
 \operatorname{Tr}_{E/K}(\zeta)=1,\qquad
 \operatorname{N}_{E/K}(\zeta)=\zeta(\zeta+1)=\omega,
\]
and their product polynomial is $(Z-\zeta)(Z-(\zeta+1))=Z^2+Z+\omega$.
\end{example}

\paragraph{Why Frobenius matters in the reduction.}
Two uses recur in the proof. Powers $\xi^{p^j}$ are linear in
prime-field coordinates, so the Frobenius tests remain affine.
The identity $\gamma^p=\gamma$ forces coefficients into $\F_p$.
Only for $p=2$ does a nonzero coefficient then have to equal one.
For odd $p$, separate Boolean equations are required to recover
bits; coefficient normalization in the main proof handles
weights that need not equal one.

\subsection{Rank deficiency at a polynomial specialization}\label{sec:rank-deficiency}

The following determinant-multiplicity bound, used in
Song's rank argument~\cite[Lemma 4.2]{song}, controls rank loss
in the high-rank branch. We give a proof over an arbitrary field.

\begin{lemma}[Rank deficit is bounded by zero multiplicity]
\label{lem:multiplicity}
Let $C(X)\in F[X]^{\rho\times\rho}$ have nonzero determinant
$\Delta(X)$, and let $\sigma\in F$. Put
$\delta_\sigma=\rho-\rank C(\sigma)$. Then
\[
 (X-\sigma)^{\delta_\sigma}\mid\Delta(X),\qquad
 \delta_\sigma\le\operatorname{ord}_{X=\sigma}\Delta.
\]
The order denotes the multiplicity of the zero, with order zero
when $\Delta(\sigma)\ne0$.
\end{lemma}
\begin{proof}
Row reduction of the constant matrix $C(\sigma)$ gives an
invertible constant matrix $P\in F^{\rho\times\rho}$ whose
last $\delta_\sigma$ rows in $PC(\sigma)$ are zero.
Apply $P$ to $C(X)$. Every entry of each of these rows is
divisible by $X-\sigma$. Factoring this polynomial out of the
$\delta_\sigma$ rows proves that it appears to at least that
power in $\det(PC(X))$. Since $\det P\in F^\times$, the same
divisibility holds for $\Delta$.
\end{proof}

\begin{proposition}[Total rank over finitely many samples]
\label{prop:sample-rank-foundation}
Let $U(X)$ be a rectangular polynomial matrix whose entries have
degree at most $D$, and let $\rho$ be its rank over $F(X)$. For any finite set $S\subset F$,
\[
 \sum_{\sigma\in S}\rank U(\sigma)
       \ge\rho(|S|-D).
\]
\end{proposition}
\begin{proof}
If $\rho=0$, the inequality is immediate. Otherwise choose a
$\rho\times\rho$ submatrix $C$ with nonzero determinant $\Delta$.
Each determinant term has degree at most $\rho D$, so
$\deg\Delta\le\rho D$. The factors $X-\sigma$ for distinct
$\sigma$ are coprime; Lemma~\ref{lem:multiplicity} therefore gives
$\sum_{\sigma\in S}\delta_\sigma\le\deg\Delta\le\rho D$.
Since $\rank U(\sigma)\ge\rank C(\sigma)=\rho-\delta_\sigma$,
summing proves the claim.
\end{proof}

The zero multiplicity can exceed the rank deficit. For example,
$C=\operatorname{diag}(X^h,1)$ loses one rank at zero but its
determinant has multiplicity $h$. The lemma bounds how much
rank can be lost; it does not assert equality. In this part of
the proof multiplicity measures vanishing of a determinant in
$F[X]$. In the recurrence part multiplicity describes local
nilpotent coordinates of $\Omega[Z]/(\chi)$. These are distinct
uses of the same elementary polynomial notion.

\subsection{Polynomial quotient algebras, CRT, and interpolation}\label{sec:crt-foundation}

\subsubsection{The quotient $F[X]/(f)$ for a general polynomial}

Throughout this section $F$ is a field and $f\in F[X]$ is monic of
degree $N\ge1$. Set $A=F[X]/(f)$ and $x=[X]\in A$.
A bracket denotes a residue class. Division by $f$ proves that
\[
 1,x,\ldots,x^{N-1}
\]
is an $F$-basis, whether or not $f$ is irreducible. Thus $\dim_F A=N$.
An $F$-\emph{algebra} is an $F$-vector space with an associative
multiplication and identity compatible with scalar multiplication;
$A$ is a commutative example. A \emph{unit} is an invertible element,
an \emph{idempotent} satisfies $e^2=e$, and a \emph{nilpotent}
satisfies $u^h=0$ for some $h\ge1$.

\begin{proposition}[Units and the field case]\label{prop:quotient-units}
The class $[h]$ is a unit in $A$ if and only if $\gcd(h,f)=1$.
The algebra $A$ is a field if and only if $f$ is irreducible.
\end{proposition}
\begin{proof}
If $uh+vf=1$, then $[u][h]=1$. Conversely, an inverse gives
$uh\equiv1\pmod f$, hence such a B\'ezout identity and coprimality.
If $f$ is irreducible, every nonzero representative of degree less
than $N$ is coprime to $f$ and is therefore invertible. If $f=ab$
with $0<\deg a,\deg b<N$, the nonzero classes $[a]$ and $[b]$
multiply to zero, so $A$ is not a field.
\end{proof}

\subsubsection{Chinese remainder decomposition and coordinate idempotents}

\begin{theorem}[Polynomial Chinese remainder theorem]\label{thm:polynomial-crt}
Suppose $f=f_1\cdots f_s$ with monic, nonconstant, pairwise coprime
polynomials $f_i\in F[X]$. The residue map
\[
 \Phi:F[X]/(f)\longrightarrow\prod_{i=1}^s F[X]/(f_i),\qquad
 [h]\longmapsto
 \begin{pmatrix}
 {[h]_{f_1}}\\
 \vdots\\
 {[h]_{f_s}}
 \end{pmatrix}
\]
is an isomorphism of $F$-algebras. Write $M_i=f/f_i$ and choose
$u_i,v_i$ with $u_iM_i+v_if_i=1$. The classes
\begin{equation}\label{eq:crt-idempotents}
 e_i=[u_iM_i]\in A
\end{equation}
map to the coordinate unit vectors: one in component $i$ and zero
in every other component. They satisfy
\[
 e_i^2=e_i,\qquad e_ie_j=0\ (i\ne j),\qquad \sum_i e_i=1.
\]
The inverse of $\Phi$ is given by
\[
 \Phi^{-1}\!\left(
 \begin{pmatrix}
 {[h_1]_{f_1}}\\
 \vdots\\
 {[h_s]_{f_s}}
 \end{pmatrix}\right)
 =\sum_{i=1}^s e_i[h_i].
\]
\end{theorem}
\begin{proof}
The map preserves sums, products, scalars, and the identity. A
polynomial in its kernel is divisible by all $f_i$, hence by their
product: if $a,b$ are coprime and divide $h$, write $h=ak$;
B\'ezout for $a,b$ shows $b\mid k$, so $ab\mid h$, and iterate.
This proves injectivity. The displayed B\'ezout identities give the
stated residues for $e_i$, so the proposed inverse attains every
component column. Changing $h_i$ by a multiple of $f_i$ changes its contribution
by a multiple of $f$. The idempotent identities hold in each
component, and injectivity carries them back to $A$.
\end{proof}

\paragraph{Notation for the components.}
For any polynomial $h\in F[X]$, write
\[
 [h]_i=h\bmod f_i\in F[X]/(f_i),\qquad 1\le i\le s.
\]
We view $[h]_i$ as an element of the $i$th component of the CRT
product. Its unique polynomial representative of degree less than
$\deg f_i$ is the remainder on dividing $h$ by $f_i$.
Thus the subscript specifies the quotient in which the residue
class is taken, and the CRT map can be written as
\[
 \Phi([h])=\begin{pmatrix}{[h]_1}\\\vdots\\{[h]_s}\end{pmatrix}.
\]

The phrase \emph{coordinate unit basis} refers here to unit
\emph{vectors} of a product, not to invertible elements of its ring.
For $s>1$, $e_i$ has zero components and is not a unit. Indeed an
invertible idempotent must equal one, by multiplying $e_i^2=e_i$
by its inverse. Also the $e_i$ alone are an $F$-basis only when
each component has $F$-dimension one. In general $e_iA$ is the
whole $i$th component and has dimension $\deg f_i$.

Factor $f$ over $F$ as
\begin{equation}\label{eq:irreducible-crt}
 f=\prod_{i=1}^s p_i^{m_i},\qquad
 A\simeq\prod_i F[X]/(p_i^{m_i}),
\end{equation}
where the $p_i$ are distinct monic irreducibles. Different powers
are coprime, even if some $m_i>1$. Each factor has residue field
$F[X]/(p_i)$ and nilpotent ideal generated by $[p_i]$.
It has a nonzero nilpotent when $m_i>1$, since $[p_i]^{m_i}=0$
and $[p_i]\ne0$. Thus $A$ has no nonzero nilpotents exactly when
all $m_i=1$; in that case it is a product of fields.

\subsubsection{Simple roots and Lagrange interpolation}\label{sec:lagrange-interpolation}

Assume first that $f=\prod_{i=1}^N(X-a_i)$ with distinct
$a_i\in F$. Every factor quotient is $F$, by evaluation at $a_i$,
so CRT becomes
\begin{equation}\label{eq:simple-evaluation}
 A\longrightarrow F^N,\qquad [h]\longmapsto
 \begin{pmatrix}h(a_1)\\\vdots\\h(a_N)\end{pmatrix}.
\end{equation}
These evaluations are well defined because $f(a_i)=0$.
The coordinate idempotents have the degree-less-than-$N$
representatives
\begin{equation}\label{eq:lagrange-basis}
 \ell_i(X)=\prod_{j\ne i}\frac{X-a_j}{a_i-a_j}
          =\frac{f(X)}{(X-a_i)f'(a_i)}.
\end{equation}
Every denominator is nonzero. The second equality follows from
$f'(a_i)=\prod_{j\ne i}(a_i-a_j)$, obtained by differentiating the
product and evaluating at $a_i$. In particular,
$\ell_i(a_j)=\delta_{ij}$, so $e_i=[\ell_i]$.

\begin{proposition}[Interpolation as a basis expansion]\label{prop:lagrange}
For prescribed $b_1,\ldots,b_N\in F$, the unique polynomial of
degree less than $N$ taking value $b_i$ at $a_i$ is
\[
 h(X)=\sum_{i=1}^N b_i\ell_i(X).
\]
For any polynomial $h$, the same formula using its values gives
its remainder modulo $f$.
\end{proposition}
\begin{proof}
The formula has the required degree and evaluations. The difference
of two such interpolants has $N$ distinct roots and degree less
than $N$, so it is zero by the root bound. Subtracting the formula
from any $h$ gives a polynomial divisible by every $X-a_i$, hence
by $f$, which proves the remainder assertion.
\end{proof}

If $f$ has simple roots only in a splitting field $\Omega/F$, the
same evaluation isomorphism holds for $\Omega[X]/(f)$ over $\Omega$.
It need not hold as a product of copies of $F$ before extending
scalars: an irreducible polynomial of degree greater than one
defines a field over $F$, even though it splits into linear factors
over $\Omega$. A polynomial has only simple roots in a splitting
field if and only if $\gcd(f,f')=1$: at a root $a$, writing
$f=(X-a)^m h$ with $h(a)\ne0$ shows $f'(a)\ne0$ exactly when $m=1$.
Over an arbitrary field this condition is stronger than merely
having no repeated irreducible factor over $F$. Over a finite field
each irreducible polynomial is separable: if its derivative were
zero, all exponents would be multiples of the characteristic $p$;
surjectivity of Frobenius would express it as a $p$th power,
contradicting irreducibility. A nonzero derivative is coprime to an
irreducible polynomial by degree. This verifies the distinction
in the setting of the reduction.

\paragraph{Special case: interpolation on a finite field.}

Apply the simple-root construction to $f=X^q-X$ over
$K=\F_q$, the full-field special case of sampling in the reduction. Every element of $K$ is a root, and
$f'=-1\ne0$, so
\[
 K[X]/(X^q-X)\simeq\prod_{\sigma\in K}K.
\]
It represents polynomial functions on $K$. The later recurrence
algebra uses a different polynomial, which can have repeated roots.

Since $f'(X)=-1$, the Lagrange polynomial associated with a sample
$\sigma\in K$ has the particularly simple form
\begin{equation}\label{eq:finite-field-lagrange}
 \ell_\sigma(X)
 =\frac{f(X)}{(X-\sigma)f'(\sigma)}
 =-\frac{X^q-X}{X-\sigma}
 =1-(X-\sigma)^{q-1}.
\end{equation}
Here the quotient is a polynomial, since $X-\sigma$ divides
$X^q-X$. The last equality follows from
$X^q-X=(X-\sigma)^q-(X-\sigma)$, using $\sigma^q=\sigma$.
Thus $\ell_\sigma(\sigma)=1$, while for $\tau\ne\sigma$ in $K$,
$\ell_\sigma(\tau)=1-(\tau-\sigma)^{q-1}=0$.
In characteristic two, $-1=1$, so $f'(X)=1$ and the same
formula reads $\ell_\sigma(X)=(X^q-X)/(X-\sigma)
=1+(X-\sigma)^{q-1}$.

For prescribed values $b_\sigma\in K$ at all $q$ points
$\sigma\in K$, Proposition~\ref{prop:lagrange} gives the unique
polynomial $F\in K[X]$ of degree less than $q$ taking those values:
\[
 F(X)=\sum_{\sigma\in K}b_\sigma\ell_\sigma(X)
     =\sum_{\sigma\in K}b_\sigma\bigl(1-(X-\sigma)^{q-1}\bigr).
\]
This formula also shows that interpolation is linear in the
prescribed values.

\paragraph{Polynomial representatives and polynomial functions.}
A polynomial of degree $\ge q$ need not equal that canonical interpolant even if it has the same values on $K$. For example, $X^q-X$ vanishes on $K$. At an element $\zeta$ in a proper extension of $K$, it need not vanish. For a general sample set $S$, the analogous vanishing polynomial is
$g_S=\prod_{\sigma\in S}(X-\sigma)$. The main construction therefore
keeps product degrees below $|S|$ before evaluating at $\zeta$;
formula~\eqref{eq:sample-lagrange} is the corresponding general interpolation basis.

\begin{example}[Small interpolation]\label{ex:interpolation}
On $K=\F_4$, the values of
\[
 F(X)=X^2+\omega X+(1+\omega)
\]
at $0,1,\omega,1+\omega$ are, respectively,
$1+\omega,0,1+\omega,0$. Its degree is two, less than four, so these four values determine it uniquely.
\end{example}

\subsubsection{Repeated roots, Taylor coordinates, and a full basis}

Now assume that $f$ splits over $F$ as
\[
 f(X)=\prod_{i=1}^s(X-a_i)^{m_i},\qquad
 a_i\ne a_j\ (i\ne j),\qquad \sum_i m_i=N.
\]
The same discussion applies over a splitting field when necessary.
Use the component notation introduced after
Theorem~\ref{thm:polynomial-crt}, now with $f_i=(X-a_i)^{m_i}$.
Work directly with the local factors
\[
 R_i=F[X]/((X-a_i)^{m_i}),
 \qquad [h]_i=h(X)\bmod (X-a_i)^{m_i}.
\]
Every element of $R_i$ has a unique Taylor representative
\[
 \sum_{k=0}^{m_i-1}c_{i,k}(X-a_i)^k,
 \qquad c_{i,k}\in F.
\]
Indeed, expanding a polynomial in powers of $X-a_i$ and discarding
terms of degree at least $m_i$ gives existence. For uniqueness,
a difference of two such representatives has degree less than
$m_i$ and is divisible by $(X-a_i)^{m_i}$, so it is zero.
Thus $1,[X-a_i]_i,\ldots,[(X-a_i)^{m_i-1}]_i$ is an $F$-basis
of $R_i$. The element $[X-a_i]_i$ has $m_i$th power zero; it is
nonzero when $m_i>1$. When $m_i=1$, each class is determined
just by its constant value at $a_i$, identifying $R_i$ with $F$.

Writing elements of the product as columns, the CRT map is
\begin{equation}\label{eq:local-root-algebras}
 \Phi:\ A\xrightarrow{\ \sim\ }\prod_{i=1}^s R_i,
 \qquad [h(X)]\longmapsto
 \begin{pmatrix}
 {[h]_1}\\
 \vdots\\
 {[h]_s}
 \end{pmatrix}.
\end{equation}
The $i$th entry belongs to $R_i$; addition and multiplication of
these columns are componentwise. In particular, the same
polynomial $h(X)$ is reduced modulo a different power of $X-a_i$
in each entry. The image of the CRT idempotent $e_i\in A$ has
one in its $i$th entry and zero in every other entry.

For a given polynomial $h$, its Taylor coefficients at $a$ are
the \emph{Hasse derivatives}, defined by the polynomial identity
\[
 h(X)=\sum_{k\ge0}h^{[k]}(a)(X-a)^k.
\]
The sum is finite. For $h(X)=\sum_r d_rX^r$, the binomial theorem
gives $h^{[k]}(a)=\sum_{r\ge k}d_r\binom{r}{k}a^{r-k}$.
Consequently the local Taylor coordinates of $[h]$ are
$c_{i,k}=h^{[k]}(a_i)$, and the full CRT image is
\begin{equation}\label{eq:hasse-local}
 \Phi([h])=
 \begin{pmatrix}
 \displaystyle\left[\sum_{k=0}^{m_1-1}h^{[k]}(a_1)(X-a_1)^k\right]_1\\[8pt]
 \vdots\\[4pt]
 \displaystyle\left[\sum_{k=0}^{m_s-1}h^{[k]}(a_s)(X-a_s)^k\right]_s
 \end{pmatrix}.
\end{equation}
Divisibility by $(X-a_i)^{m_i}$ is equivalent to the vanishing of
these first $m_i$ Taylor coefficients. Ordinary derivatives give
$h^{(k)}(a)=k!h^{[k]}(a)$, but $k!$ may vanish in positive
characteristic, so ordinary derivatives alone can lose coordinates.

To define the CRT idempotents as in~\eqref{eq:crt-idempotents}, set
\[
 M_i(X)=\frac{f(X)}{(X-a_i)^{m_i}}
       =\prod_{j\ne i}(X-a_j)^{m_j}.
\]
Since $M_i(a_i)\ne0$, the polynomials $M_i(X)$ and
$(X-a_i)^{m_i}$ are relatively prime. The Euclidean algorithm
therefore gives polynomials $u_i(X),v_i(X)\in F[X]$ such that
\[
 u_i(X)M_i(X)+v_i(X)(X-a_i)^{m_i}=1.
\]
As in the general Chinese remainder theorem, define
\[
 E_i(X)=u_i(X)M_i(X),\qquad e_i=[E_i]\in A.
\]
The B\'ezout identity and the factors of $M_i$ give
\[
 E_i\equiv1\pmod{(X-a_i)^{m_i}},\qquad
 E_i\equiv0\pmod{(X-a_j)^{m_j}}\quad(j\ne i),
\]
so $\Phi(e_i)$ is the $i$th coordinate unit vector.
When all roots are simple, we may choose $u_i(X)=1/M_i(a_i)$;
then $E_i(X)=M_i(X)/M_i(a_i)=\ell_i(X)$ from Section~\ref{sec:lagrange-interpolation}.

We choose $u_i(X)$ with $\deg u_i<m_i$, by reducing it modulo
$(X-a_i)^{m_i}$ and adjusting $v_i(X)$ in the identity above.
Since $\deg M_i=\deg f-m_i$, this gives $\deg E_i<\deg f$.

\begin{proposition}[CRT basis and Taylor expansion]
\label{prop:crt-taylor-basis}
Let $f=\prod_{i=1}^s(X-a_i)^{m_i}$, where the $a_i\in F$ are
distinct and $m_i\ge1$. Put $N=\sum_i m_i$, $A=F[X]/(f)$,
and $x=[X]$, and let $e_i$ be the CRT idempotents defined above.
Define $N$ polynomials of degree less than $N$ by
\begin{equation}\label{eq:confluent-basis}
 b_{i,j}(X)=\operatorname{rem}_{f(X)}\bigl(E_i(X)(X-a_i)^j\bigr)
 \qquad(1\le i\le s,\ 0\le j<m_i).
\end{equation}
Their residue classes satisfy
\[
 [b_{i,j}]=e_i(x-a_i)^j.
\]
The classes $[b_{i,j}]$ form an $F$-basis of $A$, and the
polynomials $b_{i,j}$ form an $F$-basis of the polynomials in
$F[X]$ of degree less than $N$. The following expansion identities
hold in $A=F[X]/(f)$ and in $F[X]$, respectively:
\begin{equation}\label{eq:hermite-expansion}
 \begin{aligned}
 \text{(a)}\quad [h]&=\sum_{i=1}^s\sum_{j=0}^{m_i-1}
 h^{[j]}(a_i)[b_{i,j}]
 &&\bigl(h\in F[X]\bigr),\\[4pt]
 \text{(b)}\quad h&=\sum_{i=1}^s\sum_{j=0}^{m_i-1}
 h^{[j]}(a_i)b_{i,j}
 &&\bigl(h\in F[X],\ \deg h<N\bigr).
 \end{aligned}
\end{equation}
\end{proposition}
\begin{proof}
Under the CRT isomorphism, $[b_{i,j}]$ has entry $[(X-a_i)^j]_i$
in component $i$ and zero elsewhere. For each $i$, these local
powers form an $F$-basis of $R_i$. Their coordinate vectors
therefore form an $F$-basis of the product $\prod_i R_i$, so
the classes $[b_{i,j}]$ form an $F$-basis of $A$.
The quotient map restricts to an $F$-linear isomorphism from the
polynomials of degree less than $N$ onto $A$, so the $b_{i,j}$
form a basis of that polynomial space.
By~\eqref{eq:hasse-local}, the Taylor coordinates of $[h]$ in
component $i$ are $h^{[j]}(a_i)$ for $0\le j<m_i$.
Thus both sides of~\eqref{eq:hermite-expansion}\textup{(a)} have
the same CRT image, and injectivity gives their equality.
If $\deg h<N$, both polynomial representatives in
\eqref{eq:hermite-expansion}\textup{(b)} have degree less than
$N$ and represent this same class, so they are equal in $F[X]$.
\end{proof}

\begin{corollary}[Hermite interpolation at nodes with multiplicities]
\label{prop:hermite-interpolation}
Let $a_1,\ldots,a_s\in F$ be distinct, let $m_i\ge1$, and put
$N=\sum_i m_i$. Given arbitrary scalars
\[
 c_{i,k}\in F\qquad(1\le i\le s,\ 0\le k<m_i),
\]
there is a unique polynomial $h\in F[X]$ of degree less than $N$
whose Taylor coefficients satisfy
\[
 h^{[k]}(a_i)=c_{i,k}\qquad(1\le i\le s,\ 0\le k<m_i).
\]
Equivalently, $h$ satisfies the simultaneous congruences
\[
 h(X)\equiv\sum_{k=0}^{m_i-1}c_{i,k}(X-a_i)^k
       \pmod{(X-a_i)^{m_i}}\qquad(1\le i\le s).
\]
The polynomial $h(X)$ is given by the expansion in
Proposition~\ref{prop:crt-taylor-basis},
equation~\eqref{eq:hermite-expansion}\textup{(b)}:
\[
 h(X)=\sum_{i=1}^s\sum_{k=0}^{m_i-1}c_{i,k}b_{i,k}(X).
\]
\end{corollary}
\begin{proof}
Proposition~\ref{prop:crt-taylor-basis} identifies the coordinates
in the basis $b_{i,k}$ with the Taylor coefficients at the nodes.
Thus the displayed sum has the prescribed coefficients and is
unique among polynomials of degree less than $N$.
\end{proof}

\paragraph{Why Hermite, or confluent, interpolation?}
Lagrange interpolation prescribes one value at each distinct
node. Here the node $a_i$ carries $m_i$ conditions: the value
$h(a_i)=c_{i,0}$ and the next $m_i-1$ Taylor coefficients.
This is called \emph{Hermite interpolation}; the word
\emph{confluent} indicates that several interpolation conditions
are collected at the same node, which is counted with multiplicity
$m_i$. When the relevant factorials are invertible in $F$, these
conditions can equivalently be written in the familiar derivative
form $h^{(k)}(a_i)=k!c_{i,k}$. The Hasse-derivative formulation
states the problem in every characteristic, including when some
factorials vanish. If all $m_i=1$, only values are prescribed,
and the result is Lagrange interpolation.

For example, take distinct $a,b\in F$ with multiplicities two
and one. Prescribing
\[
 h(a)=\alpha,\qquad h^{[1]}(a)=\beta,\qquad h(b)=\gamma
\]
determines the degree-less-than-three polynomial
\[
 h(X)=\alpha+\beta(X-a)
 +\frac{\gamma-\alpha-\beta(b-a)}{(b-a)^2}(X-a)^2.
\]
Its first two Taylor coefficients at $a$ are $\alpha,\beta$;
substitution of $X=b$ gives $\gamma$. This is the simplest case
combining a repeated node with a second node.

In particular $E_i(a_j)=\delta_{ij}$ even with repeated roots,
but these values alone do not determine $e_i$ in $A$. The full
congruences above specify all its local Taylor coordinates.

A local element $u=[c_0+(X-a_i)v(X)]_i$ is invertible exactly when
$c_0\ne0$. For nonzero $c_0$, put
$w=c_0^{-1}[(X-a_i)v(X)]_i$; then $w^{m_i}=0$, and the finite
geometric identity gives
\[
 u^{-1}=c_0^{-1}\sum_{j=0}^{m_i-1}(-w)^j.
\]
When $c_0=0$, the element is nilpotent and cannot be invertible.
Thus $[h]\in A$ is a unit exactly when $h(a_i)\ne0$ for every
root, consistent with Proposition~\ref{prop:quotient-units}.

\begin{lemma}[Idempotents are sums of CRT coordinate idempotents]
\label{lem:crt-idempotent-classification}
Let $f=\prod_{i=1}^s(X-a_i)^{m_i}$ with distinct $a_i\in F$ and
$m_i\ge1$, and let $e_i\in A=F[X]/(f)$ be the CRT idempotent
whose image has one in component $i$ and zero elsewhere.
Every idempotent $e\in A$ has a unique expression
\[
 e=\sum_{i\in S}e_i\qquad\text{for a subset }S\subseteq\{1,\ldots,s\}.
\]
Conversely every such subset sum is idempotent. Thus $A$ has
exactly $2^s$ idempotents, and each prescribed set of zero--one
values at the distinct roots has a unique idempotent lift in $A$.
The empty subset gives zero, and the full subset gives one.
\end{lemma}
\begin{proof}
First consider an idempotent $u\in R_i=F[X]/((X-a_i)^{m_i})$.
Its constant Taylor coefficient $c=u(a_i)$ satisfies $c^2=c$
in the field $F$, so $c$ is zero or one. If $c=0$, then $u-1$
has nonzero constant coefficient and is a unit by the criterion
above; multiplying $u(u-1)=0$ by its inverse gives $u=0$.
If $c=1$, then $u$ is a unit, and the same identity gives $u=1$.
Hence the local idempotents are precisely zero and one.

The CRT isomorphism preserves products, so an element of $A$
is idempotent exactly when each component is idempotent. Its
image therefore has one precisely on a subset $S$ of the
components, and this is the image of $\sum_{i\in S}e_i$.
Injectivity of CRT gives the asserted equality and uniqueness.
Conversely the image of any such sum has only zero--one
components, so its square equals itself. These distinct columns
also prove the count and the assertion about root values.
\end{proof}

\subsubsection{The nilradical and root values}

The \emph{nilradical} $\mathcal N$ of $A=F[X]/(f)$ is the ideal
of its nilpotent elements. For the factorization
$f=\prod_i p_i^{m_i}$ into distinct monic irreducibles, put
$f_{\rm rad}=\prod_i p_i$. Then
\begin{equation}\label{eq:nilradical-quotient}
 \mathcal N=(f_{\rm rad})=f_{\rm rad}A,\qquad
 A/\mathcal N\simeq F[X]/(f_{\rm rad}).
\end{equation}
Here $f_{\rm rad}$ also denotes its residue class in $A$, and
$(f_{\rm rad})$ is the principal ideal in $A$. Indeed, $[h]$ is
nilpotent if and only if $f\mid h^k$ for some $k$, equivalently
$p_i\mid h$ for every $i$. The reverse implication follows by
choosing $k\ge\max_i m_i$. This also proves that these classes
form the stated ideal.

When $f=\prod_i(X-a_i)^{m_i}$ splits over $F$, evaluation gives
\[
 A\longrightarrow F^s,\qquad [h]\longmapsto
 \begin{pmatrix}h(a_1)\\\vdots\\h(a_s)\end{pmatrix}.
\]
The images of the $e_i$ show that this map is onto, and its
kernel is $\mathcal N$. Thus the \emph{reduced algebra}
$A/\mathcal N\simeq F^s$ records precisely the root values.
The kernel has basis $[b_{i,j}]$ with $j\ge1$, so
$\dim_F\mathcal N=\deg f-s$.
The elements $a_ie_i$ take value $a_i$ at $a_i$ and zero at the
other roots. With $x=[X]$, write
\begin{equation}\label{eq:root-nilpotent-decomposition}
 x_{\rm red}=\sum_i a_ie_i,\qquad \nu=x-x_{\rm red}.
\end{equation}
In component $i$, $\nu=[X-a_i]_i$, so
$\nu^{\max_i m_i}=0$. Root values therefore discard this
nilpotent part of $x$.

\begin{example}[One double root and one simple root]\label{ex:repeated-crt}
In $A=\F_2[X]/(X^2(X+1))$, put $x=[X]$. The CRT idempotents are
$e_0=1+x^2$ and $e_1=x^2$, as reduction modulo $X^2$ and $X+1$
verifies. Here
\[
 \nu=x+x^2\ne0,\qquad \nu^2=0,\qquad
 \mathcal N=\F_2\nu,\qquad x_{\rm red}=e_1.
\]
Indeed $x^3=x^2$, hence $x^4=x^2$ and $\nu^2=0$.
The representative $X+X^2$ has degree less than three and is
nonzero, although it vanishes at both roots.
\end{example}

\subsection{Recurrences, Hankel matrices, and algebraic moments}
\label{sec:recurrence-foundation}

\subsubsection{Recurrences and the shift operator}

Let $F$ be a field and let
\[
 \mu=(\mu_0,\mu_1,\ldots,\mu_N)^\top\in F^{N+1},\qquad N\ge0,
\]
be a finite sequence. Here $N$ denotes the last index, so the
sequence has $N+1$ terms. For $1\le \ell\le N+1$, the sequence
satisfies a \emph{monic recurrence of order $\ell$} if there exist
coefficients $\gamma_0,\ldots,\gamma_{\ell-1}\in F$ such that
\begin{equation}\label{eq:general-recurrence}
 \mu_{k+\ell}=\sum_{b=0}^{\ell-1}\gamma_b\mu_{k+b}
 \qquad(0\le k\le N-\ell).
\end{equation}
The first $\ell$ values determine all later supplied values by
induction. When $\ell=N+1$, the range $0\le k\le N-\ell=-1$ is empty,
so~\eqref{eq:general-recurrence} holds vacuously.
The recurrence polynomial is
\[
 \chi(Z)=Z^\ell-\sum_{b=0}^{\ell-1}\gamma_bZ^b.
\]
Write $(S\mu)_k=\mu_{k+1}$ for the shift on the available
indices. Equation~\eqref{eq:general-recurrence} then says
$\chi(S)\mu=0$ for $0\le k\le N-\ell$. Multiplying $\chi$ by
another polynomial gives an annihilating relation on the indices
where all its terms are available. We assign order zero only
to the zero sequence.

For an infinite sequence $(\mu_k)_{k\ge0}$, the same definitions
apply, with the recurrence required for every $k\ge0$.

The order in the definition need not be the smallest possible.
The \emph{least recurrence order} is the smallest $\ell$ for which
such coefficients exist, with order zero reserved for the zero
sequence. For finite data, this minimum depends on all the
supplied terms, and can increase when more terms are appended.

\paragraph{Example: $00100001$.}
Work over $\F_2$ and interpret this string as
\[
 \mu=(\mu_0,\ldots,\mu_7)^\top=(0,0,1,0,0,0,0,1)^\top.
\]
No recurrence of order $1\le \ell\le4$ is possible: at $k=7-\ell$,
equation~\eqref{eq:general-recurrence} would express $\mu_7=1$
as a linear combination of the preceding $\ell$ terms, all of
which are zero. Order zero is also excluded, since the sequence
is nonzero.
On the other hand,
\[
 \mu_{k+5}=\mu_k\qquad(0\le k\le2),
\]
so the least order is $5$, with recurrence polynomial $Z^5-1$.
The seven-term prefix $0010000$ has least order $3$: it satisfies
$\mu_{k+3}=0$ for $0\le k\le3$, while an order $\ell=1$ or $2$
would express $\mu_2=1$ using preceding zeros at $k=2-\ell$.
Appending the final $1$ therefore raises the least order from
$3$ to $5$.

\paragraph{Example: $00100000$.}
Now take
\[
 \mu=(\mu_0,\ldots,\mu_7)^\top=(0,0,1,0,0,0,0,0)^\top.
\]
The same first-nonzero-term argument rules out orders below $3$,
and
\[
 \mu_{k+3}=0\qquad(0\le k\le4)
\]
gives order $3$, with recurrence polynomial $Z^3$.
Thus appending $0$ to the same seven-term prefix preserves its
least order $3$, whereas appending $1$ raises it to $5$.
Continuing with zeros indefinitely also preserves least order
$3$. These two examples illustrate why a recurrence fitted to
a finite prefix need not remain valid when another term is
supplied.

\paragraph{Example: $00000001$.}
Take
\[
 \mu=(\mu_0,\ldots,\mu_7)^\top=(0,0,0,0,0,0,0,1)^\top.
\]
For every $1\le \ell\le7$, the recurrence equation at $k=7-\ell$
would express $\mu_7=1$ as a linear combination of preceding
zeros, which is impossible. The sequence is nonzero, so order
zero is also excluded. Order $8$ is allowed: since $N=7$, there
is no index in the range $0\le k\le N-8$, and hence no recurrence
identity to check within these eight terms. Its least order is
therefore $8$. For example, choosing all recurrence coefficients
zero gives polynomial $Z^8$ and continues the sequence with zeros.
The seven-term prefix $0000000$ has order zero, so appending a
single $1$ raises the least order from $0$ to $8$.

\begin{example}[Powers in a finite field extension]\label{ex:extension-power-recurrence}
Let $E/F$ be a field extension of finite degree $\ell\ge1$, and let
$\eta\in E$. Then the sequence $(\eta^k)_{k\ge0}$ satisfies a
recurrence of order at most $\ell$ with coefficients in $F$.
More precisely, if the minimal polynomial of $\eta$ over $F$ has
degree $d\le \ell$, there exist $\gamma_0,\ldots,\gamma_{d-1}\in F$
such that
\[
 \eta^{k+d}=\sum_{b=0}^{d-1}\gamma_b\eta^{k+b}
 \qquad(k\ge0).
\]

Indeed, since $\dim_F E=\ell$, the $\ell+1$ elements
$1,\eta,\ldots,\eta^\ell$ are linearly dependent over $F$.
Thus the minimal polynomial $m_\eta(Z)\in F[Z]$ has degree
$d\le \ell$. Write
\[
 m_\eta(Z)=Z^d-\sum_{b=0}^{d-1}\gamma_bZ^b,
 \qquad\gamma_b\in F.
\]
Evaluating at $\eta$ and multiplying by $\eta^k$ gives the
recurrence of order $d\le \ell$ for every $k\ge0$.
\end{example}

\subsubsection{Hankel rank and the finite window}

For any integers $L,B\ge1$, put $N=L+B-2$. For any sequence
$(\eta_i)_{i=0}^N\in F^{N+1}$, define its rectangular Hankel matrix
\[
 H_{L,B}(\eta)=(\eta_{a+b})_{0\le a<L,\ 0\le b<B}.
\]
The entries are constant on each antidiagonal, and their indices
range from zero through $N=L+B-2$.

\begin{samepage}
\begin{lemma}[Recurrences and Hankel rank]
\label{lem:finite-hankel}
Let $N=L+B-2$ and let $H_{L,B}$ be formed from
$(\eta_i)_{i=0}^N$ as above.
\begin{enumerate}[label=(\alph*),leftmargin=*]
\item Suppose the sequence satisfies a recurrence of order $\ell$. Then
\[
 \rank H_{L,B}\le \ell.
\]
As before, order zero means that the sequence is zero.

\item Suppose $\rank H_{L,B}=\rho<B$. Let column $\ell$ be the
first column dependent on its predecessors. Then $0\le \ell\le\rho$.
If $\ell=0$, $\eta_a=0$ for $0\le a<L$. If $\ell\ge1$, then there
are unique $\gamma_0,\ldots,\gamma_{\ell-1}\in F$ satisfying
\[
 \eta_{a+\ell}=\sum_{b<\ell}\gamma_b\eta_{a+b}\qquad(0\le a<L).
\]
\end{enumerate}
\end{lemma}
\end{samepage}
\begin{proof}
\textup{(a)} If $\ell=0$, the matrix is zero. If $\ell\ge B$, its rank
is at most $B\le \ell$. Otherwise let $C_b$ denote column $b$.
The recurrence applied in every row gives
\[
 C_{b+\ell}=\sum_{j=0}^{\ell-1}\gamma_j C_{b+j}
 \qquad(0\le b<B-\ell).
\]
In row $a$, the shift index is $a+b\le L+B-\ell-2=N-\ell$, so
every use of the recurrence lies within~\eqref{eq:general-recurrence}.
Induction on the column index shows that all columns lie in the
span of the first $\ell$ columns, proving the rank bound.

\textup{(b)} The first $\rho+1$ columns are dependent, so the first dependence
occurs by $\rho$. Its preceding columns are independent, which
gives uniqueness of the coefficients. Equality of the $L$ row
entries gives exactly the displayed recurrence window. If $\ell=0$,
the first column is zero, giving $\eta_a=0$ for $0\le a<L$.
\end{proof}

\begin{example}[Two simple Hankel matrices]
\label{ex:last-nonzero-hankel}
Take $L\ge1$, $B\ge2$, and $N=L+B-2$, with
\[
 (\eta_0,\ldots,\eta_N)^\top=(0,\ldots,0,1)^\top.
\]
Only the bottom-right entry of $H_{L,B}$ is one, so its rank
is one. Its first column is zero, however, and
Lemma~\ref{lem:finite-hankel}\textup{(b)} selects $\ell=0$.
Thus $\eta_a=0$ for $0\le a<L$, as asserted, even though
$\eta_N=1$. For $L=B=3$,
this is the sequence $(0,0,0,0,1)^\top$ with
\[
 H_{3,3}=\begin{pmatrix}0&0&0\\0&0&0\\0&0&1\end{pmatrix}.
\]
Deleting the last column gives $H_{L,B-1}=0$, formed from the
prefix $(\eta_0,\ldots,\eta_{N-1})$.

For another simple example, keep $L=B=3$ and take the sequence
$00100$, that is,
\[
 (\eta_0,\ldots,\eta_4)^\top=(0,0,1,0,0)^\top.
\]
Then
\[
 H_{3,3}=\begin{pmatrix}0&0&1\\0&1&0\\1&0&0\end{pmatrix}.
\]
Its determinant is $-1\ne0$, so its rank is $3$.
The least recurrence order of this sequence is also $3$:
$\eta_{k+3}=0$ holds for $k=0,1$, whereas any recurrence of
order $\ell=1$ or $2$, applied at $k=2-\ell$, would express
$\eta_2=1$ using preceding zeros. Order zero is excluded because
the sequence is nonzero.
\end{example}

A recurrence obtained from a finite Hankel matrix need not hold
beyond its stated window.
For an infinite recurrent sequence, however, its \emph{infinite
Hankel rank}, defined as $\sup_{L,B}\rank H_{L,B}$, equals its
least recurrence order. To prove the lower bound, the first $\ell$
infinite columns must be independent when $\ell$ is least; otherwise
their first dependence would give a shorter recurrence valid at
all rows. Independence of finitely many infinite columns is
witnessed by finitely many rows. Indeed the kernels of their
restrictions to the first $L$ rows form a decreasing sequence of
subspaces of $F^\ell$; after finitely many strict dimension drops
they stabilize, and their intersection is zero. Some finite
restriction therefore has rank $\ell$. The upper bound follows
by applying Lemma~\ref{lem:finite-hankel}\textup{(a)} to every
finite Hankel matrix.

\subsubsection{Recurrence sequences, root expansions, and Frobenius}
\label{sec:recurrence-root-expansions}

Our goal is to prove the general root expansion in
Lemma~\ref{lem:general-root-expansion} and its Frobenius refinement
in Lemma~\ref{lem:atoms}. Both apply to any chosen monic recurrence
polynomial $\chi$, including the minimal recurrence polynomial.

Let $(\eta_k)_{k\ge0}$ be any infinite sequence in $F$, and let
$\Omega/F$ be any extension field. Define an $\Omega$-linear map
$G:\Omega[Z]\longrightarrow\Omega$ by
\[
 G\left(\sum_{k=0}^d c_kZ^k\right)=\sum_{k=0}^d c_k\eta_k
 \qquad(c_k\in\Omega).
\]
This is the unique $\Omega$-linear map satisfying $G(Z^k)=\eta_k$
for every $k\ge0$, because the monomials form a basis of $\Omega[Z]$.
The scalars $\eta_k=G(Z^k)$ are called the \emph{algebraic moments}
of $G$ with respect to the monomials $1,Z,Z^2,\ldots$; in
particular, the zeroth moment is $\eta_0=G(1)$. No measure or
positivity is required for this terminology.

Now suppose that the sequence satisfies a monic recurrence of
order $\ell\ge0$, with recurrence polynomial
\[
 \chi(Z)=Z^\ell-\sum_{b=0}^{\ell-1}\gamma_bZ^b\in F[Z],
 \qquad
 \eta_{k+\ell}=\sum_{b=0}^{\ell-1}\gamma_b\eta_{k+b}\quad(k\ge0).
\]
For $\ell=0$, this means $\chi=1$ and the sequence is zero; the
sums are empty. No minimality assumption on $\chi$ is made.
For every $a\ge0$, the recurrence gives
\[
 G(Z^a\chi)=\eta_{a+\ell}-\sum_{b=0}^{\ell-1}\gamma_b\eta_{a+b}=0.
\]
By linearity, $G$ vanishes on the ideal $(\chi)$ and therefore descends to an
$\Omega$-linear functional
\begin{equation}\label{eq:moment-functional}
 g:\mathcal A=\Omega[Z]/(\chi)\longrightarrow\Omega,
 \qquad g(z^k)=\eta_k,\quad z=[Z].
\end{equation}
Thus the same scalars $\eta_k$ are the algebraic moments of $g$
with respect to the powers $1,z,z^2,\ldots$.
Conversely, any such linear functional defines a sequence in
$\Omega$ obeying
the recurrence, because $z^\ell=\sum_{b<\ell}\gamma_bz^b$ in $\mathcal A$.
Usually $g$ is not a ring homomorphism, and $g(1)$ need not equal one.

The Hankel entries are values of the bilinear pairing
$B_g(u,v)=g(uv)$: specifically $\eta_{a+b}=B_g(z^a,z^b)$.
This connects sequence shifts with multiplication in a
finite-dimensional algebra.

\begin{lemma}[Root expansion for an arbitrary recurrence polynomial]
\label{lem:general-root-expansion}
Let $(\eta_k)_{k\ge0}$ be a sequence in $F$ satisfying a monic
recurrence with polynomial $\chi\in F[Z]$ of degree $\ell\ge0$.
Assume $\Omega/F$ contains all roots of $\chi$, and write
\[
 \chi(Z)=\prod_{i=1}^s(Z-\beta_i)^{m_i},
 \qquad \beta_i\ne\beta_{i'}\text{ for }i\ne i'.
\]
Then there are coefficients $c_{i,j}\in\Omega$, for
$1\le i\le s$ and $0\le j<m_i$, such that
\begin{equation}\label{eq:generalized-moments}
 \eta_k=\sum_{i=1}^s\sum_{j=0}^{\min(k,m_i-1)}
      c_{i,j}\binom{k}{j}\beta_i^{k-j}\qquad(k\ge0).
\end{equation}
All integer binomial coefficients are interpreted in $\Omega$.
We use $0^0=1$, and the displayed range avoids negative exponents
when $\beta_i=0$. For $\chi=1$, the product and sum are empty.
In terms of the quotient functional $g$ in~\eqref{eq:moment-functional}
and the CRT idempotents $e_i$, the coefficients are
\[
 c_{i,j}=g\bigl(e_i(z-\beta_i)^j\bigr).
\]
\end{lemma}
\begin{proof}
If $\ell=0$, the sequence is zero and the assertion is immediate.
For $\ell\ge1$, Proposition~\ref{prop:crt-taylor-basis}, with variable
$Z$, gives the Taylor expansion in $\mathcal A$:
\[
 z^k=\sum_{i=1}^s\sum_{j=0}^{\min(k,m_i-1)}
 \binom{k}{j}\beta_i^{k-j}e_i(z-\beta_i)^j.
\]
Applying $g$ and using $g(z^k)=\eta_k$ proves
\eqref{eq:generalized-moments}. This argument works in every
characteristic, since it uses polynomial Taylor coefficients
rather than division by factorials.
\end{proof}

When every $m_i=1$, put $c_i=g(e_i)$. The CRT decomposition gives
\[
 g([p])=\sum_{i=1}^s c_i p(\beta_i)\qquad(p\in\Omega[Z]),
 \qquad
 \eta_k=\sum_{i=1}^s c_i\beta_i^k.
\]
Thus the functional is a linear combination of evaluations at
the distinct roots, and the sequence is a linear combination
of the \emph{power sequences} $(\beta_i^k)_{k\ge0}$.
When $g$ annihilates the nilradical, these same formulas hold
even if the recurrence polynomial has repeated roots.

For a representation $\eta_k=\sum_{i=1}^s c_i\beta_i^k$ with
distinct $\beta_i$ and all $c_i\ne0$, the first $s\times s$
Hankel matrix factors as
\[
 H_{s,s}=V\operatorname{diag}(c_1,\ldots,c_s)V^\top,
 \qquad V_{a,i}=\beta_i^a\quad(0\le a<s).
\]
The matrix $V$ is invertible: a polynomial of degree less than $s$
vanishing at all $\beta_i$ is zero, so evaluation has trivial kernel.
Thus $H_{s,s}$ has rank $s$. The polynomial
$\prod_i(Z-\beta_i)$ annihilates each power sequence, so it gives a
recurrence of order $s$. The rank bound then proves that this
order is minimal.

The next lemma shows that finitely many Frobenius identities
eliminate the nilpotent contributions and force the remaining
coefficients to lie in the prime field. It applies in arbitrary
prime characteristic.

\begin{lemma}[Recurrence sequences as sums of powers]
\label{lem:atoms}
Let $F$ be a field of characteristic $p>0$, and let
$(\mu_k)_{k\ge0}$ be a sequence in $F$. Let $\chi\in F[Z]$ be
any monic recurrence polynomial for this sequence, of degree $\ell\ge0$.
The polynomial $\chi$ need not be minimal. Suppose
\begin{equation}\label{eq:foundation-frobenius-window}
 \mu_{pb}=\mu_b^p\qquad(0\le b<\ell).
\end{equation}
In any extension field $\Omega/F$ containing all roots of $\chi$ there
is a set $\mathcal E$ of distinct roots, with $|\mathcal E|\le \ell$, and nonzero
coefficients $c_\beta\in\F_p$, the prime subfield of $\Omega$, such that
\begin{equation}\label{eq:foundation-atoms}
 \mu_k=\sum_{\beta\in\mathcal E}c_\beta\beta^k\qquad(k\ge0).
\end{equation}
The set may be empty and may contain zero. Conversely, every
sequence of this form satisfies $\mu_{pk}=\mu_k^p$ for all $k$
and a recurrence with polynomial $\prod_{\beta\in\mathcal E}(Z-\beta)$
when $\mathcal E$ is nonempty. In characteristic two, every nonzero
coefficient $c_\beta$ is one.
\end{lemma}
\begin{proof}
If $\ell=0$, the sequence is zero and the empty set gives the result.
Assume $\ell\ge1$ and construct $g$ by~\eqref{eq:moment-functional}
for the given sequence and the chosen polynomial $\chi$.
Every element of
$\mathcal A$ has a unique representative
$\psi=\sum_{b=0}^{\ell-1}\eta_b z^b$. In characteristic $p$,
\begin{equation}\label{eq:functional-frobenius}
 g(\psi^p)=\sum_{b<\ell}\eta_b^p\mu_{pb}
           =\sum_{b<\ell}\eta_b^p\mu_b^p=g(\psi)^p.
\end{equation}
This uses linearity of $g$, additivity of the $p$th-power map, and precisely
the $\ell$ hypotheses in~\eqref{eq:foundation-frobenius-window}.
It proves an identity for every element of the quotient algebra,
including powers that have already been reduced modulo $\chi$.

If $\psi$ is nilpotent, choose $h$ such that $\psi^{p^h}=0$.
Iterating~\eqref{eq:functional-frobenius} yields
$g(\psi)^{p^h}=0$. Since the codomain $\Omega$ is a field, this
forces $g(\psi)=0$. Thus $g$ annihilates the entire nilradical
$\mathcal N$ and descends to
\[
 \overline g:\mathcal A/\mathcal N
       \simeq\prod_{\beta:\,\chi(\beta)=0}\Omega\longrightarrow\Omega.
\]
Apply Lemma~\ref{lem:crt-idempotent-classification} to the split
recurrence polynomial $\chi$ over $\Omega$. For each distinct
root $\beta$, its coordinate unit vector in $\mathcal A/\mathcal N$
has a unique idempotent lift $e_\beta\in\mathcal A$, with value one
at $\beta$ and zero at all other roots. Because $e_\beta^2=e_\beta$,
we also have $e_\beta^p=e_\beta$, so
the Frobenius identity gives
\[
 c_\beta^p=c_\beta,\qquad c_\beta=g(e_\beta).
\]
The $p$ elements of $\F_p$ are distinct roots of $Y^p-Y$.
Since this polynomial has degree $p$, it has no other roots in
$\Omega$. Thus $c_\beta\in\F_p$. Modulo the nilradical,
$z^k=\sum_\beta \beta^k e_\beta$, including $k=0$ with $0^0=1$.
Applying $g$ gives $\mu_k=\sum_\beta c_\beta\beta^k$. Keep precisely
the roots with $c_\beta\ne0$. There are at most $\ell$ distinct roots.

For the converse, the $p$th-power map is additive and fixes each
$c_\beta\in\F_p$, so
\[
 \mu_k^p=\sum_{\beta\in\mathcal E}c_\beta^p\beta^{pk}
        =\sum_{\beta\in\mathcal E}c_\beta\beta^{pk}=\mu_{pk}.
\]
Each $\beta$ is a root of the displayed product polynomial, so
taking the linear combination of its shifted root identities
with coefficients $c_\beta$ gives the corresponding recurrence.
Finally, $\F_2^\times=\{1\}$ gives the characteristic-two assertion.
\end{proof}

The conclusion describes the functional and the sequence; it
allows repeated roots in the chosen annihilating polynomial.
The minimal recurrence of a nonzero resulting sequence is
square-free by the Hankel argument above. Also
$g(1)=\sum_{\beta\in\mathcal E}c_\beta$; in characteristic two this is the
parity of $|\mathcal E|$. A nonzero sequence may therefore have zeroth
moment zero. In particular one must not assume normalization for
each individual moment sequence in the later table construction.

\begin{example}[A repeated root visible in the moments]
In characteristic two take $\chi=(Z-1)^2=Z^2+1$.
In $\mathcal A$, write $z=1+\varepsilon$ with $\varepsilon^2=0$.
The functional $g(1)=g(\varepsilon)=1$ gives
\[
 \eta_k=g((1+\varepsilon)^k)=1+(k\bmod2),
\]
so its sequence is $1,0,1,0,\ldots$. Its leading $2\times2$
Hankel matrix is the identity, although its recurrence has only
one distinct root. The extra rank comes from the nilpotent
coordinate. It fails $\eta_2=\eta_1^2$, since these values are
one and zero. If instead $g(\varepsilon)=0$, the sequence is
constantly one, its minimal recurrence has order one, and the
original repeated-root polynomial still annihilates it.
Thus the same expansion lemma applies with either the chosen
polynomial $(Z-1)^2$ or the minimal polynomial $Z-1$ for the
constant sequence.
\end{example}

\begin{example}[A delayed Fibonacci sequence over different fields]
\label{ex:delayed-fibonacci}
Consider the infinite sequence
\[
 (\eta_k)_{k\ge0}=0,0,0,1,1,2,3,5,\ldots,
 \qquad \eta_{k+1}=\eta_k+\eta_{k-1}\quad(k>3).
\]
Over $\F_p$, the displayed integers are reduced modulo $p$.
Although its tail follows the Fibonacci recurrence, the whole
sequence satisfies
\[
 \eta_{k+4}=\eta_{k+3}+\eta_{k+2}\quad(k\ge0),
 \qquad \chi(Z)=Z^2(Z^2-Z-1).
\]
This polynomial is minimal over every field: a recurrence of
order $1\le \ell\le3$, applied at $k=3-\ell$, would express
$\eta_3=1$ using preceding zeros, and order zero is impossible.

We compute the coefficients in~\eqref{eq:generalized-moments}
by CRT, as prescribed in Lemma~\ref{lem:general-root-expansion}.
Let $q(Z)=Z^2-Z-1$, work over an extension $\Omega$ containing
all its roots, and put $\mathcal A=\Omega[Z]/(\chi)$, $z=[Z]$.
The initial values give
\[
 g(r_0+r_1z+r_2z^2+r_3z^3)=r_3\qquad(r_i\in\Omega).
\]
Thus $g$ reads the cubic coefficient of the representative of
degree less than four.

For the double root zero, a CRT idempotent representative is
\[
 E_0(Z)=(Z-1)q(Z)=1-2Z^2+Z^3,
\]
since $E_0\equiv1\pmod{Z^2}$ and $E_0\equiv0\pmod q$.
With $e_0=[E_0]$, reduction modulo $\chi$ gives
\[
 \operatorname{rem}_{\chi}(ZE_0)=Z+Z^2-Z^3.
\]
Consequently the coefficients at zero in
\eqref{eq:generalized-moments} are
\[
 c_{0,0}=g(e_0)=1,\qquad c_{0,1}=g(e_0z)=-1.
\]
Here coefficients are indexed by their root and Taylor order.
We continue to use $0^0=1$ as in~\eqref{eq:generalized-moments}.

If $q$ has distinct roots $\alpha,\beta$, its other CRT
idempotents have representatives
\[
 E_\alpha(Z)=\frac{Z^2(Z-\beta)}{\alpha^2(\alpha-\beta)},
 \qquad
 E_\beta(Z)=\frac{Z^2(Z-\alpha)}{\beta^2(\beta-\alpha)}.
\]
Each is one at its own simple root, zero at the other, and
divisible by $Z^2$. Reading their cubic coefficients gives
\[
 c_{\alpha,0}=g([E_\alpha])=\frac{1}{\alpha^2(\alpha-\beta)},
 \qquad
 c_{\beta,0}=g([E_\beta])=\frac{1}{\beta^2(\beta-\alpha)}.
\]
Since $\alpha\beta=-1$, substituting these CRT coefficients
into~\eqref{eq:generalized-moments} gives, for all $k\ge0$,
\[
 \eta_k=\sum_{j=0}^{\min(k,1)}(-1)^j\binom{k}{j}0^{k-j}
       +\frac{\beta^2\alpha^k-\alpha^2\beta^k}{\alpha-\beta}.
\]
The cases below specialize these coefficients and the root
multiplicities in~\eqref{eq:generalized-moments}.

\begin{enumerate}[label=(\alph*),leftmargin=*]
\item \textbf{Over $\R$.}
With $\alpha=(1+\sqrt5)/2$ and $\beta=(1-\sqrt5)/2$, the formula is
\[
 \eta_k=\sum_{j=0}^{\min(k,1)}(-1)^j\binom{k}{j}0^{k-j}
       +\frac{\beta^2\alpha^k-\alpha^2\beta^k}{\sqrt5}
 \qquad(k\ge0).
\]

\item \textbf{Over $\F_p$ when $Z^2-Z-1$ is irreducible.}
Take $\Omega=\F_p[Z]/(Z^2-Z-1)\simeq\F_{p^2}$ and let
$\alpha=[Z]$. Frobenius sends $\alpha$ to the other root
$\beta=\alpha^p=1-\alpha$: it permutes the roots, and
$\alpha^p\ne\alpha$ because $\alpha\notin\F_p$. Thus
\[
 \eta_k=\sum_{j=0}^{\min(k,1)}(-1)^j\binom{k}{j}0^{k-j}
       +\frac{\beta^2\alpha^k-\alpha^2\beta^k}{\alpha-\beta}
 \qquad(k\ge0).
\]
The fraction lies in $\F_p$: taking its $p$th power exchanges
the two roots and negates both numerator and denominator, so
the quotient is fixed by Frobenius. For example, $p=3$ satisfies
the irreducibility condition, since none of $0,1,2$ is a root.

\item \textbf{A reducible example over $\F_5$.}
Here $\chi(Z)=Z^2(Z-3)^2$. Applying
CRT to the factors $Z^2$ and $(Z-3)^2$, the idempotent for the
root $3$ is $e_3=1-e_0$, with representative
\[
 E_3(Z)=2Z^2-Z^3,\qquad
 \operatorname{rem}_{\chi}\bigl(E_3(Z)(Z-3)\bigr)=-2Z^2-Z^3.
\]
Reading cubic coefficients therefore gives
\[
 c_{3,0}=g(e_3)=-1,\qquad
 c_{3,1}=g\bigl(e_3(z-3)\bigr)=-1.
\]
Together with $c_{0,0}=1$ and $c_{0,1}=-1$ computed above,
formula~\eqref{eq:generalized-moments} now gives
\[
 \eta_k=\sum_{j=0}^{\min(k,1)}(-1)^j\binom{k}{j}0^{k-j}
       -\sum_{j=0}^{\min(k,1)}\binom{k}{j}3^{k-j}
 \quad\text{in }\F_5\quad(k\ge0).
\]
The integer binomial coefficients are interpreted modulo five.

\item \textbf{Two distinct roots over $\F_{11}$.}
Take $p=11$. Then
\[
 Z^2-Z-1=(Z-8)(Z-4)\quad\text{in }\F_{11}[Z],
\]
so the two roots $\alpha=8$ and $\beta=4$ are distinct and
already lie in the base field. The coefficients obtained from
CRT can be read from the degree-three representatives
\[
 E_8(Z)=4Z^2(Z-4),\qquad E_4(Z)=6Z^2(Z-8).
\]
Thus $c_{8,0}=g([E_8])=4$ and $c_{4,0}=g([E_4])=6$.
Substitution into~\eqref{eq:generalized-moments} yields
\[
 \eta_k=\sum_{j=0}^{\min(k,1)}(-1)^j\binom{k}{j}0^{k-j}
       +4\cdot8^k+6\cdot4^k
 \quad\text{in }\F_{11}\qquad(k\ge0).
\]

\end{enumerate}
\end{example}

\subsection{BCH codes and syndrome decoding}\label{sec:bch-background}

\paragraph{Optional background.}
This subsection may be omitted on a first reading of the
hardness proof. It gives decoding applications of the algebra
developed above; no BCH decoding algorithm is invoked in the
construction or soundness argument.

Bose--Chaudhuri--Hocquenghem (BCH) codes illustrate how
consecutive powers, Hankel matrices, and recurrences combine in
a decoding algorithm. We give the narrow-sense construction over
an arbitrary finite field and prove a decoding guarantee. Massey's
paper~\cite[Sections III--V]{massey} develops the connection
between BCH decoding and shortest linear recurrences.
For a textbook treatment of BCH codes and decoding, see
Huffman and Pless~\cite[Chapter 5]{huffman-pless}.

\subsubsection{Definition and consecutive-root parity checks}

In this subsection write $F=\F_q$, let
$\Omega=\F_{q^m}$ be a represented extension, and choose
$\alpha\in\Omega^\times$ of multiplicative order $B\ge2$.
Thus $B\mid q^m-1$, so the characteristic does not divide $B$.
The symbol $B$ denotes the code length here. Identify a vector
$c=(c_0,\ldots,c_{B-1})^\top\in F^B$ with
$c(X)=\sum_{i=0}^{B-1}c_iX^i$.

Fix an integer $2\le\delta\le B$, called the
\emph{designed distance}. For $\beta\in\Omega$, let
$m_\beta(X)\in F[X]$ be its monic minimal polynomial. Define
\begin{equation}\label{eq:bch-generator}
 g(X)=\mathop{\rm lcm}_{1\le k\le\delta-1}m_{\alpha^k}(X),
 \qquad
 C_{\rm BCH}=\{c\in F^B:g(X)\mid c(X)\}.
\end{equation}
Equivalently,
\begin{equation}\label{eq:bch-checks}
 c(\alpha^k)=\sum_{i=0}^{B-1}c_i\alpha^{ik}=0
       \qquad(1\le k\le\delta-1).
\end{equation}
Indeed a polynomial over $F$ vanishes at $\alpha^k$ exactly
when $m_{\alpha^k}$ divides it. Each equation over $\Omega$
expands into at most $m$ linear equations over $F$.

The polynomial $X^B-1$ has the $B$ distinct roots $\alpha^i$,
$0\le i<B$, so $g\mid X^B-1$. Multiplication by $X$ modulo
$X^B-1$ cyclically shifts a codeword and preserves divisibility
by $g$. Thus $C_{\rm BCH}$ is a cyclic linear code, of dimension
$B-\deg g$: in $F[X]/(X^B-1)$ its elements are
\[
 g(X)u(X)\pmod{X^B-1},\qquad \deg u<B-\deg g.
\]
In this degree range the product already has degree less than
$B$, and these representatives are unique. The parity checks automatically include
Frobenius conjugates, since
$c(\beta^q)=c(\beta)^q$ for $c\in F[X]$.
The least common multiple in \eqref{eq:bch-generator} accounts
for this redundancy. In particular, $\deg g\le m(\delta-1)$.

\begin{proposition}[The BCH distance bound]\label{prop:bch-bound}
Every nonzero codeword in $C_{\rm BCH}$ has Hamming weight at
least $\delta$.
\end{proposition}
\begin{proof}
Suppose a nonzero codeword has support
$\{i_1,\ldots,i_w\}$ with $1\le w<\delta$. Put
$\beta_j=\alpha^{i_j}$ and $a_j=c_{i_j}\ne0$. These location
values are distinct and nonzero. The first $w$ parity checks give
\[
 \sum_{j=1}^w a_j\beta_j^k=0\qquad(1\le k\le w).
\]
Their coefficient matrix is a Vandermonde matrix with its $j$th
column multiplied by $\beta_j$. Its determinant is
\[
 \left(\prod_{j=1}^w\beta_j\right)
       \prod_{i<j}(\beta_j-\beta_i)\ne0.
\]
Hence every $a_j=0$, a contradiction.
\end{proof}

The bound implies that a received vector has at most one
codeword within distance $\lfloor(\delta-1)/2\rfloor$, by the
triangle inequality. The designed distance is a proved lower
bound; it need not equal the actual minimum distance.

\subsubsection{Syndromes as moments and their recurrence}
\label{sec:bch-syndromes}

Let $t=c+e\in F^B$ be received, where $c\in C_{\rm BCH}$.
For $1\le k\le\delta-1$, compute the \emph{syndromes}
\begin{equation}\label{eq:bch-syndromes}
 s_k=t(\alpha^k)=e(\alpha^k)
       =\sum_{j=1}^w a_j\beta_j^k,
 \qquad \beta_j=\alpha^{i_j},\quad a_j=e_{i_j}\in F^\times,
\end{equation}
where $\{i_1,\ldots,i_w\}$ is the error support. These values
are known from $t$ although the support and magnitudes are
unknown. The same power-sum formula defines $s_k=e(\alpha^k)$
for every $k\ge0$, but only the displayed range is immediately
available as syndromes of the received vector.

Define the monic polynomial with roots at the error locations,
and its reciprocal, by
\begin{equation}\label{eq:bch-locator}
 \begin{aligned}
 \chi(Z)&=\prod_{j=1}^w(Z-\beta_j)
          =Z^w+\lambda_1Z^{w-1}+\cdots+\lambda_w,\\
 \chi^*(Z)&=Z^w\chi(1/Z)
          =\prod_{j=1}^w(1-\beta_j Z)
          =1+\lambda_1Z+\cdots+\lambda_wZ^w.
 \end{aligned}
\end{equation}
The polynomial $\chi^*$ is the \emph{error-locator polynomial}
in this convention: its roots are the inverses $\beta_j^{-1}$.
Substitution gives the recurrence
\begin{equation}\label{eq:bch-recurrence}
 s_{k+w}+\lambda_1s_{k+w-1}+\cdots+\lambda_ws_k
    =\sum_j a_j\beta_j^k\chi(\beta_j)=0
       \qquad(k\ge0).
\end{equation}
Equivalently, $\chi(S)s=0$ for the shift operator of
Appendix~\ref{sec:recurrence-foundation}.

We can rewrite this recurrence as a polynomial identity. Define
\begin{equation}\label{eq:bch-syndrome-polynomial}
 S_{\rm syn}(X)=\sum_{k=1}^{\delta-1}s_kX^{k-1},\qquad
 M(X)=X^{\delta-1}.
\end{equation}
Thus the known coefficients start with $s_1$, not $s_0$.
The infinite series associated with the error satisfies
\[
 \sum_{k\ge1}s_kX^{k-1}
    =\sum_{j=1}^w\frac{a_j\beta_j}{1-\beta_jX}.
\]
Multiplication by $\chi^*(X)$ gives the \emph{key equation}
\begin{equation}\label{eq:bch-key-equation}
 \chi^*(X)S_{\rm syn}(X)
   =\Psi(X)+A(X)M(X),\qquad \deg\Psi<w,
\end{equation}
for some $A(X)\in\Omega[X]$. Equivalently,
$\chi^* S_{\rm syn}\equiv\Psi\pmod{X^{\delta-1}}$, where
\[
 \Psi(X)=\sum_{j=1}^w a_j\beta_j
                 \prod_{\ell\ne j}(1-\beta_\ell X).
\]
For $w\ge1$, $\Psi$ and $\chi^*$ are relatively prime:
at $X=\beta_j^{-1}$ only the $j$th summand survives, and it is
nonzero. The Euclidean solution of such syndrome key equations
is due to Sugiyama, Kasahara, Hirasawa, and
Namekawa~\cite{sugiyama}.

There is also a Frobenius identity:
\[
 s_k^q=\sum_j a_j^q\beta_j^{qk}=s_{qk}.
\]
Indices can be reduced modulo $B$ because $\beta_j^B=1$.
For a prime-field code $F=\F_p$, this becomes $s_k^p=s_{pk}$,
exactly the form of Condition 3 in Section~\ref{sec:table}.
For a larger alphabet $F=\F_{p^a}$, the $p$th-power version
requires the additional restriction $a_j\in\F_p$; the
$q$th-power version holds for all magnitudes in $F$.

\begin{lemma}[Finite syndrome recovery]\label{lem:bch-finite-recovery}
Suppose $s_k=\sum_{j=1}^w a_j\beta_j^k$ with distinct nonzero
$\beta_j$ and nonzero $a_j$. If $w\ge1$, the first $2w$
syndromes $s_1,\ldots,s_{2w}$ have least recurrence order $w$,
and uniquely determine the monic polynomial $\chi$ of
\eqref{eq:bch-locator}.
\end{lemma}
\begin{proof}
Let $V=(\beta_j^i)_{0\le i<w,\,1\le j\le w}$. The Hankel
matrix of these data factors as
\[
 H_w=(s_{i+j+1})_{0\le i,j<w}
     =V\operatorname{diag}(a_1\beta_1,\ldots,a_w\beta_w)V^\top.
\]
Both factors $V$ and the diagonal matrix are invertible, so
$\rank H_w=w$. Lemma~\ref{lem:finite-hankel}\textup{(a)}
therefore excludes every recurrence of order less than $w$.
Equation~\eqref{eq:bch-recurrence} supplies order $w$.
Write $\chi(Z)=Z^w+\sum_{a=0}^{w-1}\gamma_aZ^a$.
Its coefficients satisfy
\[
 H_w(\gamma_0,\ldots,\gamma_{w-1})^\top
       =-(s_{w+1},\ldots,s_{2w})^\top.
\]
Invertibility gives uniqueness. Substitution of the true root
polynomial verifies that this unique solution is the stated
$\chi$.
\end{proof}

This distinguishes two uses of consecutive powers. If locations
are fixed, $w$ syndromes determine their $w$ magnitudes by
Vandermonde inversion. With the locations unknown, $2w$
syndromes suffice to determine the recurrence and hence the
locations as well.

\subsubsection{BCH decoding by the extended Euclidean algorithm}
\label{sec:bch-decoding}
\label{sec:bch-euclidean}

Put $w_{\max}=\lfloor(\delta-1)/2\rfloor$. The following
procedure uses only the received vector, the represented fields,
and the fixed element $\alpha$.
\begin{enumerate}[leftmargin=*]
\item Compute $s_1,\ldots,s_{\delta-1}$ by
\eqref{eq:bch-syndromes}. If they are all zero, return $t$,
which already belongs to the code by \eqref{eq:bch-checks}.
\item Apply the extended Euclidean algorithm of Sugiyama,
Kasahara, Hirasawa, and Namekawa~\cite{sugiyama} to
$M=X^{\delta-1}$ and $S_{\rm syn}$, keeping the coefficients
of each remainder. Stop at the first remainder $r$ with
$\deg r<(\delta-1)/2$, and write
\[
 r=uM+vS_{\rm syn}.
\]
If $v(0)=0$, report failure. Otherwise set
\[
 \chi^*(X)=\frac{v(X)}{v(0)},\qquad \ell=\deg\chi^*.
\]
Require $1\le\ell\le w_{\max}$; otherwise report failure.
Proposition~\ref{prop:bch-euclidean} below proves that, within
the designed decoding radius, this gives the error locator.
\item Evaluate $\chi^*(\alpha^{-i})$ for $0\le i<B$.
Keep the zero positions $i_1,\ldots,i_\ell$; this enumeration
of possible locations is often called a \emph{Chien search}.
Require $\deg\chi^*=\ell$ and exactly $\ell$ such positions;
otherwise report failure.
\item Solve for the magnitudes using
\[
 \sum_{j=1}^\ell a_j\alpha^{i_jk}=s_k
       \qquad(1\le k\le\ell).
\]
The coefficient matrix is invertible by the same Vandermonde
argument as in the BCH bound. Require $a_j\in F^\times$,
and form $\widehat e\in F^B$ with these entries at the recovered
positions and zero elsewhere.
\item Verify $\wt(\widehat e)\le w_{\max}$ and
$t-\widehat e\in C_{\rm BCH}$ by the parity checks. If they
hold, return $t-\widehat e$; otherwise report failure.
\end{enumerate}

We record the elementary Euclidean facts needed for the stopping
rules. All computations below take place in $\Omega[X]$.
\begin{lemma}[Coefficients in a partial Euclidean algorithm]
\label{lem:partial-euclidean-degrees}
Let $M,T\ne0$ with $\deg T<\deg M$. Starting with
$r_0=M$, $r_1=T$, write the successive Euclidean remainders as
$r_j=u_jM+v_jT$, with coefficient rows $(u_0,v_0)=(1,0)$ and
$(u_1,v_1)=(0,1)$. Consecutive coefficient rows have determinant
$\pm1$. For $j\ge1$, including the final zero remainder,
\[
 \deg v_j=\deg M-\deg r_{j-1}.
\]
Consequently, at the first remainder with $\deg r_j<H<\deg M$,
\[
 \deg v_j\le\deg M-\lceil H\rceil.
\]
\end{lemma}
\begin{proof}
Each division gives
$r_{j+1}=r_{j-1}-q_jr_j$ and the same recurrence for $u_j,v_j$.
It changes the determinant of consecutive rows by a sign.
The degree identity starts with $v_1=1$.
Since $\deg q_j=\deg r_{j-1}-\deg r_j$, induction gives
$\deg(q_jv_j)=\deg M-\deg r_j>\deg v_{j-1}$, so the leading
term cannot cancel in $v_{j+1}$. This proves the identity.
At the first stopping remainder the preceding degree is at least
$H$, proving the last assertion.
\end{proof}

\begin{proposition}[Euclidean recovery of the locator]
\label{prop:bch-euclidean}
Suppose $1\le w\le w_{\max}$. Apply the extended Euclidean
algorithm to $M=X^{\delta-1}$ and $S_{\rm syn}$, stopping at the
first remainder $r$ satisfying
\begin{equation}\label{eq:syndrome-euclidean-stop}
 \deg r<\frac{\delta-1}{2}.
\end{equation}
If $r=uM+vS_{\rm syn}$, then $v(0)\ne0$ and
\[
 \chi^*(X)=\frac{v(X)}{v(0)},\qquad
 \Psi(X)=\frac{r(X)}{v(0)}.
\]
Thus Step 2 of the decoder recovers the true error locator.
\end{proposition}
\begin{proof}
By the lemma, $\deg v\le w_{\max}$. Both
$\Psi v-\chi^* r$ and $0$ are congruent modulo $M$ by the
key equation. Its degree is less than $\delta-1$, because
\[
 \deg(\chi^* r)<w+\frac{\delta-1}{2}\le\delta-1,
 \qquad
 \deg(\Psi v)\le w-1+w_{\max}\le\delta-2.
\]
Therefore $\Psi v-\chi^* r=0$ as a polynomial.
To see that the proportionality factor is a constant, use the
previous Euclidean remainder $r'$ with coefficient $v'$.
The determinant identity in the lemma gives
$r'v-v'r=\pm M$. Since $\Psi\equiv\chi^* S_{\rm syn}\pmod M$,
the pair $(\Psi,\chi^*)$ is a polynomial linear combination
of $(r',v')$ and $(r,v)$: the corresponding coefficient rows
form an invertible polynomial matrix. Its coefficient on the
first pair is $(\Psi v-\chi^* r)/(\pm M)=0$.
Hence $(\Psi,\chi^*)=b(X)(r,v)$ for a polynomial $b$.
Relative primality of $\Psi,\chi^*$ makes $b$ a nonzero
constant. Since $\chi^*(0)=1$, normalization gives the claim.
Finally the stopping remainder exists and is nonzero:
$\gcd(M,S_{\rm syn})$ divides $\Psi$, whose degree is less than
$w\le(\delta-1)/2$.
\end{proof}

The stopping inequality is strict. If $\delta-1=2w_{\max}$,
it means $\deg r\le w_{\max}-1$. Allowing equality with
$(\delta-1)/2$ need not give the locator. With no error promise,
retain the root, magnitude, and codeword checks in the decoder,
and report failure if the necessary normalization is impossible.
The all-zero syndrome case is handled before running this step.

\begin{theorem}[BCH decoding within the designed radius]
\label{thm:bch-decoding}
If $t=c+e$, where $c\in C_{\rm BCH}$ and
$\wt(e)\le w_{\max}$, the procedure returns $c$.
For any received vector, it returns the unique codeword within
that radius if one exists, and reports failure if none exists.
Its running time is polynomial in $B$, the extension degree,
and the bit length of the represented field data.
\end{theorem}
\begin{proof}
Write $w=\wt(e)$. If $w=0$, all syndromes vanish and the
first step returns $t=c$. If $w\ge1$, then $2w\le\delta-1$,
and the invertible Hankel matrix in
Lemma~\ref{lem:bch-finite-recovery} ensures that the syndromes
are not all zero.
Proposition~\ref{prop:bch-euclidean} shows that Step 2 finds
exactly $\chi^*=\prod_j(1-\beta_j X)$, with degree $w$.
In particular, its normalization is defined and its degree passes
the required check.
The next steps recover precisely its distinct inverse locations
and the unique magnitudes. They therefore give $\widehat e=e$
and return $c$.

The BCH bound and the triangle inequality give uniqueness within
radius $w_{\max}$. The final check ensures that a returned vector
is such a codeword even without an error promise. If one exists,
the preceding argument guarantees success; if none exists, no
output can pass that check. Evaluation over at most $B$ possible
locations, the extended Euclidean algorithm, the magnitude
linear system, and the final checks require polynomially many field
operations on the represented fields. These operations have
polynomial bit cost in their representations.
\end{proof}

The theorem is a decoding guarantee for this structured family
within a prescribed radius. It does not solve unrestricted CCP
for arbitrary received vectors and codes.

\subsubsection{Gao's interpolation formulation with modulus \texorpdfstring{$X^B-1$}{X to the B minus 1}}
\label{sec:gao-interpolation-decoding}

Gao's method~\cite{gao-rs} uses interpolation and a partial
Euclidean algorithm to recover the message polynomial directly.
For the present BCH code it applies through a Reed--Solomon
code over $\Omega$. Put $k_{\rm RS}=B-\delta+1$ and define
\[
 \operatorname{RS}=\{(f(\alpha^i))_{0\le i<B}:
                         f\in\Omega[X],\ \deg f<k_{\rm RS}\}.
\]
Then
\begin{equation}\label{eq:bch-rs-subfield}
 C_{\rm BCH}=\operatorname{RS}\cap F^B.
\end{equation}
The message polynomial below means the unique ambient
Reed--Solomon interpolant of the decoded vector. Its coefficients
lie in $\Omega$ and must satisfy the subfield evaluation condition;
they are not arbitrary BCH message coordinates in $F^{k_{\rm RS}}$.
To verify this equality and identify the received polynomial,
interpolate $R(\alpha^i)=t_i$ with $\deg R<B$. Explicitly,
\begin{equation}\label{eq:gao-received-interpolant}
 R(X)=\frac1B\sum_{j=0}^{B-1}t(\alpha^{-j})X^j.
\end{equation}
Here $B\ne0$ in $\Omega$. The geometric-sum identity
\[
 \sum_{j=0}^{B-1}\alpha^{j(i-h)}
   =\begin{cases}B,&i=h\pmod B,\\0,&i\ne h\pmod B\end{cases}
\]
shows that evaluating \eqref{eq:gao-received-interpolant} at
$\alpha^i$ gives $t_i$. For $t=c\in F^B$, the BCH checks
$c(\alpha^k)=0$, $1\le k\le\delta-1$, are exactly the vanishing
of the coefficients with $B-\delta+1\le j<B$ in this interpolant.
This proves \eqref{eq:bch-rs-subfield}.

If $R(X)=0$, return the zero message polynomial. Henceforth
assume $R(X)\ne0$. Run the extended Euclidean algorithm on
\begin{equation}\label{eq:gao-euclidean-inputs}
 R(X)\quad\text{and}\quad
 G(X)=\prod_{i=0}^{B-1}(X-\alpha^i)=X^B-1,
\end{equation}
starting the divisions with the higher-degree polynomial $G$.
Write its current remainder as $v=uG+rR$, so $r(X)$ is the
multiplier of $R(X)$. Stop at the first remainder with
\begin{equation}\label{eq:gao-euclidean-stop}
 \deg v<\frac{B+k_{\rm RS}}2
         =B-\frac{\delta-1}{2}.
\end{equation}
Lemma~\ref{lem:partial-euclidean-degrees} guarantees $r\ne0$. Divide
$v(X)$ by $r(X)$, say
\[
 v(X)=f_1(X)r(X)+r_0(X),\qquad \deg r_0<\deg r.
\]
If $r_0(X)\ne0$ or $\deg f_1\ge k_{\rm RS}$, return failure.
Otherwise verify that the evaluation vector of $f_1$ lies in
$F^B$ and differs from $t$ in at most $w_{\max}$ positions.
If a check fails, return failure; otherwise return $f_1(X)$
as the message polynomial. Its evaluation vector is the decoded
BCH codeword. The degree test includes equality because the
Reed--Solomon message space requires $\deg f_1<k_{\rm RS}$.

\begin{proposition}[Interpolation decoding for the BCH code]
\label{prop:gao-bch-decoding}
With the stopping convention \eqref{eq:gao-euclidean-stop},
the interpolation procedure returns the message polynomial of
the unique BCH codeword within distance $w_{\max}$ whenever
that codeword exists.
\end{proposition}
\begin{proof}
Suppose $t$ differs in $w\le w_{\max}$ positions from
$c=(f_c(\alpha^i))_i\in C_{\rm BCH}$, with
$\deg f_c<k_{\rm RS}$. The case $R=0$ is immediate.
Lemma~\ref{lem:partial-euclidean-degrees} gives $r\ne0$ and
$\deg r\le\lfloor(B-k_{\rm RS})/2\rfloor=w_{\max}$.
At each of the $B-w$ correct positions, the congruence
$v\equiv rR\pmod G$ implies $v(\alpha^i)=r(\alpha^i)f_c(\alpha^i)$.
But
\[
 \deg(v-rf_c)<\frac{B+k_{\rm RS}}2,
 \qquad
 B-w\ge B-w_{\max}\ge\frac{B+k_{\rm RS}}2.
\]
The root bound forces $v=rf_c$. Division therefore has
$r_0=0$ and returns $f_1=f_c$, which passes all the checks.
Uniqueness follows from the BCH distance bound. Without the
promise, the checks guarantee that the returned message
encodes a BCH codeword within the stated radius.
\end{proof}

The two Euclidean formulations use different data and stopping
bounds:
\[
\begin{array}{c|c|c|c}
 \text{formulation}&\text{modulus}&\text{second polynomial}
                  &\text{stop when}\\\hline
 \text{syndromes}&X^{\delta-1}&S_{\rm syn}(X)
                  &\deg r<(\delta-1)/2\\
 \text{interpolation}&X^B-1&R(X)
                  &\deg v<B-(\delta-1)/2
\end{array}
\]
The syndrome polynomial cannot simply replace $R(X)$ while
retaining the modulus $X^B-1$. Formula
\eqref{eq:gao-received-interpolant} explains their relation:
its top coefficients satisfy $[X^{B-k}]R=s_k/B$ for
$1\le k\le\delta-1$. Thus the two approaches encode the same
parity-check information in different polynomial forms.

\subsubsection{Example: two errors in a binary BCH code}

Take $F=\F_2$ and
$\Omega=\F_2[U]/(U^4+U+1)$, with $\theta=[U]$.
The defining polynomial has no linear factor and is not the
square of the only irreducible quadratic $X^2+X+1$ over
$\F_2$, so it is irreducible. The element $\theta$ has order
$15$: its order divides $15$, while orders $1$ and $3$ contradict
its degree four, and order $5$ is excluded by
$\theta^5=\theta^2+\theta\ne1$.

Use $B=15$, $\alpha=\theta$, and $\delta=5$.
The Frobenius orbits of the exponents $1,2,3,4$ are
$\{1,2,4,8\}$ and $\{3,6,12,9\}$ modulo $15$. Their minimal
polynomials give
\[
 \begin{aligned}
 g(X)&=(X^4+X+1)(X^4+X^3+X^2+X+1)\\
     &=X^8+X^7+X^6+X^4+1.
 \end{aligned}
\]
For the second factor, $\theta^3$ has order five, so it is a
root of $1+X+X^2+X^3+X^4$; its orbit has size four.
The code has dimension $15-8=7$ and minimum distance at least
five, so the theorem guarantees correction of two errors.

Suppose a codeword is changed in positions $0$ and $1$, so
$e(X)=1+X$. The four syndromes are
\[
 s_1=1+\theta,\qquad s_2=1+\theta^2,\qquad
 s_3=1+\theta^3,\qquad s_4=1+\theta^4=\theta.
\]
The order-two recurrence coefficients solve
\[
 \begin{pmatrix}s_2&s_1\\s_3&s_2\end{pmatrix}
 \begin{pmatrix}\lambda_1\\\lambda_2\end{pmatrix}
 =\begin{pmatrix}s_3\\s_4\end{pmatrix}.
\]
Its determinant is
$s_2^2+s_1s_3=\theta+\theta^3\ne0$.
Substitution gives $\lambda_1=1+\theta$, $\lambda_2=\theta$,
and hence
\[
 \chi^*(Z)=1+(1+\theta)Z+\theta Z^2
            =(1+Z)(1+\theta Z).
\]
The inverse-location search finds $1$ and $\theta^{-1}$, so the
error positions are $0$ and $1$. Solving the first two syndrome
equations yields both magnitudes equal to one. Subtracting this
error vector recovers the transmitted codeword.

\subsubsection{The connection to the moment-table reduction}

There are three direct correspondences. The Vandermonde proof
of the BCH bound is the coefficient-isolation rule in
Section~\ref{sec:synchronization}. The Hankel factorization in
Lemma~\ref{lem:bch-finite-recovery} is the power-sum counterpart
of the low-rank recurrence argument in Section~\ref{sec:atoms}.
The prime-field syndrome identity $s_{pk}=s_k^p$ is the
Frobenius compatibility imposed by the table construction.

In BCH decoding the possible locations belong to the known set
$\{\alpha^i:0\le i<B\}$. The reduction first obtains recurrence
roots in a splitting field, and then uses its encoding equations
to prove that they represent assignments in $E$. Its moment
window must also synchronize the moment sequences associated
with different pairs $(i,j)$. Lemma~\ref{lem:block-column-atoms}
obtains one recurrence polynomial for all $n+1$ sequences in a
fixed block column. Their union of root supports consequently
has size at most $\kappa-1$ before synchronization. Symmetry
uses at most two such supports, while Frobenius uses indices
up to $p(\ell_j-1)$ inside the proved recurrence prefix.
These bounds give $L=c_p(\kappa-1)+1$, with
$c_p=\max\{2,p-1\}$, as specified in
Section~\ref{sec:cutoff-digits}. The joint recurrence is the
additional matrix information beyond the scalar BCH analogy.

\section{Notation and the main proof dependencies}\label{app:notation}

\begin{center}
\small
\begin{tabular}{p{0.18\textwidth}p{0.70\textwidth}}
\toprule
Symbol & Meaning\\\midrule
$s,m_2,m_3$ & Original variables, two-literal clauses, and three-literal clauses.\\
$m=m_2+m_3$ & Total number of clauses.\\
$n=s+m_3,v$ & Variables after clause auxiliaries, and number of quadratic forms.\\
$k,\nu,\ell$ & Power or moment index, quadratic-equation index, and recurrence order.\\
$t$ & Target vector for CCP or CVP.\\
$p,q=p^a$ & Fixed prime alphabet, and coefficient-field size.\\
$E=K(\zeta),c$ & Assignment field and degree over $K$.\\
$\alpha_i,\xi$ & Prime-field independent encoding directions, and $\Enc(x)$.\\
$d=p^b,r$ & Radix and number of label axes.\\
$c_p=\max\{2,p-1\}$ & Window constant, distinct from the extension degree $c$.\\
$L,\kappa$ & Row range $c_p(\kappa-1)+1$, and rank cutoff.\\
$J_0=L+\kappa-1$ & Moment count $(c_p+1)(\kappa-1)+1$; indices $0\le k<J_0$.\\
$B_j,\chi_j,\ell_j$ & Block column, its common recurrence polynomial, and its degree.\\
$\mathcal E_j$ & Joint root support within column $j$, of size at most $\ell_j\le\rho$.\\
$d^r\ge J_0$ & Ordinary-digit capacity for the required window.\\
$D=2r(d-1)(c-1)$ & Uniform interpolation-degree budget.\\
$S$ & Sample subset of $K$, with cardinality $|S|$.\\
$N$ & Number of prime-field table symbols and lattice rank.\\
$R_1,R_0$ & Honest weight upper bound and soundness threshold.\\
$\beta,\mathcal E$ & A recovered recurrence root and the set of active roots.\\
$A_\beta,a_\beta$ & Prime-field coefficient matrix and its nonzero $(0,0)$ entry.\\
\bottomrule
\end{tabular}
\end{center}
The weight bound is proved in Lemmas~\ref{lem:moment-evaluations}
and~\ref{lem:rank-weight}. Its determinant-multiplicity argument
is developed over an arbitrary field in
Lemma~\ref{lem:multiplicity} and
Proposition~\ref{prop:sample-rank-foundation}. The low-rank proof
uses the shared recurrence in Lemma~\ref{lem:block-column-atoms},
the finite-window reasoning of Lemma~\ref{lem:finite-hankel},
Lemma~\ref{lem:atoms}, and Proposition~\ref{prop:lagrange}.
Polynomial CRT, repeated-root
bases, and nilpotents explain the proof of Lemma~\ref{lem:atoms},
rather than being additional assumptions on its recurrence.
The mathematical prerequisites are proved in
Appendix~\ref{app:mathematics-background}.

\end{document}